\documentclass{article}
\usepackage{PRIMEarxiv}
\usepackage[utf8]{inputenc}
\usepackage[T1]{fontenc}
\usepackage{hyperref}
\usepackage{url}
\usepackage{microtype}
\usepackage{fancyhdr}
\usepackage{subcaption}
\usepackage{mathtools}
\usepackage{bbold}
\usepackage{amssymb}
\usepackage{amsmath}
\usepackage{amsfonts}
\usepackage{amsthm}
\newtheorem{proposition}{Proposition}
\usepackage{tabularx}
\usepackage{longtable}
\usepackage{booktabs}
\usepackage{array}
\usepackage{ragged2e}
\usepackage{tikz}
\usepackage{varwidth}
\usepackage{centernot}
\usepackage{microtype}
\usetikzlibrary{arrows.meta, fit}
\usetikzlibrary{backgrounds}
\usetikzlibrary{calc}
\newcolumntype{L}[1]{>{\raggedright\arraybackslash}p{#1}}
\newcolumntype{C}[1]{>{\centering\arraybackslash}p{#1}}
\usepackage{pifont}
\newcommand{\cmark}{\ding{51}}
\newcommand{\xmark}{}

\newcommand{\arrowbox}[1]{%
  \makebox[\widthof{$\Longleftrightarrow$}][l]{$#1$}%
}

\allowdisplaybreaks[4]

\title{Graph Pattern-based Association Rules \\Evaluated Under No-repeated-anything Semantics \\in the Graph Transactional Setting}

\author{
Basil Ell\\
Center for Cognitive Interaction Technology, Bielefeld University, Germany\\
Norwegian Centre for Knowledge-driven Machine Learning, University of Oslo, Norway\\
\texttt{bell@techfak.uni-bielefeld.de} \\
}

\newcommand{\pgone}{
\begin{tikzpicture}[>=latex, scale=0.2]
  \coordinate (g1t1) at (1,0);
  \coordinate (g1t2) at (0,1);
  \coordinate (g1t3) at (2,1);
  \coordinate (g1t4) at (0,2);
  \coordinate (g1t5) at (1,2);
  \coordinate (g1t6) at (2,2);
  \draw[thick] (g1t1) -- (g1t2);
  \draw[thick] (g1t1) -- (g1t3);
  \draw[thick] (g1t2) -- (g1t4);
  \draw[thick] (g1t3) -- (g1t6);  
  \draw[thick] (g1t1) -- (g1t5);
\end{tikzpicture}
}

\newcommand{\pgtwo}{
\begin{tikzpicture}[scale=0.2]
  \coordinate (g2t1) at (4,0);
  \coordinate (g2t2) at (3,1);
  \coordinate (g2t3) at (5,1);
  \coordinate (g2t4) at (4,2);
  \coordinate (g2t5) at (5,2);
  \draw[thick] (g2t1) -- (g2t2);
  \draw[thick] (g2t1) -- (g2t3);
  \draw[thick] (g2t2) -- (g2t4);
  \draw[thick] (g2t1) -- (g2t4);
  \draw[thick] (g2t3) -- (g2t5);
\end{tikzpicture}
}

\newcommand{\pgthree}{
\begin{tikzpicture}[scale=0.2]
  \coordinate (g3t1) at (7,0);
  \coordinate (g3t2) at (6,1);
  \coordinate (g3t3) at (8,1);
  \coordinate (g3t4) at (6,2);
  \coordinate (g3t5) at (7,2);
  \draw[thick] (g3t1) -- (g3t2);
  \draw[thick] (g3t1) -- (g3t3);
  \draw[thick] (g3t2) -- (g3t4);
  \draw[thick] (g3t1) -- (g3t5);
  \draw[thick] (g3t3) -- (g3t5);
\end{tikzpicture}
}

\newcommand{\pgfour}{
\begin{tikzpicture}[scale=0.2]
  \coordinate (g4t1) at (10,0);
  \coordinate (g4t2) at (9,1);
  \coordinate (g4t3) at (11,1);
  \coordinate (g4t4) at (10,1);
  \coordinate (g4t5) at (10,2);
  \draw[thick] (g4t1) -- (g4t2);
  \draw[thick] (g4t1) -- (g4t3);
  \draw[thick] (g4t1) -- (g4t4);
  \draw[thick] (g4t2) -- (g4t5);
  \draw[thick] (g4t3) -- (g4t5);
\end{tikzpicture}
}

\newcommand{\pgfive}{
\begin{tikzpicture}[scale=0.2]
  \coordinate (g5t1) at (13,0);
  \coordinate (g5t2) at (12,1);
  \coordinate (g5t3) at (14,1);
  \coordinate (g5t4) at (13,2);
  \draw[thick] (g5t1) -- (g5t2);
  \draw[thick] (g5t1) -- (g5t3);
  \draw[thick] (g5t1) -- (g5t4);
  \draw[thick] (g5t2) -- (g5t4);
  \draw[thick] (g5t3) -- (g5t4);
\end{tikzpicture}

}

\begin{document}
\maketitle

\begin{abstract}
We introduce graph pattern-based association rules (GPARs) for directed labeled multigraphs such as RDF graphs. GPARs support both generative tasks, where a graph is extended, and evaluative tasks, where the plausibility of a graph is assessed. The framework goes beyond related formalisms such as graph functional dependencies, graph entity dependencies, relational association rules, graph association rules, multi-relation and path association rules, and Horn rules. Given a collection of graphs, we evaluate graph patterns under no-repeated-anything semantics, which allows the topology of a graph to be taken into account more effectively. We define a probability space and derive confidence, lift, leverage, and conviction in a probabilistic setting. We further analyze how these metrics relate to their classical itemset-based counterparts and identify conditions under which their characteristic properties are preserved.
\end{abstract}


\newcommand{\fulltau}{\mathcal{T}^n_{g,p_1,p_2}}
\newcommand{\fulltauG}{\mathcal{T}^n_{G,p_1,p_2}}
\newcommand{\indicatorboth}{\sum_{T \in \fulltau} \mathbb{1}_{\lbrace m_g(T,p_1,V_1) \land m_g(T, p_2,V_2)\rbrace}}
\newcommand{\indicatorleft}{\sum_{T \in \fulltau} \mathbb{1}_{\lbrace m_g(T,p_1,V_1)\rbrace}}
\newcommand{\indicatorright}{\sum_{T \in \fulltau} \mathbb{1}_{\lbrace m_g(T,p_2,V_2)\rbrace}}
\newcommand{\pleft}{\hat{P}_g(e(p_1,V_1;p_2))}
\newcommand{\pright}{\hat{P}_g(e(p_2,V_2;p_1))}
\newcommand{\pboth}{\hat{P}_g(e(p_1,V_1;p_2) \cap e(p_2,V_2;p_1))}

\section{Introduction}

Many systems can naturally be represented as graphs (specifically: directed labeled multigraphs), including social networks (with persons, groups, and organizations as node types and friendship, membership, or influence as relation types), communication networks (with individuals, devices, and servers connected by message exchange or physical connection relations), citation networks (with papers, authors, and journals linked by citation, authorship, or publication relations), transportation networks (with locations, routes, and vehicles connected by service, connection, or transfer relations), trade networks (with countries, firms, and commodities connected by export or import relations), biological interaction networks (with genes, proteins, and complexes linked by encoding, interaction, or regulation relations), %
molecular structure graphs (with atoms, molecules, and functional groups linked by bonding, containment, or spatial proximity relations), and chemical reaction networks (with compounds, reactions, and catalysts connected by input, output, or catalytic relations).

Beyond these concrete systems, formally represented knowledge can likewise be modeled as a graph. In formal knowledge representation, entities, concepts, and instances typically correspond to nodes, while semantic relations between them correspond to typed edges. This structure underlies frameworks such as the Resource Description Framework (RDF) and knowledge graphs more broadly, where information is expressed as triples linking subjects, predicates, and objects. %

Even when the primary data is not inherently graph-structured -- such as a sequence of tokens in a natural language text -- one obtains graph-structured data when the primary data is combined with annotations (e.g., where in the text an entity of type person is mentioned, or which relation is expressed between which entities), metadata (e.g., from which document, created by which author, the text was extracted, or a description of how the annotations were created), and background knowledge (e.g., the profession of the mentioned person or that the expressed relation is symmetric).

It can be interesting to study regularities in these graphs by mining frequent patterns and rules. For two frequent patterns $p_1$ and $p_2$ that we combine into a rule, we might find that it frequently is the case that if the graph is matched by the pattern $p_1$ in some area of the graph, then the graph is also matched by the pattern $p_2$ in that same area of the graph. In the context of an evaluative task, if we find that $p_1$ matches but $p_2$ does not, then this might indicate an irregularity in that area. In the context of a generative task, we might want to extend the graph so that $p_2$ matches in that area.

Whereas in some situations we are dealing with a single graph, which we refer to as the \textit{single graph setting}, in other situations we are dealing with a set of graphs of the same kind, which we refer to as the \textit{graph-transactional setting}. For example, we may study multiple different social networks or multiple snapshots of the same social network over time. Similarly, we may work with collections of communication networks, citation networks, transportation networks, trade networks, biological networks, molecular structures, or chemical reaction networks. Furthermore, we can transform a set of natural language sentences into a set of graphs where each sentence is enriched with annotations and metadata, thus transforming the corpus into a set of graphs.

Having available a set of graphs can be different from having available a single graph. 
While there might be situations where we can combine a set of graphs into one graph and then analyze that graph, in other situations it can be more appropriate to keep the graphs separate. For example, let's consider that we have given a set of natural language sentences and represent each sentence, together with the information which relation is expressed in that sentence, as a graph and then combine these graphs into a single graph. 
Then we mine frequent patterns from this graph where these patterns are intended to be used to classify which relation is expressed in a graph representing a sentence where we don't know which relation is expressed.
We would not want to obtain patterns that span multiple sentences, as in the classification scenario only a single sentence is given. Thus, either we need to restrict pattern mining to "stay within a sentence" or we keep the graphs separated. Thus, keeping the graphs apart can be relevant.

In traditional data mining, predictive relationships such as, "if a set contains all members of the set $A$, then it will also contain all members of the set $B$," are expressed as association rules written as $A\Rightarrow B$. Association rules are well known from market basket analysis and are particularly attractive because they are inherently explainable. In this paper, we extend the concept of association rules
from itemset-based association rules to graph pattern-based association rules.

Our main contributions are the following:
\begin{itemize}
\item We introduce the concept of graph pattern-based association rule for directed labeled multigraphs (which corresponds to the structure of RDF graphs).
\item We consider the case where graph patterns are evaluated under the no-repeated-anything semantics, giving graph patterns a higher discriminative power compared to when they are evaluated under the standard homomorphism semantics.

\item We show that graph pattern-based association rules evaluated under no-repeated-anything semantics admit a probability space formulation. 

\item 
We define the metrics support, confidence, lift, leverage, and conviction, which are well-known from itemset-based association rules, both for the single graph setting, as well as for the graph transactional setting where in the latter we distinguish between micro-averaged and macro-averaged scores. 
\item We analyze how these metrics relate to the classical metrics for itemset-based association rules, identifying which characteristics they share and how they differ in capturing associations.

\end{itemize}

The paper is structured as follows. Chapter \ref{sec:preliminaries} covers preliminaries, including itemset-based association rules, their interpretation and quality metrics, directed labeled multigraphs, graph patterns, and graph pattern evaluation semantics. Chapter~\ref{sec:gpar} introduces the concept of graph pattern-based association rules, their interpretation, and potential applications, including generative and evaluative tasks. Chapter~\ref{sec:metrics} presents metrics for graph pattern-based association rules. First, we define the underlying probability space in Section~\ref{sec:probspace}. Section~\ref{sec:metricssingle} introduces metrics for the single-graph setting, while Section~\ref{sec:metricstransactional} addresses the graph transactional setting, in both micro-averaged and macro-averaged variants. In Section~\ref{sec:reframing}, we show how evaluation of metrics in the single-graph setting and micro-averaged graph-transactional setting can be reframed in terms of metrics for classical itemset-based association rules. Section \ref{sec:characteristics} provides a detailed analysis of all metrics, highlighting the extent to which they resemble classical metrics in terms of their characteristics. Chapter~\ref{sec:relatedwork} discusses related work, and Chapter~\ref{sec:summaryandconclusion} concludes with a summary and conclusions. The Appendix contains proofs regarding the characteristics of the proposed metrics.

\section{Preliminaries}\label{sec:preliminaries}

\subsection{Association Rules and their Quality Metrics}
An association rule $A \Rightarrow B$ is constructed from two itemsets $A$ and $B$ where $A$ and $B$ are non-empty and disjoint and $A$ and $B$ are both subset of a set $I$ of items. $A$ is called the antecedent and $B$ is called the consequent of the rule. To distinguish these association rules from the graph pattern-based association rules (GPARs or GPA rules) introduced later in this paper, we here call them itemset-based association rules (ISARs or ISA rules). Itemset-based association rules are mined from a non-empty bag $T$ (called transaction database) of transactions where each transaction is an itemset. %
Itemset-based association rules are known from shopping basket analysis \cite{apriori}. For example, $ \lbrace bread, cheese \rbrace \Rightarrow \lbrace butter \rbrace$ is an association rule. 
We formalize the meaning of a rule $A \Rightarrow B$ in the context of some itemset $S$ as follows:
\begin{equation*}A \subseteq S \vdash_p B \subseteq S\end{equation*}

The statement expresses that if we have observed that A is a subset of $S$, then, given the association rule $A \Rightarrow B$, we can probabilistically entail that we can observe that $B$ is a subset of $S$. In the context of shopping basket analysis we can imagine a customer in the process of shopping where, until checkout, items are added to (but never removed from) a shopping basket. Since we do not know the set $S$ of items the customer will have selected at the end of the shopping process, we can only make observations based on the items that were selected so far and make predictions about which items the customer will have selected at the end of the shopping process.
Thus, if we observe that a customer has added bread and cheese to their shopping basket, then we predict that the customer will add butter to the shopping basket. It is important to note that the rule does not express that $A \subseteq S$ \textit{causes} $B \subseteq S$ to be the case.\\

We regard a rule $A \Rightarrow B$ with $B \subseteq A$ as trivial.
The rule is trivial because for any itemset $S$ with $A \subseteq S$ we can conclude, without even making further observations about $S$, especially without making observations about $S \setminus A$, that $B \subseteq S$, i.e., knowing that $A$ is a subset of a transaction $S$ is sufficient for knowing that $B$ is subset of the transaction $S$ for any transaction database. Thus, the rule does not allow us to make non-trivial predictions about $S$.
For any itemset $S$ with $A \subseteq S$, the fact that $B \subseteq S$ does not depend on $S$, but it depends on the specific relation between $A$ and $B$, namely $B \subseteq A$. %
Note that in practice one discards rules where $A \cap B \neq \varnothing$ and $B \not \subseteq A$, as these rules are partially redundant.\\

Multiple metrics have been introduced in the context of association rules.
We present and discuss well-known metrics to score itemsets (namely: absolute and relative support) and to score association rules (namely: confidence, lift, leverage, and conviction). We assume the transaction database $T$ to be non-empty. %
Metrics for a non-empty itemset $S$:
\[
\begin{array}{rl}
\operatorname{absolute-support}_T : & \mathcal{P}(I) 
 \longrightarrow \mathbb{N}, \\[6pt]
 & S \longmapsto |\lbrace \, t \in T \mid S \subseteq t \,\rbrace|\\[6pt]
 & = \sum_{t \in T} \mathbb{1}_{\left\lbrace S \subseteq t \right\rbrace}\\[6pt]
\operatorname{relative-support}_T : & \mathcal{P}(I) 
 \longrightarrow [0,1] \cap\mathbb{Q}, \\[6pt]
 & S \longmapsto \dfrac{\operatorname{absolute-support}_T(S)}{|T|}\\[6pt]
 & = \dfrac{\sum_{t \in T} \mathbb{1}_{\left\lbrace S \subseteq t \right\rbrace}}{|T|}
\end{array}
\]

Absolute support ($\operatorname{absolute-support}_T(S)$) of an itemset $S$ is the number of transactions that contain the members of the given set $S$, 
whereas relative support\footnote{Relative support was introduced as \textit{support} by Agrawal and Srikant~\cite{apriori}.} 
($\operatorname{relative-support}_T(S)$) is the fraction of transactions in the transaction database that contain the members of the given set $S$. 
Typically, one is interested in finding itemsets with sufficiently high absolute or relative support, from which then rules are built. 
Both metrics are anti-monotonic, i.e., for any transaction database $T$ and any two sets $S$ and $S'$ such that $S \subset S'$ it is the case that $\operatorname{absolute-support}_T(S) \geq \operatorname{absolute-support}_T(S')$ and $\operatorname{relative-support}_T(S) \geq \operatorname{relative-support}_T(S')$. This property is relevant, as it allows to prune the search space when mining frequent itemsets: any itemset that has a non-frequent subset (i.e., that subset has a support value below a user-defined threshold value) cannot be frequent. Instead of evaluating the support of every itemset in the powerset $\mathcal{P}(I)$, once we found that a set is non-frequent we can discard all members of the powerset that are supersets of this set. This is known as the Apriori principle \cite{apriori}.\\

Itemset-based association rules admit a probability space formulation \cite{sharma2022novel}.
Although association rules are usually defined in terms of itemsets and transaction counts, they can equivalently be formulated in a probability space $(\Omega,\mathcal{F},\hat{P}_T)$ as follows:
\begin{itemize}
\item The sample space $\Omega$ is the bag $T$ of transactions, because we pick a transaction uniformly at random from the transaction database $T$.
\item The sigma-algebra $\mathcal{F}$ is the powerset $\mathcal{P}(\Omega)$, because $\Omega$ is finite and thus every subset of $\Omega$ is a measurable event.
\item The empirical probability measure $\hat{P}_T : \mathcal{F} \to [0,1]$ is defined for every event $X \subseteq \Omega$ by $\hat{P}_T(X) = \tfrac{|X|}{|T|}$.
Given an itemset $A \subseteq \mathcal{I}$, we associate it to the event $E_A \coloneq \lbrace \,t \in T \mid A \subseteq t \,\rbrace$. Thus, $\hat{P}_T(E_A)$ is the probability that for a transaction $t$, selected uniformly at random from $T$, $A \subseteq t$ is true. The empirical probability is measured as follows:
\[ \hat{P}_T(E_A) = \frac{|\lbrace \,t \in T \mid A \subseteq t \,\rbrace|}{|T|}\]
\end{itemize}

Assuming the transaction database $T$ to be finite and non-empty, $(\Omega,\mathcal{F},\hat{P}_T)$ is a probability space, because it satisfies the Kolmogorov axioms:
\begin{enumerate}
\item Non-negativity (Axiom 1).

For all $X \subseteq \Omega$, $|X| \geq 0$, and $|T| > 0$. Thus, $\hat{P}_T(A) = \tfrac{|A|}{|T|}$ is a rational and thus a real number. Furthermore, since $|A| \geq 0$ we have $\hat{P}_T(A) \geq 0$ and $\tfrac{|A|}{|T|} \geq 0$.

\item Normalization (Axiom 2).

Because $\Omega = T$, it follows that $\hat{P}_T(\Omega) = \tfrac{|\Omega|}{|T|} = \tfrac{|\Omega|}{|\Omega|} = 1$.

\item Countable additivity (Axiom 3).

Let $X = (X_1, \ldots,X_m)$ be a sequence of pairwise disjoint non-empty subsets of $\Omega$. Then
\[ \hat{P}_T\left(\bigcup_{X_i \in X} X_i\right) =
\frac{\left|\displaystyle \bigcup_{X_i \in X}(X_i)\right|}{|T|} = \frac{\displaystyle \sum_{X_i \in X}|X_i|}{|T|} = \sum_{X_i \in X} \frac{|X_i|}{|T|} = \sum_{X_i \in X} \hat{P}_T(X_i)\]
\end{enumerate}

Thus, $\operatorname{relative-support}_T(S)$ is equal to $\hat{P}_T(E_S)$.
Furthermore, we can describe the empirical joint and conditional probabilities as follows:
\[
\begin{array}{rl}
\hat{P}_T(E_A \cap E_B) =& \hat{P}_T(E_{A \cup B}) = \dfrac{|\lbrace \, t \in T \mid A \subseteq t \land B \subseteq t\,\rbrace|}{|T|}\\[10pt]
\hat{P}_T(E_B \mid E_A) =& \dfrac{|\lbrace \, t \in T \mid A \subseteq t \land B \subseteq t\,\rbrace|}{|\lbrace \, t \in T \mid A \subseteq t \,\rbrace|}
\end{array}
\]
$\hat{P}_T(E_A \cap E_B)$ denotes the empirical probability, estimated from the transaction database $T$, that for a transaction $t$, picked uniformly at random from $T$, it is the case that $A \subseteq t$ is true and $B \subseteq t$ true. 
$\hat{P}_T(E_B \mid E_A)$ denotes the empirical conditional probability, estimated from the transaction database $T$, that for a transaction $t$, picked uniformly at random from $T$, $B \subseteq t$ is true under the condition that $A \subseteq t$ is true.

In what follows we introduce metrics both based on transaction counts and, equivalently, based on probabilities. 
To quantify the strength of the associations of the two itemsets (where the first itemset has a support > 0) in a rule the confidence metric was introduced.\footnote{The confidence metric was introduced by Agrawal and Srikant \cite{apriori}.} It is defined as follows:
\[
\begin{array}{rl}
\operatorname{confidence}_T : & \lbrace \, A \Rightarrow B \mid A,B \in \mathcal{P}(I)\setminus \lbrace \varnothing\rbrace \, \rbrace 
 \longrightarrow [0,1] \cap \mathbb{Q}, \\[6pt]
& A \Rightarrow B \longmapsto \dfrac{\operatorname{relative-support}_T(A \cup B)}{\operatorname{relative-support}_T(A)}\\[8pt]
& = \dfrac{|\lbrace ~t \in T \mid A \cup B \subseteq t ~\rbrace|}{|\lbrace ~t \in T \mid A \subseteq t ~\rbrace|}\\[8pt]
 & = \dfrac{\sum_{t \in T} \mathbb{1}_{\left\lbrace A \subseteq t \land B \subseteq t\right\rbrace}}{\sum_{t \in T} \mathbb{1}_{\left\lbrace A \subseteq t \right\rbrace}}\\[8pt]
 & = \hat{P}_T(E_B \mid E_A)
\end{array}
\]

Confidence is the fraction of the number of true predictions made by the rule (the number of transactions that contain both the antecedent $A$ and the consequent $B$) to the number of all predictions made by the rule (i.e., the number of rules that contain the antecedent $A$). Rules with a high confidence are desirable. 
Note that we equally define confidence either as the ratio of relative support of $A \cup B$ to the relative support of $A$ or as the ratio of the absolute support of $A \cup B$ to the absolute support of $A$. 
Confidence of the rule $A \Rightarrow B$ can be interpreted as the empirical conditional probability $\hat{P}_T(E_B \mid E_A)$.

Typically, one is interested in rules with a support score above some threshold and with a confidence score above some threshold where these thresholds are application-specific. 
However, association rules with high support and high confidence do not necessarily disclose truly interesting event relationships \cite{brin1997beyond}. Consider, for example, that one has found that the rule $\lbrace butter \rbrace \Rightarrow \lbrace bread \rbrace$ has both a sufficiently high support and a sufficiently high confidence. But, if the transaction database we have analyzed comes from a bakery where almost all transactions contain bread, the the rule might actually be not interesting, because bread is purchased frequently regardless of whether butter is purchased. 
Therefore, the lift\footnote{The lift metric was introduced as \textit{interest} by Brin et al. \cite{brin1997dynamic}.} score was introduced, which divides confidence of the rule by the relative support of the consequent.
\[
\begin{array}{rl}
\operatorname{lift}_T : & \lbrace \, A \Rightarrow B \mid A,B \in \mathcal{P}(I)\setminus \lbrace \varnothing\rbrace \, \rbrace 
 \longrightarrow  [0,\infty) \cap \mathbb{Q}, \\[6pt]
& A \Rightarrow B \longmapsto \dfrac{\operatorname{confidence}_T(A \Rightarrow B)}{\operatorname{relative-support}_T(B)}\\[8pt]
& = \dfrac{|\lbrace ~t \in T \mid A \cup B \subseteq t ~\rbrace|\cdot|T|}{|\lbrace ~t \in T \mid A \subseteq t ~\rbrace|\cdot|\lbrace ~t \in T \mid B \subseteq t ~\rbrace|}\\[8pt]
 & = \dfrac{\hat{P}_T(E_B \mid E_A)}{\hat{P}_T(E_B)}\\[8pt]
 & = \dfrac{\left[\sum_{t \in T} \mathbb{1}_{\left\lbrace A \subseteq t \land B \subseteq t\right\rbrace}\right] \cdot |T|}{\left[ \sum_{t \in T} \mathbb{1}_{\left\lbrace A \subseteq t \right\rbrace} \right] \cdot \left[ \sum_{t \in T} \mathbb{1}_{\left\lbrace B \subseteq t\right\rbrace} \right]}
\end{array}
\]

A drawback of the lift score is that it only uses relative numbers. That means, it produces a high value for rules where the event $E_B$ strongly depends on the event $E_A$, even though the absolute support of these events might be very small. Thus, rules that might not be very interesting can have a high lift score.

For example, the rule $ \lbrace capers \rbrace \Rightarrow \lbrace anchovies \rbrace $ and the rule
$ \lbrace apples \rbrace \Rightarrow \lbrace bread \rbrace$ might have the same lift score, but the absolute support of the itemset $\lbrace capers \rbrace$ might be much smaller than the absolute support of the itemset $\lbrace apples \rbrace$
and the absolute support of the itemset $\lbrace anchovies \rbrace$ might be much smaller than the absolute support of the itemset $\lbrace bread \rbrace$. Thus, one could be more interested in the latter rule than in the former.

A metric that takes into account the correlation of the events $E_A$ and $E_B$, but also takes into account the support scores, is leverage, introduced by Piatetsky-Shapiro \cite{leverage}, which is defined as follows:
\[
\begin{array}{rl}
\operatorname{leverage}_T : & \lbrace \, A \Rightarrow B \mid A,B \in \mathcal{P}(I)\setminus \lbrace \varnothing\rbrace \, \rbrace
\longrightarrow \left[-\tfrac{1}{4}, \tfrac{1}{4}\right] \cap \mathbb{Q}, \\[6pt]
& A \Rightarrow B \longmapsto \operatorname{rel.-support}_T(A \cup B) - \operatorname{rel.-support}_T(A) \cdot \operatorname{rel.-support}_T(B)\\[6pt]
& = \dfrac{|\lbrace ~t \in T \mid A \cup B \subseteq t ~\rbrace|}{|T|} - \dfrac{|\lbrace ~t \in T \mid A \subseteq t ~\rbrace|}{|T|} \cdot \dfrac{|\lbrace ~t \in T \mid B \subseteq t ~\rbrace|}{|T|}\\[6pt]
& = \dfrac{\sum_{t \in T} \mathbb{1}_{\left\lbrace A \subseteq t \land B \subseteq t \right\rbrace}}{|T|} - \dfrac{\sum_{t \in T} \mathbb{1}_{\left\lbrace A \subseteq t \right\rbrace}}{|T|} \cdot \dfrac{\sum_{t \in T} \mathbb{1}_{\left\lbrace B \subseteq t \right\rbrace}}{|T|}\\[6pt]
& = \hat{P}_T(E_{A \cup B}) - \hat{P}_T(E_A) \cdot \hat{P}_T(E_B)
\end{array}
\]

A drawback of leverage is that the metric is symmetric. Thus, what leverage metric is correlation but not the degree of implication. A metric that is not symmetric but instead is sensitive to the direction of the implication %
is conviction, introduced by Brin et al. \cite{brin1997dynamic}. 
The conviction metric makes use of two kinds of failure rates:
\begin{enumerate}
\item

One the one hand, we are interested in the probability with which the rule $A \Rightarrow B$ fails, i.e., the \textit{empirical failure rate}: %
$\hat{P}_T(E_A \cap \neg E_B)$ (with $\neg E_B = \Omega \setminus E_B$) 
which is equal to $\hat{P}_T(A) \cdot  \operatorname{confidence}_T(A \Rightarrow \neg B)$ which is equal to $\hat{P}_T(A) \cdot (1 - \operatorname{confidence}_T(A \Rightarrow B))$.

\item
On the other hand, we are interested in the \textit{expected failure rate under independence assumption} of the events $E_A$ and $E_B$ (or, equivalently, the events $E_A$ and $\neg E_B$). 
Under the assumption of independence, 
$\hat{P}_T(E_A \cap \neg E_B) = \hat{P}_T(E_A) \cdot \hat{P}_T(\neg E_B) = \hat{P}_T(E_A) \cdot (1-\hat{P}_T(E_B))$, which is equal to $\operatorname{relative-support}_T(A)\cdot (1 - \operatorname{relative-support}_T(B))$.
\end{enumerate}

Conviction of a rule is then defined as the fraction of the expected failure rate under independence assumption to the empirical failure rate:
\[
\begin{array}{rl}
\operatorname{conviction}_T : & \lbrace \, A \Rightarrow B \mid A,B \in \mathcal{P}(I)\setminus \lbrace \varnothing\rbrace \, \rbrace 
 \longrightarrow  \mathbb{Q}_{> 0} \cup \lbrace \infty \rbrace, \\[6pt]
& A \Rightarrow B  \longmapsto \dfrac{1 - \operatorname{relative-support}_T(B)}{1 - \operatorname{confidence}_T(A \Rightarrow B)}\\[10pt]
& =  \dfrac{1 - \hat{P}_T(E_B)}{1 - \hat{P}_T(E_B \mid E_A)}\\[8pt]
& =  \dfrac{\left(\sum_{t \in T} \mathbb{1}_{\left\lbrace A \subseteq t \right\rbrace}\right) \cdot \left(|T| - \sum_{t \in T} \mathbb{1}_{\left\lbrace B \subseteq t \right\rbrace}\right)}{|T| \cdot \left(\left[\sum_{t \in T} \mathbb{1}_{\left\lbrace A \subseteq t\right\rbrace}\right] - \left[\sum_{t \in T} \mathbb{1}_{\left\lbrace A \subseteq t \land B \subseteq t \right\rbrace}\right]\right)} \end{array}
\]

Conviction is undefined if $\hat{P}_T(E_B) = 1$. In the case that $\hat{P}_T(E_B \mid E_A) = 1$, or, equivalently, $\operatorname{confidence}_T(A \Rightarrow B) = 1$, then $\operatorname{conviction}_T(A \Rightarrow B) \to \infty$: since the denominator approaches zero from the positive side, the fraction diverges to infinity.
We define $\operatorname{conviction}_T(A \Rightarrow B)$ to \textit{be} $\infty$ in the case that $\operatorname{confidence}_T(A \Rightarrow B) = 1$.

Note that for a rule $A \Rightarrow B$ and a non-empty transaction database,
the confidence score is not defined if $\hat{P}_T(E_A) = 0$ and
the lift score is not defined if $\hat{P}_T(E_A) = 0$ or $\hat{P}_T(E_B) = 0$.
In practice, however, one derives itemsets from a non-empty transaction database and constructs rules from itemsets with a non-zero support. Thus, 
the situation does not occur that 
the probability of the antecedent (i.e., $\hat{P}_T(E_A)$) or the probability of the consequent
(i.e., $\hat{P}_T(E_B)$) is zero which would
lead to the confidence score or the lift score to be undefined.

\subsection{Graphs, Graph Patterns, and Graph Pattern Evaluation Semantics}\label{sec:graphsintro}

Let $\mathcal{T} = \lbrace t_1, t_2, \ldots \rbrace$ denote an infinite set of terms. A directed labeled multigraph (DLM) $g$ is a subset of $\mathcal{T} \times \mathcal{T} \times \mathcal{T}$. We refer to an element of a graph as a triple and we interpret a triple $(t_i,t_j,t_k) \in g$ as a directed labeled edge where the edge is labeled with the term $t_j$, the edge originates from the node labeled with the term $t_i$, and the edge points to the node labeled with the term $t_k$. Given a graph $g$, $\mathcal{T}_g \subseteq \mathcal T$ denotes the set of terms that occur in $g$.\footnote{I.e., $\mathcal{T}_g = \bigcup_{(t_i,t_j,t_k) \in g} \lbrace t_i, t_j, t_k\rbrace$}

We require that for a graph $g$ for any term $t \in \mathcal{T}_g$ there exists at most one node with that label in $g$ -- thus, a node is \textit{identified} by its label. It follows from our definition of DLM that given an edge created in a graph for a triple $(t_i,t_j,t_k)$, the edge is not identified by the label of the edge, $t_j$, but instead by the triple $(t_i,t_j,t_k)$. 
This requirement has two important consequences: i) the topology of a graph is uniquely determined by a set of triples, ii) in order to check whether some graph $s$ is contained in some graph $g$, instead of having to carry out the subgraph isomorphism check we can carry out a less complex task, namely the subset test: $s \subseteq g$.

A further consequence of this requirement is that we can define the concept of graph association rule (GAR) where a rule has the form $g_1 \Rightarrow g_2$ and $g_1$ and $g_2$ are DLMs. We can apply the rule to probabilistically entail that if a graph $s$ contains the subgraph $g_1$, i.e., $g_1 \subseteq s$, then it also contains the subgraph $g_2$, i.e., $g_2 \subseteq s$:
\begin{equation*}g_1 \subseteq s \vdash_p g_2 \subseteq s\end{equation*}

In the context of graph association rules, a triple is seen as an item and the set $I$ of items thus contains triples, i.e., $I \subseteq \mathcal{T} \times \mathcal{T} \times \mathcal{T}$. 
Frequent itemset mining on a transaction database where each transaction is a graph is then equal to frequent subgraph mining and the metrics defined in the ISAR context can also be applied in the GAR context. Analogously to ISARs, a GAR of the form $g_1 \Rightarrow g_2$ where $g_2 \subseteq g_1$ can be considered trivial. Furthermore, the rule would be discarded if $g_1 \cap g_2 \neq \varnothing$ and $g_2 \not \subseteq g_1$ because it is partially redundant.\\

Figure \ref{fig:examplegraph} shows an example of an DLM, as it could occur in the context of natural language processing as a rich representation of the sentence \texttt{"Alice Smith likes Bob and Charlie likes Bob."} The graph contains a node for each of the eight tokens. These nodes are labeled/identifies with the terms $t_1, \ldots, t_8$. Each token node is connected via an \texttt{next} edge to the next token, e.g., $(t_3, next, t_4)$ and with a \texttt{label} edge to the actual token string from the sentence. Further information represents results from the task of named entity recognition. For example, the tokens $t_1$ and $t_2$ are both parts of a mention (\texttt{pom}) which is a mention of the entity (\texttt{moe}) \texttt{Alice}. A result of the relation detection task is that the sentence expresses that the entity \texttt{Alice} likes the entity Bob. Thus, the graph expresses that there is an instance of the relation (\texttt{ior}) \texttt{likes} where \texttt{Alice} is the head of the relation instance (\texttt{hori}) and \texttt{Bob} is the tail of the relation instance (\texttt{tori}). The relation instance is an instance of the relation (\texttt{ior}) \texttt{likes}. Further information comes from background knowledge, e.g., that \texttt{Alice} in of type (\texttt{isA}) \texttt{Person}, and that it is known that \texttt{Alice} likes \texttt{Bob}. %
This graphs serves as an example for the situation where the primary data is not inherently graph-structured, but a graph-structured representation is obtained when the data is combined with annotations and background knowledge.\\

\begin{figure}
$\begin{aligned}
g ~=~ & \lbrace \,(t_1,label,"Alice"), (t_2,label,"Smith"), (t_3,label,"likes"), (t_4,label,"Bob"),\\
& \phantom{\lbrace\,} (t_5,label,"and"), (t_6,label,"Charlie"),(t_7,label,"likes"),(t_8,label,"Bob")\\
& \phantom{\lbrace\,} (t_1, next, t_2), (t_2, next, t_3), (t_3, next, t_4), (t_4, next, t_5), (t_5, next, t_6), (t_6, next, t_7),\\
& \phantom{\lbrace\,}  (t_7, next, t_8),\\
& \phantom{\lbrace\,} (t_1, pom, em_1), (t_2, pom, em_1), (t_4, pom, em_2), (t_6, pom, em_3), (t_8, pom, em_4),\\
& \phantom{\lbrace\,} (em_1, moe, Alice), (em_2, moe, Bob), (em_3, moe, Charlie), (em_4, moe, Bob),\\
& \phantom{\lbrace\,} (Alice, hori, ri_1), (Bob, tori, ri_1), (Charlie, hori, ri_2), (Bob, tori, ri_2),\\
& \phantom{\lbrace\,} (ri_1, ior, likes), (ri_2, ior, likes),\\
& \phantom{\lbrace\,} (Alice, likes, Bob),\\
& \phantom{\lbrace\,} (Alice, isA, Person), (Bob, isA, Person), (Charlie, isA, Person) \,\rbrace%
\end{aligned}$

\begin{tikzpicture}[
    >=latex,
    node style/.style={circle, draw, minimum size=6pt, inner sep=0pt, fill=white},
    nodelabel/.style={font=\small, inner sep=1pt},
    edgelabel/.style={font=\small, inner sep=1pt, fill=white}
  ]
  
\begin{scope}
\clip (0,-8) rectangle (12,4);

  \node[node style] (lAlice) at (1,1.2) {};
  \node[node style] (lSmith) at (2.5,1.2) {};
  \node[node style] (llikes) at (4,1.2) {};
  \node[node style] (lBob) at (5.5,1.2) {};
  \node[node style] (land) at (7,1.2) {};
  \node[node style] (lCharlie) at (8.5,1.2) {};  
  
  \node[nodelabel, left=0pt, above=5pt] at (lAlice) {$"Alice"$};
  \node[nodelabel, left=0pt, above=5pt] at (lSmith) {$"Smith"$};
  \node[nodelabel, left=6pt, above=5pt] at (llikes) {$"likes"$};
  \node[nodelabel, left=6pt, above=5pt] at (lBob) {$"Bob"$};
  \node[nodelabel, left=0pt, above=5pt] at (land) {$"and"$};
  \node[nodelabel, left=0pt, above=5pt] at (lCharlie) {$"Charlie"$};
  
  \node[node style] (tAlice) at (1,0) {};
  \node[node style] (tSmith) at (2.5,0) {};
  \node[node style] (tlikes1) at (4,0) {};
  \node[node style] (tBob1) at (5.5,0) {};
  \node[node style] (tand) at (7,0) {};
  \node[node style] (tCharlie) at (8.5,0) {};  
  \node[node style] (tlikes2) at (10,0) {};
  \node[node style] (tBob2) at (11.5,0) {};

  \node[nodelabel, left=6pt, above=3pt] at (tAlice) {$t_1$};
  \node[nodelabel, left=6pt, above=3pt] at (tSmith) {$t_2$};
  \node[nodelabel, left=6pt, above=3pt] at (tlikes1) {$t_3$};
  \node[nodelabel, left=6pt, above=3pt] at (tBob1) {$t_4$};
  \node[nodelabel, left=6pt, above=3pt] at (tand) {$t_5$};
  \node[nodelabel, left=6pt, above=3pt] at (tCharlie) {$t_6$};
  \node[nodelabel, left=6pt, above=3pt] at (tlikes2) {$t_7$};
  \node[nodelabel, left=6pt, above=3pt] at (tBob2) {$t_8$};
  
  \draw[->, thick] (tAlice) -- (lAlice);
  \draw[->, thick] (tSmith) -- (lSmith);
  \draw[->, thick] (tlikes1) -- (llikes);
  \draw[->, thick] (tBob1) -- (lBob);
  \draw[->, thick] (tand) -- (land);
  \draw[->, thick] (tCharlie) -- (lCharlie); 
      
  \node[edgelabel] at ($(tAlice)!0.5!(lAlice) + (0,0)$) {$label$};
  \node[edgelabel] at ($(tSmith)!0.5!(lSmith) + (0,0)$) {$label$};
  \node[edgelabel] at ($(tlikes1)!0.5!(llikes) + (0,0)$) {$label$};
  \node[edgelabel] at ($(tBob1)!0.5!(lBob) + (0,0)$) {$label$};  
  \node[edgelabel] at ($(tand)!0.5!(land) + (0,0)$) {$label$};
  \node[edgelabel] at ($(tCharlie)!0.5!(lCharlie) + (0,0)$) {$label$};

\draw[->] (tlikes2) .. controls +(1,3) and +(2,2) .. (llikes) node[midway, above=0pt] {$label$};
\draw[->] (tBob2) .. controls +(1,2.3) and +(2,2) .. (lBob) node[midway, above=0pt] {$label$};      

  \draw[->, thick] (tAlice) -- (tSmith);
  \draw[->, thick] (tSmith) -- (tlikes1);
  \draw[->, thick] (tlikes1) -- (tBob1);
  \draw[->, thick] (tBob1) -- (tand);
  \draw[->, thick] (tand) -- (tCharlie);
  \draw[->, thick] (tCharlie) -- (tlikes2);
  \draw[->, thick] (tlikes2) -- (tBob2);  
  
  \node[edgelabel] at ($(tAlice)!0.5!(tSmith) + (-0.0,-0.0)$) {$next$};  
  \node[edgelabel] at ($(tSmith)!0.5!(tlikes1) + (-0.0,-0.0)$) {$next$};  
  \node[edgelabel] at ($(tlikes1)!0.5!(tBob1) + (-0.0,-0.0)$) {$next$};  
  \node[edgelabel] at ($(tBob1)!0.5!(tand) + (-0.0,-0.0)$) {$next$};
  \node[edgelabel] at ($(tand)!0.5!(tCharlie) + (-0.0,-0.0)$) {$next$};
  \node[edgelabel] at ($(tCharlie)!0.5!(tlikes2) + (-0.0,-0.0)$) {$next$};  
  \node[edgelabel] at ($(tlikes2)!0.5!(tBob2) + (-0.0,-0.0)$) {$next$};

  \node[node style] (e1) at (1.75,-1) {};
  \node[node style] (e2) at (5.5,-1) {};
  \node[node style] (e3) at (8.5,-1) {};
  \node[node style] (e4) at (11.5,-1) {};    
  \node[nodelabel, left=4pt] at (e1) {$em_1$};
  \node[nodelabel, left=4pt] at (e2) {$em_2$};
  \node[nodelabel, left=4pt] at (e3) {$em_3$};
  \node[nodelabel, right=4pt] at (e4) {$em_4$};
    
  \draw[->, thick] (tAlice) -- (e1);
  \draw[->, thick] (tSmith) -- (e1);
  \draw[->, thick] (tBob1) -- (e2);
  \draw[->, thick] (tCharlie) -- (e3);
  \draw[->, thick] (tBob2) -- (e4);
  
  \node[edgelabel] at ($(tAlice)!0.5!(e1) + (-0.0,-0.0)$) {$pom$};
  \node[edgelabel] at ($(tSmith)!0.5!(e1) + (-0.0,-0.0)$) {$pom$};    
  \node[edgelabel] at ($(tBob1)!0.5!(e2) + (-0.0,-0.0)$) {$pom$}; 
  \node[edgelabel] at ($(tBob2)!0.5!(e4) + (-0.0,-0.0)$) {$pom$};  
  \node[edgelabel] at ($(tCharlie)!0.5!(e3) + (-0.0,-0.0)$) {$pom$};

  \node[node style] (eAlice) at (1.75,-2) {};
  \node[node style] (eBob) at (7,-3) {};
  \node[node style] (eCharlie) at (8.5,-2) {};
  \node[nodelabel, left=4pt] at (eAlice) {$Alice$};
  \node[nodelabel, right=8pt, above=7pt] at (eBob) {$Bob$};
  \node[nodelabel, left=4pt] at (eCharlie) {$Charlie$};
    
  \draw[->, thick] (e1) -- (eAlice);
  \node[edgelabel] at ($(e1)!0.5!(eAlice) + (-0.0,-0.0)$) {$moe$};
  
  \draw[->, thick] (e2) -- (eBob);
  \node[edgelabel] at ($(e2)!0.5!(eBob) + (-0.0,-0.0)$) {$moe$};
  
  \draw[->, thick] (e3) -- (eCharlie);
  \node[edgelabel] at ($(e3)!0.5!(eCharlie) + (-0.0,-0.0)$) {$moe$};
  
  \draw[->, thick] (e4) -- (eBob);
  \node[edgelabel] at ($(e4)!0.5!(eBob) + (-0.0,-0.0)$) {$moe$};

  \node[node style] (ri1) at (4,-4) {};
  \node[node style] (ri2) at (8,-4) {};
  \node[nodelabel, left=4pt] at (ri1) {$ri_1$};
  \node[nodelabel, right=4pt] at (ri2) {$ri_2$};
      
  \draw[->, thick] (eAlice) -- (ri1);
  \draw[->, thick] (eBob) -- (ri1);
  \draw[->, thick] (eCharlie) -- (ri2);
  \draw[->, thick] (eBob) -- (ri2);
  
  \node[edgelabel] at ($(eAlice)!0.5!(ri1) + (-0.0,-0.0)$) {$hori$};
  \node[edgelabel] at ($(eBob)!0.5!(ri1) + (-0.0,-0.0)$) {$tori$};
  \node[edgelabel] at ($(eCharlie)!0.5!(ri2) + (-0.0,-0.0)$) {$hori$};
  \node[edgelabel] at ($(eBob)!0.5!(ri2) + (-0.0,-0.0)$) {$tori$};
  
  \node[node style] (plikes) at (6,-5) {};
  \draw[->, thick] (ri1) -- (plikes);
  \draw[->, thick] (ri2) -- (plikes);
  \node[edgelabel] at ($(ri1)!0.6!(plikes) + (-0.0,-0.0)$) {$ior$};
  \node[edgelabel] at ($(ri2)!0.5!(plikes) + (-0.0,-0.0)$) {$ior$};

  \node[nodelabel, left=6pt] at (plikes) {$likes$};      
        
  \node[node style] (cPerson) at (2.5,-6) {};
  \node[nodelabel, left=4pt] at (cPerson) {$Person$};
  
  \draw[->, thick] (eBob) -- (cPerson);
  \draw[->, thick] (eAlice) -- (cPerson);
  \node[edgelabel] at ($(eBob)!0.5!(cPerson) + (-0.3,-0.3)$) {$isA$};
  \node[edgelabel] at ($(eAlice)!0.5!(cPerson) + (-0.0,-0.0)$) {$isA$};

  \draw[->] (eCharlie) .. controls +(8,-5) and +(0,-2) .. (cPerson) node[midway, above=0pt] {$isA$};

  \draw[->, thick] (eAlice) -- (eBob);
  \node[edgelabel] at ($(eAlice)!0.5!(eBob) + (-0.0,-0.0)$) {$likes$};

\end{scope}
\end{tikzpicture}
\caption{Example of a directed labeled multigraph (DLM) as it could occur in the context of Natural Language Processing, representing the sentence \texttt{"Alice Smith likes Bob and Charlie likes Bob"}, some annotations (e.g., which entities are mentioned and which relations are expressed), and some background knowledge (e.g., Alice is a person and it is known that Alice likes Bob).
 $pom$ stands for \textit{part of mention},
$em$ stands for \textit{entity mention},
$moe$ stands for \textit{mention of entity},
$hori$ stands for \textit{head of relation instance},
$tori$ stands for \textit{tail of relation instance},
$ri$ stands for \textit{relation instance}, and
$ior$ stands for \textit{instance of relation}.}
\label{fig:examplegraph}

\end{figure}
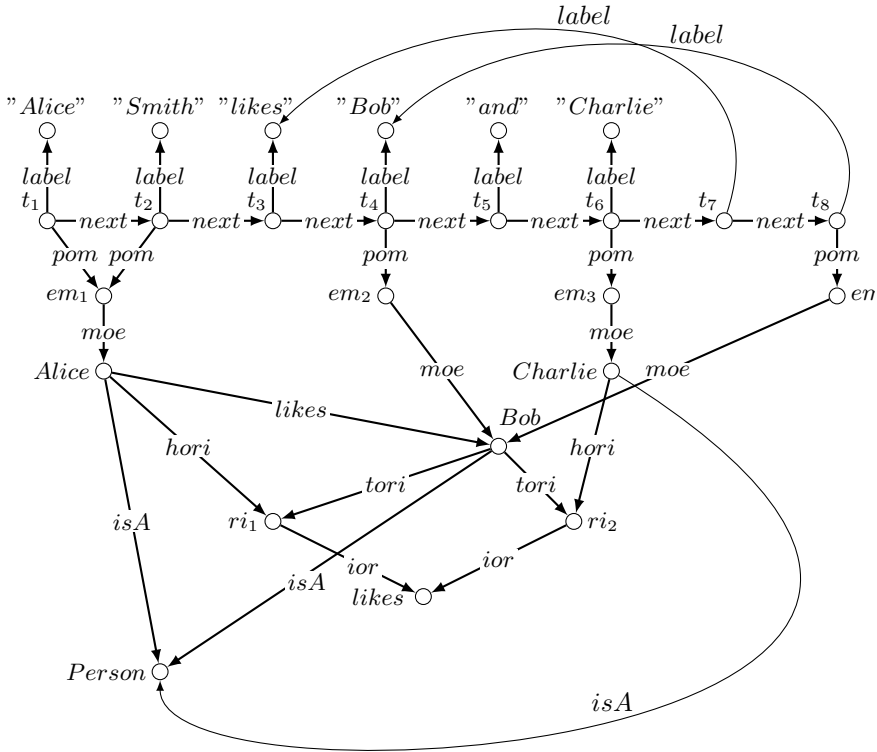

An RDF graph is a set of RDF triples of the form $(s,p,o) \in (\mathcal{I} \cup \mathcal{B}) \times \mathcal{I} \times (\mathcal{I} \cup \mathcal{B} \cup \mathcal{L})$, where $\mathcal{I}$, $\mathcal{B}$ and $\mathcal{L}$ are infinite pairwise disjoint sets of IRIs, blank nodes, and literal values, respectively. 
We can represent the structure of an RDF graph $g$ in the form of an DLM when we ignore the distinction between $\mathcal{I}$, $\mathcal{B}$, and $\mathcal{L}$ by letting the set $\mathcal{T}$ of terms be $\mathcal{I} \cup \mathcal{B} \cup \mathcal{L}$. The RDF data model\footnote{See \url{https://www.w3.org/TR/rdf11-concepts/}} can be seen as a specialization of DLM. 
The concept of graph association rules is directly applicable to RDF graphs under the condition that the RDF graphs are ground RDF graphs, which means that these graphs no not contain blank nodes. 
In the presence of blank nodes one would need to go beyond subgraph checks and perform subgraph
 isomorphism\footnote{See \url{https://www.w3.org/TR/rdf12-concepts/\#graph-isomorphism}} and \textit{simple entailment}\footnote{See \url{https://www.w3.org/TR/rdf11-mt/\#simpleentailment}}
checks, since blank nodes are interpreted as existentially qualified variables.\\

Graph association rules have important limitations. 
Let's assume we would like to derive a rule from a set of graphs that are similar to the graph presented in Figure \ref{fig:examplegraph} where the rule predicts whether in a graph it is expressed that one person likes another person. 
The antecedent graph of such a (not particularly suitable) rule could contain the triples
$(t_6, next, t_7)$, $(t_7, label, "likes")$, and $(t_7, next, t_8)$.
The terms $t_6$, $t_7$, and $t_8$ are identifiers of tokens in a sentence. Thus, the relation can only be detected in a sentence that contains at least eight tokens and where the label of the 7th token is "likes". Thus, the rule would not be general. Furthermore, what the rule does not consider is that directly before and directly after the 7th token, two entities of type Person need to be mentioned, but therefore the antecedent would need to also contain terms such as $em_3$ and $em_4$, which are specific to that sentence. Better generalization can be achieved by letting antecedent and consequent of a rule be graph patterns, which we define in what follows.\\

Let $\mathcal{V} = \lbrace v_1, v_2, \ldots \rbrace$ denote an infinite set of variables. A graph pattern is a subset of $(\mathcal{T} \cup \mathcal{V}) \times (\mathcal{T} \cup \mathcal{V}) \times (\mathcal{T} \cup \mathcal{V})$. We refer to an element of a graph pattern as a triple pattern. Given a graph pattern $p$, $\mathcal{T}_p \subseteq \mathcal V$ denotes the set of terms that occur in $p$,\footnote{I.e., $\mathcal{T}_g = \bigcup_{(x_i,x_j,x_k) \in g} \left[ \lbrace x_i, x_j, x_k \rbrace \cap \mathcal{V} \right]$}
and $\mathcal{V}_p \subseteq \mathcal V$ denotes the set of variables that occur in $p$.\footnote{I.e., $\mathcal{V}_p = \bigcup_{(x_i,x_j,x_k) \in g} \left[ \lbrace x_i, x_j, x_k \rbrace \cap \mathcal{T} \right]$}

The result of matching a graph pattern $p$ against a graph $g$ is a set $\Omega_{p,g}$ of mapping functions from $\mathcal{V}_p \cup \mathcal{T}_p$ to $\mathcal{T}_g$ where each mapping function has to satisfy a set of conditions. Which conditions need to be satisfied depends on the choice of an evaluation semantics.\footnote{See Angles et al. \cite{angles2017foundations} for an overview of graph pattern evaluation semantics.}
In the most general case, under a graph pattern evaluation semantics known as \textit{homomorphism semantics (hom)}, the set of evaluation results, $\Omega^{hom}_{p,g}$, is characterized as follows: 
\begin{equation*}
\Omega^{hom}_{p,g} = \Big\lbrace ~ \mu \in \lbrace \, \mathcal{V}_p \cup \mathcal{T}_p \to \mathcal{T}_g \, \rbrace \mid \forall t \in \mathcal{T}_p : \mu(t) = t \land
\mu(p) \subseteq g ~
\Big\rbrace\end{equation*}

Here, $\lbrace \, \mathcal{V}_p \cup \mathcal{T}_p \to \mathcal{T}_g \, \rbrace$ denotes the set of total functions with domain $\mathcal{V}_p \cup \mathcal{T}_p$ and co-domain $\mathcal{T}_g$.
Furthermore, $\mu(t)$ denotes the image of $t$ under $\mu$ and $\mu(p)$ denotes the result of applying the mapping function $\mu$ to every element of every triple pattern in the graph pattern $p$, i.e., each triple pattern $(x_i,x_j,x_k)$ is replaced by $(\mu(x_i), \mu(x_j), \mu(x_k))$. 
We call $\mu(p)$ the embedding corresponding to the pattern $p$ and the mapping $\mu$.
Homomorphism semantics requires for a mapping function $\mu$ in the context of a graph pattern $p$ and a graph $g$ that i) $\mu$ maps each term in the pattern to itself and that ii) the embedding $\mu(p)$ is a subgraph of $g$. 

A drawback of homomorphism semantics is that some topologically different graph structures cannot be distinguished. For example, the graph pattern $\lbrace (v_1,v_2,v_3) \rbrace$ matches the graph $g_1 = \lbrace (t_1, t_2, t_3) \rbrace$ (which does not contain a self-loop) via $\mu_1 = \lbrace (v_1, t_1), (v_2, t_2), (v_3, t_3)\rbrace$ and it also matches the graph $g_2 = \lbrace (t_1, t_2, t_1) \rbrace$ (which contains a self-loop) via $\mu_2 = \lbrace (v_1, t_1), (v_2, t_2), (v_3, t_1) \rbrace$. 

Furthermore, under homomorphism semantics it is possible that the size of the embedding is smaller than the size of the pattern. For example, let $g$ be $\lbrace (t_1,t_2,t_3) \rbrace$ and let $p$ be $\lbrace (v_1,t_2,v_2),$ $(v_3,t_2,v_2)\rbrace$. The pattern will match the graph with $\mu = \lbrace (v_1,t_1), (v_2,t_3), (v_3,t_1)\rbrace$. The embedding $\mu(p)$ is $\lbrace (t_1,t_2,t_3)\rbrace$. Although the pattern consists of two triple patterns, the embedding consists of one triple.

To ensure that the structure of each embedding corresponds to the structure of the graph pattern, we can use a specific isomorphism-based graph pattern evaluation semantics known as \textit{no-repeated-anything semantics (nra)}, which is defined as follows: 
\begin{equation*}
\Omega^{nra}_{p,g} = \Big\lbrace ~\mu \in \Omega^{hom}_{p,g} \mid \forall (x_i,t_i),(x_j,t_j) \in \mu : t_i = t_j \Rightarrow x_i = x_j ~\Big\rbrace\end{equation*}

In no-repeated-anything semantics, mapping functions need to be injections. Thus, i) no two variables can be bound to the same term, ii) a variable in a pattern cannot be bound to a term that occurs in the pattern, and iii) a pattern that consists of $n$ triple patterns can only match a subgraph that consists of $n$ triples. %

To better understand when it can be beneficial to use \textit{nra} semantics instead of \textit{hom} semantics we look at a task from the field of Sociology of Education. We might be interested to find out what difference it makes on the education level of a person whether the person resides in the birth country of the parents.
Figure \ref{fig:graphsandpatterns} shows five situations represented as graphs. The edge labels \textit{hF}, \textit{hM}, \textit{cob}, and \textit{cor} stand for \textit{has father}, \textit{has mother}, \textit{country of birth}, and \textit{country or residence}, respectively. For example, in the graph $g_1$ the person ($t_1$) has a father ($t_2$) that was born in country $t_4$,
has a mother ($t_3$) that was born in country $t_6$, and resides in country $t_5$. Each of the five graphs $g_1,\ldots, g_5$ is abstracted into a graph pattern by replacing terms in node positions by variables, resulting in the patterns $p_1,\ldots, p_5$.

Table \ref{tab:homnra} provides an overview of which of these patterns matches which of these graphs under \textit{hom} semantics (a) and under \textit{nra} semantics (b). For example, under \textit{hom} semantics $p_2$ matches $g_2$ and $g_5$,
under \textit{nra} semantics $p_5$ matches only $g_5$.
Under \textit{hom} semantics most of the patterns match more than one graph, thus, they cannot distinguish between different situations. For example, under \textit{hom} semantics the pattern $p_2$ matches the graph $g_2$ which expresses that the person resides in the father's but not in the mother's birth country, and it matches the graph $g_5$ which expresses that the person's country of residence coincides with both parents' country of birth. In contrast, under \textit{nra} semantics each pattern matches exactly one of the graphs. In case that these distinctions are actually relevant for accurately predicting the educational level obtained by a person, then \textit{hom} semantics would be insufficient, as it cannot make relevant distinctions -- metaphorically speaking, in case that in this situation these distinctions are relevant, then a model based on the distinctions that can be made under \textit{hom} semantics does not carve nature at the joints. Instead, a model based on the distinctions that can be made under \textit{nra} semantics does carve nature at the joints.
In the case that the patterns that can be distinguished under \textit{nra} semantics are more specific than necessary, then 
there exists not a single pattern from which a situation of interest can be predicted, but a set of patterns. Then, a model can form a disjunction of these patterns and predict the situation of interest if at least one of these patterns matches.

\begin{figure}
\begin{tikzpicture}[
    >=latex,
    node style/.style={circle, draw, minimum size=6pt, inner sep=0pt, fill=white},
    nodelabel/.style={font=\small, inner sep=1pt},
    edgelabel/.style={font=\small, inner sep=1pt, fill=white}
  ]
  \node[node style] (g1t1) at (1,0) {};
  \node[node style] (g1t2) at (0,1) {};
  \node[node style] (g1t3) at (2,1) {};
  \node[node style] (g1t4) at (0,2) {};
  \node[node style] (g1t5) at (1,2) {};
  \node[node style] (g1t6) at (2,2) {};
  \node[nodelabel, below=4pt] at (g1t1) {$t_1$};
  \node[nodelabel, left=4pt] at (g1t2) {$t_2$};
  \node[nodelabel, left=4pt] at (g1t3) {$t_3$};
  \node[nodelabel, above=4pt] at (g1t4) {$t_4$};
  \node[nodelabel, above=4pt] at (g1t5) {$t_5$};
  \node[nodelabel, above=4pt] at (g1t6) {$t_6$};
  \draw[->, thick] (g1t1) -- (g1t2);
  \draw[->, thick] (g1t1) -- (g1t3);
  \draw[->, thick] (g1t2) -- (g1t4);
  \draw[->, thick] (g1t3) -- (g1t6);  
  \draw[->, thick] (g1t1) -- (g1t5);  
  \node[edgelabel] at ($(g1t1)!0.5!(g1t2) + (-0.3,-0.2)$) {$hF$};
  \node[edgelabel] at ($(g1t1)!0.5!(g1t3) + (0.3,-0.2)$) {$hM$};
  \node[edgelabel] at ($(g1t2)!0.5!(g1t4) + (-0.3,0.2)$) {$cob$};
  \node[edgelabel] at ($(g1t3)!0.5!(g1t6) + (-0.3,0.2)$) {$cob$};
  \node[edgelabel] at ($(g1t1)!0.5!(g1t5) + (0,0)$) {$cor$};

  \node[node style] (g2t1) at (4,0) {};
  \node[node style] (g2t2) at (3,1) {};
  \node[node style] (g2t3) at (5,1) {};
  \node[node style] (g2t4) at (4,2) {};
  \node[node style] (g2t5) at (5,2) {};
  \node[nodelabel, below=4pt] at (g2t1) {$t_1$};
  \node[nodelabel, left=4pt] at (g2t2) {$t_2$};
  \node[nodelabel, left=4pt] at (g2t3) {$t_3$};
  \node[nodelabel, above=4pt] at (g2t4) {$t_4$};
  \node[nodelabel, above=4pt] at (g2t5) {$t_5$};
  \draw[->, thick] (g2t1) -- (g2t2);
  \draw[->, thick] (g2t1) -- (g2t3);
  \draw[->, thick] (g2t2) -- (g2t4);
  \draw[->, thick] (g2t1) -- (g2t4);
  \draw[->, thick] (g2t3) -- (g2t5); 
  \node[edgelabel] at ($(g2t1)!0.5!(g2t2) + (-0.3,-0.2)$) {$hF$};
  \node[edgelabel] at ($(g2t1)!0.5!(g2t3) + (0.3,-0.2)$) {$hM$};  
  \node[edgelabel] at ($(g2t2)!0.5!(g2t4) + (-0.2,0.2)$) {$cob$};
  \node[edgelabel] at ($(g2t3)!0.5!(g2t5) + (-0.3,0.2)$) {$cob$};
  \node[edgelabel] at ($(g2t1)!0.5!(g2t4) + (0,0)$) {$cor$};
    
  \node[node style] (g3t1) at (7,0) {};
  \node[node style] (g3t2) at (6,1) {};
  \node[node style] (g3t3) at (8,1) {};
  \node[node style] (g3t4) at (6,2) {};
  \node[node style] (g3t5) at (7,2) {};
  \node[nodelabel, below=4pt] at (g3t1) {$t_1$};
  \node[nodelabel, left=4pt] at (g3t2) {$t_2$};
  \node[nodelabel, left=4pt] at (g3t3) {$t_3$};
  \node[nodelabel, above=4pt] at (g3t4) {$t_4$};
  \node[nodelabel, above=4pt] at (g3t5) {$t_5$};
  \draw[->, thick] (g3t1) -- (g3t2);
  \draw[->, thick] (g3t1) -- (g3t3);
  \draw[->, thick] (g3t2) -- (g3t4);
  \draw[->, thick] (g3t1) -- (g3t5);
  \draw[->, thick] (g3t3) -- (g3t5);
  \node[edgelabel] at ($(g3t1)!0.5!(g3t2) + (-0.3,-0.2)$) {$hF$};
  \node[edgelabel] at ($(g3t1)!0.5!(g3t3) + (0.3,-0.2)$) {$hM$};
  \node[edgelabel] at ($(g3t2)!0.5!(g3t4) + (-0.3,0.2)$) {$cob$};
  \node[edgelabel] at ($(g3t3)!0.5!(g3t5) + (0.2,0.2)$) {$cob$};
  \node[edgelabel] at ($(g3t1)!0.5!(g3t5) + (0,0)$) {$cor$};
        
  \node[node style] (g4t1) at (10,0) {};
  \node[node style] (g4t2) at (9,1) {};
  \node[node style] (g4t3) at (11,1) {};
  \node[node style] (g4t4) at (10,1) {};
  \node[node style] (g4t5) at (10,2) {};
  \node[nodelabel, below=4pt] at (g4t1) {$t_1$};
  \node[nodelabel, left=4pt] at (g4t2) {$t_2$};
  \node[nodelabel, left=4pt] at (g4t3) {$t_3$};
  \node[nodelabel, above=4pt] at (g4t4) {$t_4$};
  \node[nodelabel, above=4pt] at (g4t5) {$t_5$};
  \draw[->, thick] (g4t1) -- (g4t2);
  \draw[->, thick] (g4t1) -- (g4t3);
  \draw[->, thick] (g4t1) -- (g4t4);
  \draw[->, thick] (g4t2) -- (g4t5);
  \draw[->, thick] (g4t3) -- (g4t5);
  \node[edgelabel] at ($(g4t1)!0.5!(g4t2) + (-0.3,-0.2)$) {$hF$};
  \node[edgelabel] at ($(g4t1)!0.5!(g4t3) + (0.3,-0.2)$) {$hM$};
  \node[edgelabel] at ($(g4t2)!0.5!(g4t5) + (-0.2,0.2)$) {$cob$};
  \node[edgelabel] at ($(g4t3)!0.5!(g4t5) + (0.2,0.2)$) {$cob$};
  \node[edgelabel] at ($(g4t1)!0.5!(g4t4) + (0,0)$) {$cor$};
          
  \node[node style] (g5t1) at (13,0) {};
  \node[node style] (g5t2) at (12,1) {};
  \node[node style] (g5t3) at (14,1) {};
  \node[node style] (g5t4) at (13,2) {};
  \node[nodelabel, below=4pt] at (g5t1) {$t_1$};
  \node[nodelabel, left=4pt] at (g5t2) {$t_2$};
  \node[nodelabel, left=4pt] at (g5t3) {$t_3$};
  \node[nodelabel, above=4pt] at (g5t4) {$t_4$};
  \draw[->, thick] (g5t1) -- (g5t2);
  \draw[->, thick] (g5t1) -- (g5t3);
  \draw[->, thick] (g5t1) -- (g5t4);
  \draw[->, thick] (g5t2) -- (g5t4);
  \draw[->, thick] (g5t3) -- (g5t4);
  \node[edgelabel] at ($(g5t1)!0.5!(g5t2) + (-0.3,-0.2)$) {$hF$};
  \node[edgelabel] at ($(g5t1)!0.5!(g5t3) + (0.3,-0.2)$) {$hM$};
  \node[edgelabel] at ($(g5t2)!0.5!(g5t4) + (-0.2,0.2)$) {$cob$};
  \node[edgelabel] at ($(g5t3)!0.5!(g5t4) + (0.2,0.2)$) {$cob$};
  \node[edgelabel] at ($(g5t1)!0.5!(g5t4) + (0,0)$) {$cor$};   

  \node[nodelabel, below=18pt] at (g1t1) {$g_1$};
  \node[nodelabel, below=18pt] at (g2t1) {$g_2$};
  \node[nodelabel, below=18pt] at (g3t1) {$g_3$};
  \node[nodelabel, below=18pt] at (g4t1) {$g_4$};
  \node[nodelabel, below=18pt] at (g5t1) {$g_5$};

  \node[node style] (p1t1) at (1,-3.7) {};
  \node[node style] (p1t2) at (0,-2.7) {};
  \node[node style] (p1t3) at (2,-2.7) {};
  \node[node style] (p1t4) at (0,-1.7) {};
  \node[node style] (p1t5) at (1,-1.7) {};
  \node[node style] (p1t6) at (2,-1.7) {};
  \node[nodelabel, below=4pt] at (p1t1) {$v_1$};
  \node[nodelabel, left=4pt] at (p1t2) {$v_2$};
  \node[nodelabel, left=4pt] at (p1t3) {$v_3$};
  \node[nodelabel, above=4pt] at (p1t4) {$v_4$};
  \node[nodelabel, above=4pt] at (p1t5) {$v_5$};
  \node[nodelabel, above=4pt] at (p1t6) {$v_6$};
  \draw[->, thick] (p1t1) -- (p1t2);
  \draw[->, thick] (p1t1) -- (p1t3);
  \draw[->, thick] (p1t2) -- (p1t4);
  \draw[->, thick] (p1t3) -- (p1t6);  
  \draw[->, thick] (p1t1) -- (p1t5);  
  \node[edgelabel] at ($(p1t1)!0.5!(p1t2) + (-0.3,-0.2)$) {$hF$};
  \node[edgelabel] at ($(p1t1)!0.5!(p1t3) + (0.3,-0.2)$) {$hM$};
  \node[edgelabel] at ($(p1t2)!0.5!(p1t4) + (-0.3,0.2)$) {$cob$};
  \node[edgelabel] at ($(p1t3)!0.5!(p1t6) + (-0.3,0.2)$) {$cob$};
  \node[edgelabel] at ($(p1t1)!0.5!(p1t5) + (0,0)$) {$cor$};  

  \node[node style] (p2t1) at (4,-3.7) {};
  \node[node style] (p2t2) at (3,-2.7) {};
  \node[node style] (p2t3) at (5,-2.7) {};
  \node[node style] (p2t4) at (4,-1.7) {};
  \node[node style] (p2t5) at (5,-1.7) {};
  \node[nodelabel, below=4pt] at (p2t1) {$v_1$};
  \node[nodelabel, left=4pt] at (p2t2) {$v_2$};
  \node[nodelabel, left=4pt] at (p2t3) {$v_3$};
  \node[nodelabel, above=4pt] at (p2t4) {$v_4$};
  \node[nodelabel, above=4pt] at (p2t5) {$v_5$};
  \draw[->, thick] (p2t1) -- (p2t2);
  \draw[->, thick] (p2t1) -- (p2t3);
  \draw[->, thick] (p2t2) -- (p2t4);
  \draw[->, thick] (p2t1) -- (p2t4);
  \draw[->, thick] (p2t3) -- (p2t5); 
  \node[edgelabel] at ($(p2t1)!0.5!(p2t2) + (-0.3,-0.2)$) {$hF$};
  \node[edgelabel] at ($(p2t1)!0.5!(p2t3) + (0.3,-0.2)$) {$hM$};  
  \node[edgelabel] at ($(p2t2)!0.5!(p2t4) + (-0.2,0.2)$) {$cob$};
  \node[edgelabel] at ($(p2t3)!0.5!(p2t5) + (-0.3,0.2)$) {$cob$};
  \node[edgelabel] at ($(p2t1)!0.5!(p2t4) + (0,0)$) {$cor$};

  \node[node style] (p3t1) at (7,-3.7) {};
  \node[node style] (p3t2) at (6,-2.7) {};
  \node[node style] (p3t3) at (8,-2.7) {};
  \node[node style] (p3t4) at (6,-1.7) {};
  \node[node style] (p3t5) at (7,-1.7) {};
  \node[nodelabel, below=4pt] at (p3t1) {$v_1$};
  \node[nodelabel, left=4pt] at (p3t2) {$v_2$};
  \node[nodelabel, left=4pt] at (p3t3) {$v_3$};
  \node[nodelabel, above=4pt] at (p3t4) {$v_4$};
  \node[nodelabel, above=4pt] at (p3t5) {$v_5$};
  \draw[->, thick] (p3t1) -- (p3t2);
  \draw[->, thick] (p3t1) -- (p3t3);
  \draw[->, thick] (p3t2) -- (p3t4);
  \draw[->, thick] (p3t1) -- (p3t5);
  \draw[->, thick] (p3t3) -- (p3t5);
  \node[edgelabel] at ($(p3t1)!0.5!(p3t2) + (-0.3,-0.2)$) {$hF$};
  \node[edgelabel] at ($(p3t1)!0.5!(p3t3) + (0.3,-0.2)$) {$hM$};
  \node[edgelabel] at ($(p3t2)!0.5!(p3t4) + (-0.3,0.2)$) {$cob$};
  \node[edgelabel] at ($(p3t3)!0.5!(p3t5) + (0.2,0.2)$) {$cob$};
  \node[edgelabel] at ($(p3t1)!0.5!(p3t5) + (0,0)$) {$cor$};

  \node[node style] (p4t1) at (10,-3.7) {};
  \node[node style] (p4t2) at (9,-2.7) {};
  \node[node style] (p4t3) at (11,-2.7) {};
  \node[node style] (p4t4) at (10,-2.7) {};
  \node[node style] (p4t5) at (10,-1.7) {};
  \node[nodelabel, below=4pt] at (p4t1) {$v_1$};
  \node[nodelabel, left=4pt] at (p4t2) {$v_2$};
  \node[nodelabel, left=4pt] at (p4t3) {$v_3$};
  \node[nodelabel, above=4pt] at (p4t4) {$v_4$};
  \node[nodelabel, above=4pt] at (p4t5) {$v_5$};
  \draw[->, thick] (p4t1) -- (p4t2);
  \draw[->, thick] (p4t1) -- (p4t3);
  \draw[->, thick] (p4t1) -- (p4t4);
  \draw[->, thick] (p4t2) -- (p4t5);
  \draw[->, thick] (p4t3) -- (p4t5);
  \node[edgelabel] at ($(p4t1)!0.5!(p4t2) + (-0.3,-0.2)$) {$hF$};
  \node[edgelabel] at ($(p4t1)!0.5!(p4t3) + (0.3,-0.2)$) {$hM$};
  \node[edgelabel] at ($(p4t2)!0.5!(p4t5) + (-0.2,0.2)$) {$cob$};
  \node[edgelabel] at ($(p4t3)!0.5!(p4t5) + (0.2,0.2)$) {$cob$};
  \node[edgelabel] at ($(p4t1)!0.5!(p4t4) + (0,0)$) {$cor$};
  
  \node[node style] (p5t1) at (13,-3.7) {};
  \node[node style] (p5t2) at (12,-2.7) {};
  \node[node style] (p5t3) at (14,-2.7) {};
  \node[node style] (p5t4) at (13,-1.7) {};
  \node[nodelabel, below=4pt] at (p5t1) {$v_1$};
  \node[nodelabel, left=4pt] at (p5t2) {$v_2$};
  \node[nodelabel, left=4pt] at (p5t3) {$v_3$};
  \node[nodelabel, above=4pt] at (p5t4) {$v_4$};
  \draw[->, thick] (p5t1) -- (p5t2);
  \draw[->, thick] (p5t1) -- (p5t3);
  \draw[->, thick] (p5t1) -- (p5t4);
  \draw[->, thick] (p5t2) -- (p5t4);
  \draw[->, thick] (p5t3) -- (p5t4);
  \node[edgelabel] at ($(p5t1)!0.5!(p5t2) + (-0.3,-0.2)$) {$hF$};
  \node[edgelabel] at ($(p5t1)!0.5!(p5t3) + (0.3,-0.2)$) {$hM$};
  \node[edgelabel] at ($(p5t2)!0.5!(p5t4) + (-0.2,0.2)$) {$cob$};
  \node[edgelabel] at ($(p5t3)!0.5!(p5t4) + (0.2,0.2)$) {$cob$};
  \node[edgelabel] at ($(p5t1)!0.5!(p5t4) + (0,0)$) {$cor$};   

  \node[nodelabel, below=18pt] at (p1t1) {$p_1$};
  \node[nodelabel, below=18pt] at (p2t1) {$p_2$};
  \node[nodelabel, below=18pt] at (p3t1) {$p_3$};
  \node[nodelabel, below=18pt] at (p4t1) {$p_4$};
  \node[nodelabel, below=18pt] at (p5t1) {$p_5$};
\end{tikzpicture}
\caption{The graphs $g_1,\ldots,g_5$ show different situations of a person ($t_1$), their country of residence, and their father's and mother's country of birth. The edge labels \textit{hF}, \textit{hM}, \textit{cob}, and \textit{cor} stand for \textit{has father}, \textit{has mother}, \textit{country of birth}, and \textit{country or residence}, respectively. Each of these graphs is abstracted into a graph pattern by replacing the node terms by variables, resulting in the patterns $p_1,\ldots,p_5$.}
\label{fig:graphsandpatterns}
\end{figure}
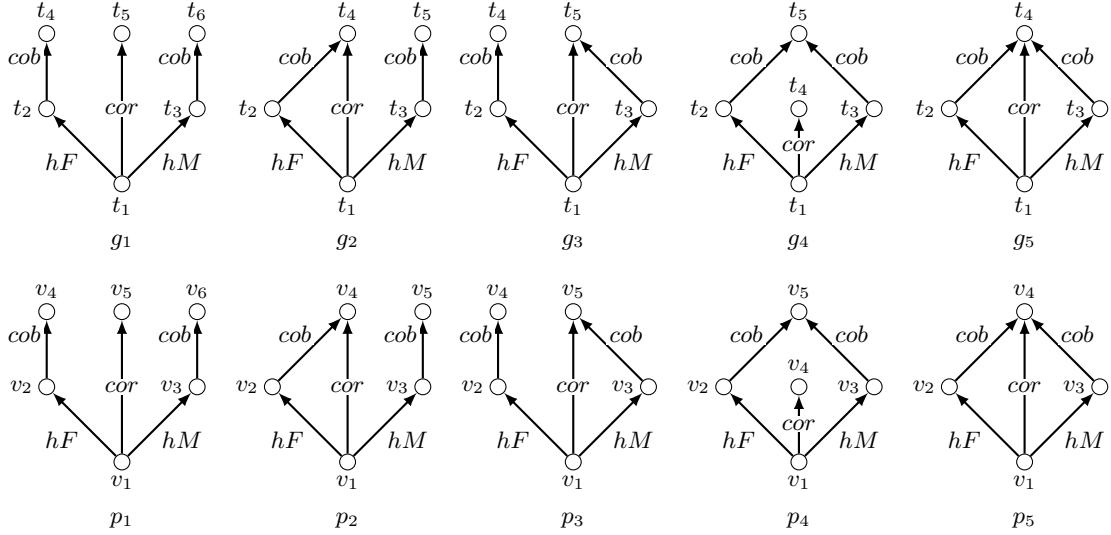

{
\setlength{\tabcolsep}{1pt}

\begin{table}
\caption{Overview of which of the patterns $p_1,\ldots,p_5$ matches which of the graphs $g_1,\ldots,g_5$ under \textit{hom} semantics (a) and under \textit{nra} semantics (b).}
\centering
\begin{subtable}{0.5\textwidth}
\centering
\captionsetup{justification=centering}
\caption{Matches under \textit{hom} semantics}
\begin{tabular}{cccccc}
 & \pgone $g_1$ & \pgtwo $g_2$ & \pgthree $g_3$ & \pgfour $g_4$ & \pgfive $g_5$\\
\pgone $p_1$ 	& \cmark & \cmark & \cmark & \cmark & \cmark \\
\pgtwo $p_2$ 	& \xmark & \cmark & \xmark & \xmark & \cmark \\
\pgthree $p_3$ 	& \xmark & \xmark & \cmark & \xmark & \cmark \\
\pgfour $p_4$ 	& \xmark & \xmark & \xmark & \cmark & \cmark \\
\pgfive $p_5$ 	& \xmark & \xmark & \xmark & \xmark & \cmark \\
\end{tabular}
\end{subtable}\hfill
\begin{subtable}{0.5\textwidth}
\centering
\captionsetup{justification=centering}
\caption{Matches under \textit{nra} semantics}
\begin{tabular}{cccccc}
 & \pgone $g_1$ & \pgtwo $g_2$ & \pgthree $g_3$ & \pgfour $g_4$ & \pgfive $g_5$\\
\pgone $p_1$ 	& \cmark & \xmark & \xmark & \xmark & \xmark \\
\pgtwo $p_2$ 	& \xmark & \cmark & \xmark & \xmark & \xmark \\
\pgthree $p_3$ 	& \xmark & \xmark & \cmark & \xmark & \xmark \\
\pgfour $p_4$ 	& \xmark & \xmark & \xmark & \cmark & \xmark \\
\pgfive $p_5$ 	& \xmark & \xmark & \xmark & \xmark & \cmark \\
\end{tabular}
\end{subtable}
\label{tab:homnra}
\end{table}
}
\newpage

In no-repeated-anything semantics, the shape of a pattern $p$ is identical with the shape of each of its embeddings $\mu(p)$. Formally, given a pattern $p$ and one of its embeddings $\mu(p)$:
\begin{align*}
\forall t \in \mathcal{T}_p : &\\
 & |\lbrace \, (s,p,o) \in p \mid s=t \,\rbrace| = |\lbrace \, (s,p,o) \in \mu(p) \mid s=t\,\rbrace|\\
 & ~~~~\text{(i.e., same out-degree for term nodes in pattern and embedding)}\\
 & \land |\lbrace \, (s,p,o) \in p \mid o=t \,\rbrace| = |\lbrace \, (s,p,o) \in \mu(p) \mid o=t\,\rbrace|\\
 & ~~~~\text{(i.e., same in-degree for term nodes in pattern and embedding)}\\
\land \, \forall v \in \mathcal{V}_p : &\\ 
 & |\lbrace \, (s,p,o) \in p \mid s=v \,\rbrace| = |\lbrace \, (s,p,o) \in \mu(p) \mid s=\mu(v)\,\rbrace|\\
 & ~~~~\text{(i.e., same out-deg. for var. nodes in pattern and corresp. term node in emb.)}\\
 & \land |\lbrace \, (s,p,o) \in p \mid o=v \,\rbrace| = |\lbrace \, (s,p,o) \in \mu(p) \mid o=\mu(v)\,\rbrace|\\
 & ~~~~\text{(i.e., same in-deg. for var. nodes in pattern and corresp. term node in emb.)}
\end{align*}

Thus, if a term occurs in subject position of $n$ triple patterns in the pattern (i.e., the out-degree of that node in the pattern is $n$), then the term occurs in subject position of $n$ triples in the embedding (i.e., the out-degree of that node in the embedding is $n$). The same holds for the in-degree. Furthermore,
if a variable $v$ occurs in subject position of $n$ triple patterns in the pattern (i.e., the out-degree of that node in the pattern is $n$), then the corresponding term $\mu(v)$ occurs in subject position of $n$ triples in the embedding (i.e., the out-degree of that node in the embedding is $n$). The same holds for the in-degree.\\

Note that for practical purposes it is sufficient for a mapping function to only map variables to terms.
Thus, it is sufficient that $\mu \in \lbrace \mathcal{V}_{p} \mapsto \mathcal{T}_g \rbrace$ instead of that $\mu \in \lbrace \mathcal{V}_p \cup \mathcal{T}_p \to \mathcal{T}_g \rbrace$. Because each mapping function $\mu \in \Omega^{nra}_{p,g}$ maps each term $t \in \mathcal{T}_p$ to itself, these term-to-term mappings would otherwise be stored redundantly for each mapping function, consuming memory space. Furthermore, these term-to-term mappings can be derived from $\mathcal{T}_p$. Therefore, in practice for a mapping function only the variable-to-term mappings are stored.

\section{Graph Pattern-based Association Rules}\label{sec:gpar}

\subsection{Preliminary Considerations}

In the ISAR context, an itemset is either subset of another itemset or it is not. When we move on to rules based on graph patterns, we face the situation that a graph pattern can match a graph multiple times. We could ignore the actual number of matches and only care about whether there is a match or whether there is no match. Thus, the condition we would be dealing with is $\Omega^{nra}_{p,g} \neq \varnothing$. Then, we could interpret a rule of the form $p_1 \Rightarrow p_2$, given a graph $g$, as follows:
\[ \Omega^{nra}_{p_1,g} \neq \varnothing \vdash_p \Omega^{nra}_{p_2,g} \neq \varnothing\]
Thus, under the condition that the graph $g$ is matched by the pattern $p_1$, we probabilistically entail that the graph $g$ is also matched by the pattern $p_2$. 
However, this concept has a limitation. 
We'd miss the opportunity to specify constraints on how a match of $p_2$ needs to correspond to a match of $p_1$. For example, we'd like to express that for each match $\mu_1$ of $p_1$ on $g$ (i.e., $\mu_1 \in \Omega^{nra}_{p_1,g}$) there exists a match $\mu_2$ of $p_2$ on $g$ (i.e., $\mu_2 \in \Omega^{nra}_{p_2,g}$) in the \textit{same location} (which means that the embeddings $\mu_1(p_1)$ and $\mu_2(p_2)$ overlap).%

To transcend this limitation we can define constraints such as that in order for a pair of matches of $p_1$ and $p_2$ to correspond it is necessary that the term bound to a given variable in the match $\mu_1$ of $p_1$ is equal to the term bound to a given variable in the match $\mu_2$ of $p_2$. Thereby, we go beyond expressing mere correlation between patterns towards describing logical implications.

\subsection{Formalization and Interpretation}

Let $p_1$ and $p_2$ be graph patterns 
with $\mathcal{V}_{p_1} \neq \varnothing$ and $\mathcal{V}_{p_2} \neq \varnothing$ and let 
$V_1 = (v_{1,1}, \ldots, v_{1,n})$ and $V_2 = (v_{2,1}, \ldots, v_{2,n})$ be non-empty non-repetitive sequences of equal length of variables from $\mathcal{V}_{p_1}$ and $\mathcal{V}_{p_2}$, respectively. We refer to members of these two list as joining variables, and we call $(p_1, p_2, V_1, V_2)$ a graph pattern-based association rule. 
We formalize the meaning of a graph pattern-based association rule $(p_1, p_2, V_1, V_2)$ in the context of some graph $g$ as follows:
\begin{equation*}
\Omega^{nra}_{p_1,g} \neq \varnothing  \vdash_p \forall \mu_1 \in \Omega^{nra}_{p_1,g} : \exists \mu_2 \in \Omega^{nra}_{p_2,g} : \mu_1(V_1) = \mu_2(V_2)
\end{equation*}
By a slight misuse of notation, $\mu(V)$ denotes the result or replacing each variable $v$ in the sequence $V$ with $\mu(v)$, resulting in a non-repetitive sequence of terms. 
The meaning of a rule is that if the antecedent pattern $p_1$ matches the graph $g$ (i.e., $\Omega^{nra}_{p_1,g} \neq \varnothing$), then we probabilistically entail that for every match $\mu_1$ of the antecedent graph pattern $p_1$ on $g$ (i.e., $\forall \mu_1 \in \Omega^{nra}_{p_1,g}$) there exists a match $\mu_2$ of the consequent graph pattern $p_2$ on $g$ such that the term sequence $\mu_1(V_1)$ is equal to the term sequence $\mu_2(V_2)$, which means that in both mappings corresponding variables bind to the same term. Thus, with $\mu_1(V_1) = \mu_2(V_2)$ we express the correspondence between the two matches $\mu_1$ and $\mu_2$ of $p_1$ and $p_2$, respectively.

It is important to note that the rule does not express that the existence of some $\mu_1 \in \Omega^{nra}_{p_1,g}$ \textit{causes} some $\mu_2 \in \Omega^{nra}_{p_2,g}$ with $\mu_1(V_1) = \mu_2(V_2)$ to be the case.

\subsection{Examples}
Figures \ref{fig:Ex1}--\ref{fig:Ex3} show examples of graph pattern-based association rules. 
In the rule shown in Figure \ref{fig:Ex1}, the variables $v_2$ and $v_3$ are shared between the patterns $p_1$ and $p_2$. The graph $g$ describes that Alice and Bob are co-authors and that Alice works at Org. The rule expresses that if two entities are co-author of each other and one of them has a \texttt{worksAt} relation to something, then the other entity also has a \texttt{worksAt} relation to that thing.  The pattern $p_1$ matches the graph $g$ exactly once, resulting in the mapping $\mu$. The rule can be applied to extend the graph with an edge that expresses that Bob works at Org. This edge is obtained by the applying the mapping $\mu$ on the pattern $p_1$ (i.e., $\mu(p_2)$). How that works is explained in Section \ref{sec:applications}.

In the rule shown in Figure \ref{fig:Ex2}, the variable $v_2$ is shared between the patterns $p_1$ and $p_2$. The graph $g$ describes that Alice knows Bob and Bob knows Alice. The rule expresses that if two entities are mutually connected via a relation, then that relation is probably symmetric. The pattern $p_1$ matches the graph $g$ exactly once, resulting in the mapping $\mu$. The rule can be applied to extend the graph with an edge that expresses that the knows relation is probably symmetric. Note that in contrast to the previous example, we here make use of a variable in predicate position. Furthermore, note that the term \texttt{knows} occurs both as a node label and as an edge label.

Furthermore, note that is important to score rules based on available evidence of their correctness and then decide per use case about the minimum level of evidence required when deciding whether to apply or discard a rule. Therefore, we introduce metrics for graph pattern-based association rules in Section \ref{sec:metrics}. Finally, applying rules learned from data realizes inductive reasoning, but it can be the case a rule that is learned from data or that is manually created is actually a deduction rule. Thus, graph pattern-based association rules can express deduction rules and can be applied for the task of logical deduction.

In the rule shown in Figure \ref{fig:Ex3}, the variable $v_1$ is shared between the patterns $p_1$ and $p_2$. The graph $g$ expresses that $m_1$ is both a molecule and a Carboxylic acid. The pattern $p_1$ expresses that something is of type Carboxylic acid.
The pattern $p_2$ describes how a Carboxylic acid is defined: a node ($v_1$) representing a molecule has \texttt{hasAtom} relations to four nodes ($v_2$--$v_5$), where $v_2$ and $v_4$ are Oxygen atoms, $v_3$ is a Carbon atom, and $v_5$ is a Hydrogen atom. $v_2$ and $v_3$ are connected via double bonds, $v_3$ and $v_4$ are connected via single bonds, and  $v_4$ and $v_5$ are connected via single bonds. Note that in contrast to the previous examples, the consequent of the rule is more complex, consisting of multiple triple patterns, whereas in the previous two examples the consequent consists only of a single triple pattern. In contrast, the consequent of itemset-based association rules usually contain exactly one item. Moreover, in those example, it was the case that $\mathcal{V}_{p_2} \subseteq \mathcal{V}_{p_1}$. Whereas the obvious applications of the rules in Figure \ref{fig:Ex1} and Figure \ref{fig:Ex2} are link prediction, the application of the rule in Figure \ref{fig:Ex3} could be plausibility analysis: any molecule that is of type Carboxylic acid should be matched by the consequent pattern or it would be regarded as anomalous if that pattern does not match. Here the usefulness of complex consequent becomes more obvious.

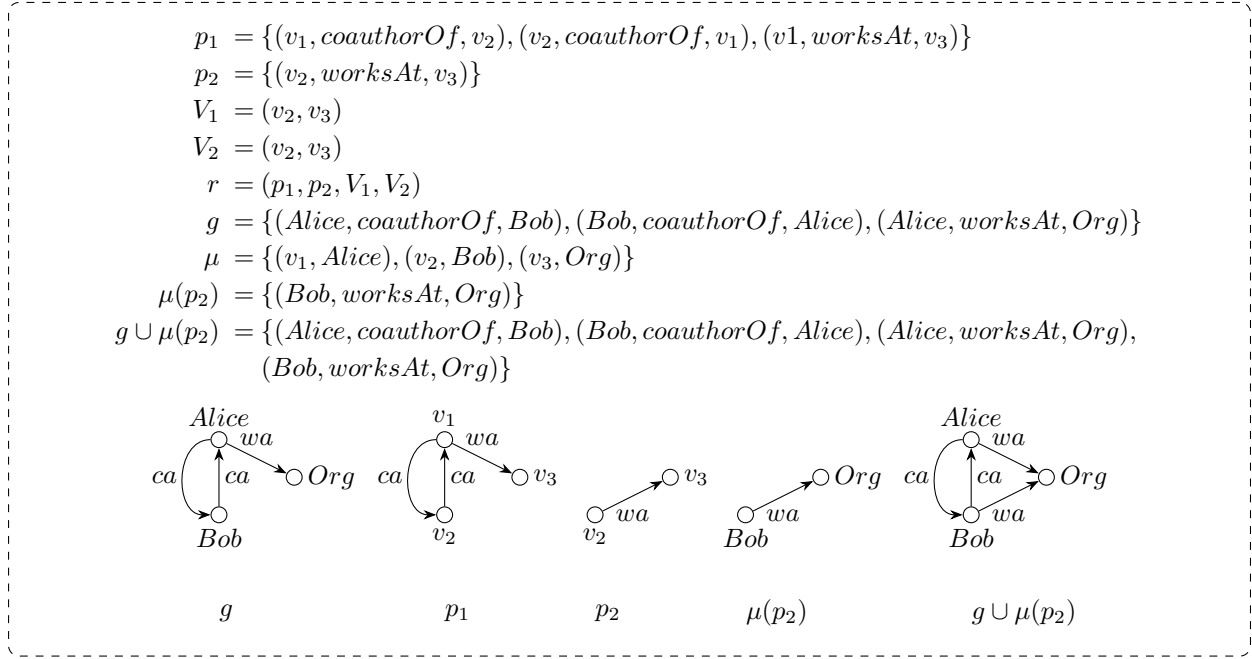
\begin{figure}
\begin{tikzpicture}
\node (content) [
    inner sep=0pt
] {
\begin{tabular}[t]{@{}c@{\hspace{0cm}}c@{}}
$\begin{aligned}
p_1 ~=~ & \lbrace (v_1, coauthorOf, v_2), (v_2, coauthorOf, v_1), (v1,worksAt,v_3) \rbrace\\

p_2 ~=~ & \lbrace (v_2, worksAt, v_3) \rbrace\\
V_1 ~=~ & (v_2,v_3)\\
V_2 ~=~ & (v_2,v_3)\\
r ~=~ & (p_1,p_2, V_1,V_2)\\
g ~=~ & \lbrace (Alice, coauthorOf, Bob), (Bob, coauthorOf, Alice), (Alice, worksAt, Org)\rbrace\\
\mu ~=~ & \lbrace (v_1, Alice), (v_2, Bob), (v_3, Org)\rbrace\\

\mu(p_2) ~=~ & \lbrace (Bob, worksAt, Org) \rbrace \\
g \cup \mu(p_2) ~=~ & \lbrace (Alice, coauthorOf, Bob), (Bob, coauthorOf, Alice), (Alice, worksAt, Org), \\
& (Bob, worksAt, Org) \rbrace\\
\end{aligned}$
&

\\[2.5cm]

\begin{tikzpicture}[
    node style/.style={circle, draw, minimum size=6pt, inner sep=0pt, fill=white},
    every path/.style={->, >=Stealth}
]

\node[node style] (g1v1) at (-0.5, 1) {};
\node[above=5pt] at (g1v1) {$Alice$};
\node[node style] (g1v2) at (-0.5, 0) {};
\node[below=5pt] at (g1v2) {$Bob$};
\draw[->] (g1v1) .. controls +(-0.6,0) and +(-0.6,0) .. (g1v2) node[midway, left=2pt] {$ca$};
\draw (g1v2) -- (g1v1)  node[midway, right=2pt ] {$ca$};
\node[node style] (g1v3) at (0.5, 0.5) {};
\node[right=5pt] at (g1v3) {$Org$};
\draw (g1v1) -- (g1v3)  node[midway, above=5pt ] {$wa$};

\node[node style] (p1v1) at (2.5, 1) {};
\node[above=5pt] at (p1v1) {$v_1$};
\node[node style] (p1v2) at (2.5, 0) {};
\node[below=5pt] at (p1v2) {$v_2$};
\draw[->] (p1v1) .. controls +(-0.6,0) and +(-0.6,0) .. (p1v2) node[midway, left=2pt] {$ca$};
\draw (p1v2) -- (p1v1)  node[midway, right=2pt ] {$ca$};
\node[node style] (p1v3) at (3.5, 0.5) {};
\node[right=5pt] at (p1v3) {$v_3$};
\draw (p1v1) -- (p1v3)  node[midway, above=5pt ] {$wa$};

\node[node style] (p2v2) at (4.5, 0) {};
\node[below=5pt] at (p2v2) {$v_2$};
\node[node style] (p2v3) at (5.5, 0.5) {};
\node[right=5pt] at (p2v3) {$v_3$};
\draw (p2v2) -- (p2v3)  node[midway, below=5pt ] {$wa$};

\node[node style] (p3v2) at (6.5, 0) {};
\node[below=5pt] at (p3v2) {$Bob$};
\node[node style] (p3v3) at (7.5, 0.5) {};
\node[right=5pt] at (p3v3) {$Org$};
\draw (p3v2) -- (p3v3)  node[midway, below=5pt ] {$wa$};

\node[node style] (g2v1) at (9.5, 1) {};
\node[above=5pt] at (g2v1) {$Alice$};
\node[node style] (g2v2) at (9.5, 0) {};
\node[below=5pt] at (g2v2) {$Bob$};
\draw[->] (g2v1) .. controls +(-0.6,0) and +(-0.6,0) .. (g2v2) node[midway, left=2pt] {$ca$};
\draw (g2v2) -- (g2v1)  node[midway, right=2pt ] {$ca$};
\node[node style] (g2v3) at (10.5, 0.5) {};
\node[right=5pt] at (g2v3) {$Org$};
\draw (g2v1) -- (g2v3)  node[midway, above=5pt ] {$wa$};
\draw (g2v2) -- (g2v3)  node[midway, below=5pt ] {$wa$};

\node[anchor=west] at (-0.5, -1.3) {$g$};
\node[anchor=west] at (2.5, -1.3) {$p_1$};
\node[anchor=west] at (4.5, -1.3) {$p_2$};
\node[anchor=west] at (6.5, -1.3) {$\mu(p_2)$};
\node[anchor=west] at (9.5, -1.3) {$g \cup \mu(p_2)$};
\end{tikzpicture}
&
\end{tabular}
};

\node[fit=(content),
minimum width=\textwidth, 
draw=black, dashed, inner sep=8pt, rounded corners=4pt] {};
\end{tikzpicture}
\caption{Example of a graph pattern-based association rule.  In the visualizations, the term \texttt{coauthorOf} is abbreviated to \texttt{ca} and the term \texttt{worksAt} is abbreviated to \texttt{wa}.}
\label{fig:Ex1}
\end{figure}

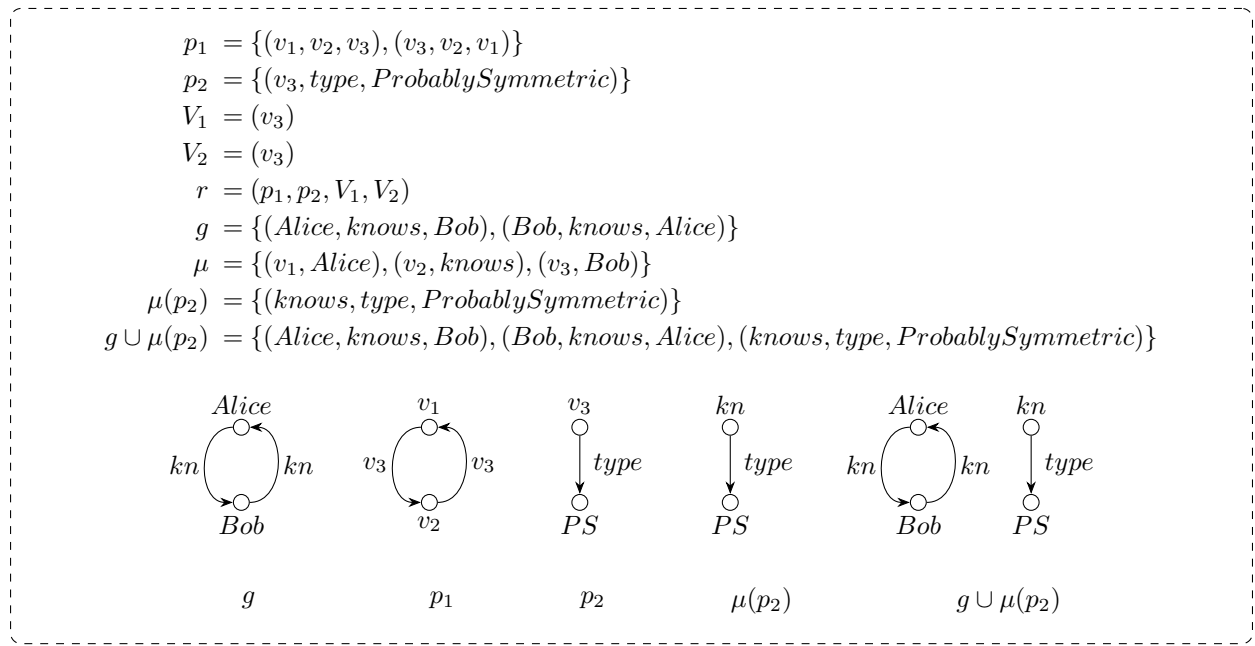
\begin{figure}
\begin{tikzpicture}
\node (content) [
    inner sep=0pt
] {
\begin{tabular}[t]{@{}c@{\hspace{0cm}}c@{}}
$\begin{aligned}
p_1 ~=~ & \lbrace (v_1, v_2, v_3), (v_3, v_2, v_1) \rbrace\\

p_2 ~=~ & \lbrace (v_3, type, ProbablySymmetric) \rbrace\\
V_1 ~=~ & (v_3)\\
V_2 ~=~ & (v_3)\\
r ~=~ & (p_1,p_2,V_1,V_2)\\
g ~=~ & \lbrace (Alice, knows, Bob), (Bob, knows, Alice)\rbrace\\
\mu ~=~ & \lbrace (v_1, Alice), (v_2, knows), (v_3, Bob)\rbrace\\

\mu(p_2) ~=~ & \lbrace (knows, type, ProbablySymmetric) \rbrace \\
g \cup \mu(p_2) ~=~ & \lbrace (Alice, knows, Bob), (Bob, knows, Alice), (knows, type, ProbablySymmetric) \rbrace\\
\end{aligned}$
&

\\[2.5cm]

\begin{tikzpicture}[
    node style/.style={circle, draw, minimum size=6pt, inner sep=0pt, fill=white},
    every path/.style={->, >=Stealth}
]

\node[node style] (g1v1) at (-0.5, 1) {};
\node[above=5pt] at (g1v1) {$Alice$};
\node[node style] (g1v2) at (-0.5, 0) {};
\node[below=5pt] at (g1v2) {$Bob$};
\draw[->] (g1v1) .. controls +(-0.6,0) and +(-0.6,0) .. (g1v2) node[midway, left=2pt] {$kn$};
\draw[->] (g1v2) .. controls +(0.6,0) and +(0.6,0) .. (g1v1) node[midway, right=2pt] {$kn$};

\node[node style] (p1v1) at (2, 1) {};
\node[above=5pt] at (p1v1) {$v_1$};
\node[node style] (p1v2) at (2, 0) {};
\node[below=5pt] at (p1v2) {$v_2$};
\draw[->] (p1v1) .. controls +(-0.6,0) and +(-0.6,0) .. (p1v2) node[midway, left=2pt] {$v_3$};
\draw[->] (p1v2) .. controls +(0.6,0) and +(0.6,0) .. (p1v1) node[midway, right=2pt] {$v_3$};

\node[node style] (p2v2) at (4, 1) {};
\node[above=5pt] at (p2v2) {$v_3$};
\node[node style] (p2v3) at (4, 0) {};
\node[below=5pt] at (p2v3) {$PS$};
\draw (p2v2) -- (p2v3)  node[midway, right=5pt ] {$type$};

\node[node style] (p2v2) at (6, 1) {};
\node[above=5pt] at (p2v2) {$kn$};
\node[node style] (p2v3) at (6, 0) {};
\node[below=5pt] at (p2v3) {$PS$};
\draw (p2v2) -- (p2v3)  node[midway, right=5pt ] {$type$};

\node[node style] (g2v1) at (8.5, 1) {};
\node[above=5pt] at (g2v1) {$Alice$};
\node[node style] (g2v2) at (8.5, 0) {};
\node[below=5pt] at (g2v2) {$Bob$};
\draw[->] (g2v1) .. controls +(-0.6,0) and +(-0.6,0) .. (g2v2) node[midway, left=2pt] {$kn$};
\draw[->] (g2v2) .. controls +(0.6,0) and +(0.6,0) .. (g2v1) node[midway, right=2pt] {$kn$};
\node[node style] (g2v2) at (10, 1) {};
\node[above=5pt] at (g2v2) {$kn$};
\node[node style] (g2v3) at (10, 0) {};
\node[below=5pt] at (g2v3) {$PS$};
\draw (g2v2) -- (g2v3)  node[midway, right=5pt ] {$type$};

\node[anchor=west] at (-0.5, -1.3) {$g$};
\node[anchor=west] at (2, -1.3) {$p_1$};
\node[anchor=west] at (4, -1.3) {$p_2$};
\node[anchor=west] at (6, -1.3) {$\mu(p_2)$};
\node[anchor=west] at (9, -1.3) {$g \cup \mu(p_2)$};
\end{tikzpicture}
&
\end{tabular}
};

\node[fit=(content),
minimum width=\textwidth, 
draw=black, dashed, inner sep=8pt, rounded corners=4pt] {};
\end{tikzpicture}
\caption{Example of a graph pattern-based association rule. In the visualizations, the term \texttt{knows} is abbreviated to \texttt{kn}, and the term \texttt{ProbablySymmetric} is abbreviated to \texttt{PS}.}
\label{fig:Ex2}
\end{figure}

\begin{figure}
\begin{tikzpicture}
\node (content) [
    inner sep=0pt
] {
\begin{tabular}[t]{@{}c@{\hspace{0cm}}c@{}}
$\begin{aligned}
p_1 ~=~ & \lbrace (v_1, type, CarboxylicAcid) \rbrace\\
p_2 ~=~ & \lbrace (v_1, hasAtom, v_2), (v_1, hasAtom, v_3), (v_1, hasAtom, v_4), (v_1, hasAtom, v_5), \\
& (v_2, doubleBond, v_3), (v_3, doubleBond, v_2), (v_3, singleBond, v_4), (v_4, singleBond, v_3), \\
& (v_4, singleBond, v_5), (v_5, singleBond, v_4), (v_2, element, O), (v_3, element, Carbon),\\
& (v_4, element, Oxygen), (v_5, element, Hydrogen)\rbrace\\
V_1 ~=~ & (v_1)\\
V_2 ~=~ & (v_1)\\
r ~=~ & (p_1,p_2,V_1,V_2)\\
g ~=~ & \lbrace (m_1, type, Molecule), (m_1, type, CarboxylicAcid)\rbrace
\end{aligned}$
&

\\[2.5cm]

\begin{tikzpicture}[
    node style/.style={circle, draw, minimum size=6pt, inner sep=0pt, fill=white},
    every path/.style={->, >=Stealth}
]

\node[node style] (g1v1) at (-0.5, 0) {};
\node[above=5pt] at (g1v1) {$M$};
\node[node style] (g1v2) at (0.25, 0) {};
\node[above=5pt] at (g1v2) {$CA$};
\node[node style] (g1v3) at (0.25, -1) {};
\node[left=5pt] at (g1v3) {$m_1$};
\draw (g1v3) -- (g1v1)  node[midway, left=5pt ] {$type$};
\draw (g1v3) -- (g1v2)  node[midway, right=5pt ] {$type$};

\node[node style] (p1v2) at (2, 0) {};
\node[above=5pt] at (p1v2) {$CA$};
\node[node style] (p1v3) at (2, -1) {};
\node[left=5pt] at (p1v3) {$v_1$};
\draw (p1v3) -- (p1v2)  node[midway, right=5pt ] {$type$};

\node[node style] (p2v1) at (4, 0) {};
\node[node style] (p2v2) at (4, -1) {};
\node[node style] (p2v3) at (5, -1.5) {};
\node[node style] (p2v4) at (6, -1.5) {};
\node[node style] (p2v5) at (3.2, -2.75) {};
\node[node style] (p2v6) at (3.8, -2.75) {};
\node[node style] (p2v7) at (4.4, -2.75) {};
\node[node style] (p2v8) at (5.5, 0.5) {};

\node[above=5pt] at (p2v8) {$v_1$};
\node[left=5pt] at (p2v1) {$v_2$};
\node[left=5pt] at (p2v2) {$v_3$};
\node[above=8pt] at (p2v3) {$v_4$};
\node[right=5pt] at (p2v4) {$v_5$};

\draw[->] (p2v1) to[out=20, in=200, bend left=40] node[midway, above, right=2pt] {$db$} (p2v2);
\draw[->] (p2v2) to[out=20, in=200, bend left=40] node[midway, above, right=1pt] {$db$} (p2v1);

\draw[->] (p2v2) to[out=20, in=200, bend left=40] node[pos=0.6, above, above=3pt] {$sb$} (p2v3);
\draw[->] (p2v3) to[out=20, in=200, bend left=40] node[pos=0.4, above, above=2pt] {$sb$} (p2v2);

\draw[->] (p2v3) to[out=20, in=200, bend left=40] node[midway, above, above=2pt] {$sb$} (p2v4);
\draw[->] (p2v4) to[out=20, in=200, bend left=40] node[midway, above, above=2pt] {$sb$} (p2v3);

\draw[->] (p2v8) to[out=20, in=200, bend right=30] node[pos=0.5, above, above=2pt] {$a$} (p2v1);
\draw[->] (p2v8) to[out=20, in=200, bend left=18] node[pos=0.25, above, left=2pt] {$a$} (p2v2);
\draw[->] (p2v8) to[out=20, in=200, bend left=18] node[pos=0.5, above, left=2pt] {$a$} (p2v3);
\draw[->] (p2v8) to[out=20, in=200, bend left=18] node[pos=0.5, above, left=2pt] {$a$} (p2v4);

\node[below=5pt] at (p2v5) {$C$};
\node[below=5pt] at (p2v6) {$O$};
\node[below=5pt] at (p2v7) {$H$};

\draw[->] (p2v1) to[out=20, in=20, bend right=40] node[midway, above, left=2pt] {$e$} (p2v6);
\draw[->] (p2v3) to[out=20, in=20, bend left=20] node[midway, above, above=2pt] {$e$} (p2v6);
\draw[->] (p2v4) to[out=20, in=20, bend left=20] node[midway, above, above=2pt] {$e$} (p2v7);
\draw[->] (p2v2) to[out=20, in=20, bend left=20] node[midway, above, left=2pt] {$e$} (p2v5);

\node[anchor=west] at (0, -3.7) {$g$};
\node[anchor=west] at (2, -3.7) {$p_1$};
\node[anchor=west] at (3.7, -3.7) {$p_2$};
\end{tikzpicture}
&
\end{tabular}
};

\node[fit=(content),
minimum width=\textwidth, 
draw=black, dashed, inner sep=8pt, rounded corners=4pt] {};
\end{tikzpicture}
\caption{Example of a graph pattern-based association rule.  In the visualizations, terms are abbreviated as follows:
\texttt{Molecule} as \texttt{M},
\texttt{CarboxylicAcid} as \texttt{CA},
\texttt{hasAtom} as \texttt{a},
\texttt{doubleBond} as \texttt{db},
\texttt{singleBond} as \texttt{sb},
\texttt{element} as \texttt{e},
\texttt{Carbon} as \texttt{C},
\texttt{Oxygen} as \texttt{O}, and
\texttt{Hygrogen} as \texttt{H}.}
\label{fig:Ex3}
\end{figure}
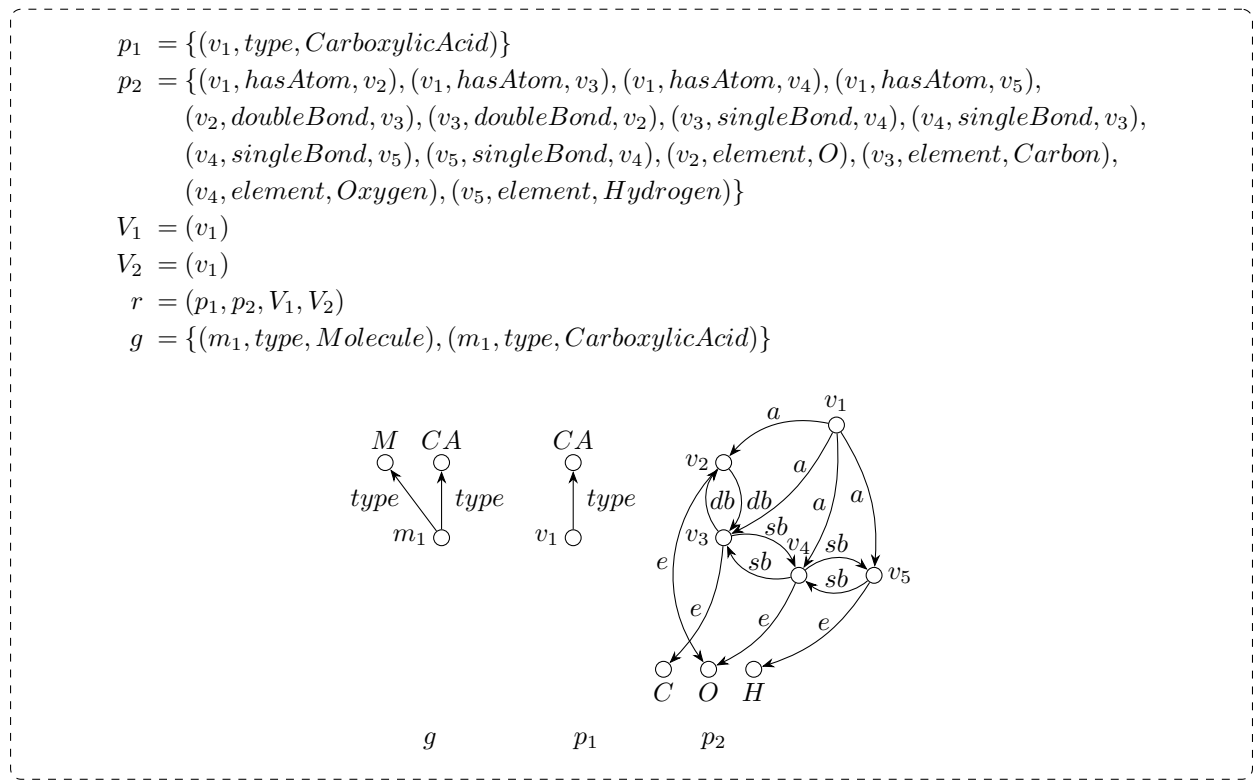

\subsection{Types of Correspondence}

In this section we'll look at different kinds of correspondences between two matches of the patterns $p_1$ and $p_2$ in a graph $g$. %

In the case that there exists a term $t \in \mathcal{T}_{p_1} \cap \mathcal{T}_{p_2}$ that occurs in both patterns in node position (i.e., in subject or object position of a triple pattern), then for each $\mu_1 \in \Omega_{p_1,g}$ and $\mu_2 \in \Omega_{p_2,g}$ it is the case that $t \in \mathcal{T}_{\mu_1(p_1)} \cap \mathcal{T}_{\mu_2(p_2)}$.
Thus, the embedding $\mu_1(p_1)$ overlaps with the embedding $\mu_2(p_2)$ because they share a node term. This can be seen as an a priori correspondence between the embeddings of matches of both patterns, because the correspondence exists independently of $\mu_1$ and $\mu_2$.
We will look at two types of interesting not a priori correspondences.

\begin{enumerate}
\item Given a match $\mu_1$ of $p_1$ on $g$ (i.e., $\mu_1 \in \Omega^{nra}_{p_1,g}$), a match $\mu_2$ of $p_2$ on $g$ (i.e., $\mu_2 \in \Omega^{nra}_{p_2,g}$), a pair $(v_1,v_2)$ of joining variables where $v_1$ occurs in node position in $p_1$ and $v_2$ occurs in node position in $p_2$,
then the embeddings $\mu_1(p_1)$ and $\mu_2(p_2)$ have a node in common if both variables bind to the same term, i.e., $\mu_1(v_1) = \mu_2(v_2)$. The correspondence lies in the sharing of a node term. Topologically, the two embeddings overlap in the graph. However, it may not be the case that all pairs of embeddings of $p_1$ and $p_2$ correspond in the same way.

\item Given a match $\mu_1$ of $p_1$ on $g$ (i.e., $\mu_1 \in \Omega^{nra}_{p_1,g}$), a match $\mu_2$ of $p_2$ on $g$ (i.e., $\mu_2 \in \Omega^{nra}_{p_2,g}$), a pair $(v_1,v_2)$ of joining variables where $v_1$ occurs only in predicate position in $p_1$ and $v_2$ occurs only in predicate position in $p_2$, then the embeddings $\mu_1(p_1)$ and $\mu_2(p_2)$ correspond in that they both contain an edge with the same label. This alone is insufficient for a topological overlap.

\item Given a match $\mu_1$ of $p_1$ on $g$ (i.e., $\mu_1 \in \Omega^{nra}_{p_1,g}$), a match $\mu_2$ of $p_2$ on $g$ (i.e., $\mu_2 \in \Omega^{nra}_{p_2,g}$), a pair $(v_1,v_2)$ of joining variables where $v_1$ occurs only in predicate position in $p_1$ and $v_2$ occurs only in node position in $p_2$, as it is the case with the rule shown in Figure \ref{fig:Ex2}, then the embeddings $\mu_1(p_1)$ and $\mu_2(p_2)$ correspond in that $\mu_1(v_1) = \mu_2(v_2)$, where $\mu_1(v_1)$ is an edge and $\mu_2(v_2)$ is a node. This alone is insufficient for a topological overlap.
\end{enumerate}

\subsection{Trivial Rules}

According to the interpretation of a graph pattern-based association rule $(p_1,p_2,V_1,V_2)$, for each graph $g \in G$ for each mapping function $\mu_1 \in \Omega^{nra}_{p_1,g}$ there exists a corresponding mapping function $\mu_1 \in \Omega^{nra}_{p_1,g}$, where corresponding means that $\mu_1(V_1)=\mu_2(V_2)$. In analogy to trivial itemset-based association rules, once we know that $\mu_1 \in \Omega^{nra}_{p_1,g}$ we can conclude, without making further observations about $g$, especially without making observations about $g \setminus \mu_1(p_1)$, that a corresponding mapping $\mu_2 \in \Omega^{nra}_{p_2,g}$ exists. Thus, the rule does not allow us to make non-trivial predictions about $g$. The fact that for any graph where a $\mu_1$ exist a corresponding $\mu_2$ exists does not depend on $g$, but it depends on the specific relation between $p_1$ and $p_2$.
In GPAR context which relation between $p_1$ and $p_2$ makes a rule trivial is a bit more evolved than in ISAR context. We'll begin with specific types of relations before we introduce the general statement about the relation between $p_1$ and $p_2$ that makes a rule trivial. 
All types of relations have in common that the following is the case:
\[ \forall g \in \mathcal{G} : \forall \mu_1 \in \Omega^{nra}_{p_1,g} : \exists \mu_2 \in \Omega^{nra}_{p_2,\mu_1(p_1)} : \mu_1(p_1) \subseteq \mu_2(p_2)\]
with $\mathcal{G}$ being the infinite set of graphs. 
Note that $p_2$ is not evaluated against $g$, but against a subgraph of $g$, namely the embedding $\mu_1(p_1)$ of $p_1$ corresponding to the match $\mu_1$. Therefore, $\mu_2 \in \Omega^{nra}_{p_2,\mu_1(p_1)}$. The embedding corresponding to the match $\mu_2$ of $p_2$ on $g$ (i.e., $\mu_2(p_2)$) is a subgraph of the embedding corresponding to the match $\mu_1$ of $p_1$ on $g$ (i.e., $\mu_2(p_2)$). 
Thus, those parts of $g$ that are not contained in some embedding of $p_1$ are irrelevant.
We know a priori that if a subgraph of $g$ is the embedding corresponding to a match of $p_1$, then that subgraph is matched by $p_2$. Knowing that $\mu_1$ is a match is sufficient for knowing that a corresponding match $\mu_2$ exists.\\

The most simple case of a trivial rule is a rule where $p_1=p_2$. 
Given that $p_1=p_2$, we can rewrite the condition of a rule being trivial, i.e.,
\[\forall g \in \mathcal{G} : \forall \mu_1 \in \Omega^{nra}_{p_1,g} : \exists \mu_2 \in \Omega^{nra}_{p_2,g} : \mu_1(p_1) \subseteq \mu_2(p_2)\]
to
\[\forall g \in \mathcal{G} : \forall \mu_1 \in \Omega^{nra}_{p_1,g} : \exists \mu_2 \in \Omega^{nra}_{p_1,g} : \mu_2(p_1) \subseteq \mu_1(p_1)\]
This expression is true if $\mu_1 = \mu_2$, which is possible because $\Omega^{nra}_{p_1,g} = \Omega^{nra}_{p_2,g}$. Given that $\mu_1 = \mu_2$, we obtain
\[\forall g \in \mathcal{G} : \forall \mu_1 \in \Omega^{nra}_{p_1,g} : \mu_1(p_1) \subseteq \mu_1(p_1)\] which is obviously true.\\

Note that so far in the discussion about trivial rules we have ignored the sequences of variables $V_1$ and $V_2$. These sequences enable to enforce particular correspondences between two mappings. However, they do not have an impact on the fact that matches of $p_2$ can be obtained by evaluating $p_2$ against embeddings corresponding to matches of $p_1$. Thus, the variable sequences do not matter for the distinction between trivial and non-trivial rules.\\

Beyond the case where $p_1=p_2$, 
which leads to the situation that knowing that $\mu_1$ is a match is sufficient for knowing that a corresponding match $\mu_2$ exists, we are in the same situation if $p_2$ subgraph-isomorphic to $p_1$ (where $p_2 \subseteq p_1$ is a special case), where $p_2$ is a (partial) instantiation of $p_1$, or a combination thereof.

Informally, $p_2$ is subgraph-isomorphic to $p_1$ if we can consistently and using an injection rename variables in $p_2$ so that the result is identical to a subgraph of $p_1$. 
Formally, $p_2$ is subgraph-isomorphic to $p_1$ if there exists a total injective function $m_1 \in \left\lbrace \mathcal{V}_{p_2} \mapsto \mathcal{V}_{p_1} \right\rbrace$ such that $m_1(p_2) \subseteq p_1$, where $m_1(p_2)$ denotes the result of replacing each variable in $p_2$ according to $m_1$.

The following rule (where we omit the variable sequences, as they do not matter here) is an example of a trivial rule where $p_2$ is subgraph-isomorphic to $p_1$, because a variable-to-variable-renaming function $m_1$ with the necessary properties exists:
\begin{align*}
p_1 ~=~& \lbrace (v_1, t_1, v_2), (v_2, t_2, v_3) \rbrace\\
p_2 ~=~& \lbrace (v_1, t_2, v_2) \rbrace\\
m_1 ~=~& \lbrace (v_1, v_2), (v_2, v_3) \rbrace\\
m_1(p_2) ~=~& \lbrace (v_2, t_2, v_3) \rbrace
\end{align*}

Informally, a pattern $p_2$ is an instantiation of a pattern $p_1$ if we can consistently and using an injection replace terms in $p_2$ by variables (that do not already occur in $p_2$) such that both patterns become identical. Formally, $p_2$ is an instantiation of $p_1$ if there exists a mapping function $m_2 \in \left\lbrace \mathcal{T}_{p_2} \mapsto \mathcal{V}_{p_1} \setminus \mathcal{V}_{p_2} \right\rbrace$ such that $m_2(p_2) = p_1$, where $m_2(p_2)$ denotes the result of replacing terms in $p_2$ according to $m_2$.

The following rule is an example of a trivial rule where $p_2$ is an instantiation of $p_1$, because a term-to-variable-renaming function $m_2$ with the necessary properties exists:
\begin{align*}
p_1 ~=~& \lbrace (v_1, t_1, v_2) \rbrace\\
p_2 ~=~& \lbrace (v_1, t_1, t_2) \rbrace\\
m_2 ~=~& \lbrace (t_2, v_2) \rbrace\\
m_2(p_2) ~=~& \lbrace (v_1, t_1, v_2) \rbrace
\end{align*}

Finally, bringing both types of relations together, informally a rule is trivial if the consequent is an instantiation of a pattern that is subgraph-isomorphic to the antecedent. 
Formally, a rule $(p_1,p_2,V_1,V_2)$ is trivial if there exists a total injective function $m \in \left\lbrace (\mathcal{V}_{p_2} \cup \mathcal{T}_{p_2}) \mapsto (\mathcal{V}_{p_1} \setminus \mathcal{T}_{p_1}) \right\rbrace$ where a term can either be mapped to a variable or itself such that $m(p_2) \subseteq p_1$, where $m(p_2)$ denotes the result of replacing each variable and each term in $p_2$ according to $m$.

The following rule is an example of a trivial rule where $p_2$ is an instantiation of a pattern that is subgraph-isomorphic to $p_1$, because a term-and-variable-renaming function $m$ with the necessary properties exists:
\begin{align*}
p_1 ~=~& \lbrace (v_1, t_1, v_2), (v_2, t_2, v_3) \rbrace\\
p_2 ~=~& \lbrace (v_1, t_2, t_3) \rbrace\\
m ~=~& \lbrace (v_1, v_2), (t_2, t_2), (t_3,v_3) \rbrace\\
m(p_2) ~=~& \lbrace (v_2, t_2, v_3) \rbrace
\end{align*}

\subsection{Simplified Notation}
So far we have described a graph pattern-based association rule with two graph patterns $p_1$ and $p_2$ and two non-repetitive sequences of variables $V_1$ and $V_2$. The two graph patterns could have been derived by two  individual mining processes. Thus, that two variables $v_i \in \mathcal{V}_{p_1}$ and $v_j \in \mathcal{V}_{p_2}$ are identical or different might not carry any information. Therefore, via the sequences $V_1$ and $V_2$ we can establish correspondences between variables and we can iterate over the possible sequences of $V_1$ and $V_2$ and thus iterate over the possible correspondences that we can establish. However, we can also rename variables in the patterns and then assume that it carries meaning if two variables $v_i \in \mathcal{V}_{p_1}$ and $v_j \in \mathcal{V}_{p_2}$ are identical or different.

Given the graph pattern-based association rule $r$ of the form $(p_1,p_2,V_1,V_2)$ we can rewrite it into the graph pattern-based association rule $r'$ of the form $(p_1,p_2'')$ %
via the following two steps.
\begin{enumerate}
\item Step 1. We need to take care that it is not the case that in the rewritten rule $r$ a variable occurs in both patterns where this correspondence was not expressed by the original rule $r$. Therefore we transform $p_2$ into $p_2'$. 
Let $V_o = \mathcal{V}_{p_1} \cap (\mathcal{V}_{p_2} \setminus V_2)$\footnote{We slightly abuse the notation and occasionally treat a sequence of variables as a set of variables.} be the set of variables that occur in $p_1$ and $p_2$ but for which no correspondence is expressed via $V_2$. For each variable $v \in V_o$ a new variable name (that occurs in neither $p_1$ nor $p_2'$) is used and each occurrence of $v$ in $p_2$ is replaced by that new name. %

\item Step 2. We need to replace variables in $p_2'$ according to the correspondences expressed via $V_1$ and $V_2$ so that an equally named variable in both patterns expresses correspondence. Therefore, we transform $p_2'$ into $p_2''$. 
Given $V_1 = (v_{1,1}, \ldots, v_{1,n})$ and $V_2 = (v_{2,1}, \ldots, v_{2,n})$, let $V_1 \triangledown V_2$ denote the set $\lbrace (v_{1,i}, v_{2,i}) ~|~ v_{1,i} \in V_1, v_{2,i} \in V_2, 1 \leq i \leq n \rbrace$. For each pair $(v_i,v_j) \in V_1 \triangledown V_2$, we replace each occurrence of $v_j$ in $p_2''$ by $v_i$. The result of the transformation of $p_2'$ is $p_2''$. 
\end{enumerate}

After the two steps have been followed we have transformed the rule $(p_1,p_2,V_1,V_2)$ into $(p_1,p_2'')$. Figure \ref{fig:simplifiednotation} shows an example of a rule $r$ and a possible result of its simplification into $r'$. Note that the two steps defined here do not ensure that the result is unique, as how a new variable is named is not specified. On the left side of the figure the initial rule $r$ is specified and the patterns $p_1$ and $p_2$ are visualized. The correspondence between variables of the patterns according to pairs of joining variables is indicated via dotted lines, labeled with the pair of joining variables that establishes the correspondence. On the right side of the figure the resulting rule $r'$ is specified and the patterns $p_1$ and $p_2''$ are visualized. The dotted lines indicate correspondences between the two patterns via identically named variables. 
In the example, $V_o = \mathcal{V}_{p_1} \cap (\mathcal{V}_{p_2} \setminus V_2) = \lbrace v_1, v_2, v_3\rbrace \cap (\lbrace v_1, v_2, v_3, v_4 \rbrace  \setminus \lbrace v_2, v_3 \rbrace) = \lbrace v_1 \rbrace$.
Thus, in Step 1 we need to replace $v_1$ in $p_2$ by a new variable. We replace it by $v_5$. We obtain the graph pattern $p_2' = \lbrace (v_5, t_3, v_2), (v_2, t_4, v_3), (v_3, t_5, v_4) \rbrace$. 
Then, in Step 2, following $V1 \triangledown V2$, we replace $v_3$ by $v_2$ and we replace $v_2$ by $v_3$ and obtain $p_2'' = \lbrace (v_5, t_3, v_3), (v_3, t_4, v_2), (v_2, t_5, v_4) \rbrace$. Thus, we have obtained the rule $r' = (p_1,p_2'')$.\\

\begin{figure}
\begin{tikzpicture}

\node (content) [
    inner sep=0pt
] {
\begin{tabular}[t]{@{}c@{\hspace{2cm}}c@{}}

$\begin{aligned}
r ~=~& (p_1,p_2,V_1,V_2)\\
p_1 ~=~& \lbrace (v_1, t_1, v_2), (v_1, t_2, v_3) \rbrace\\
p_2 ~=~& \lbrace (v_1, t_3, v_2), (v_2, t_4, v_3), (v_3, t_5, v_4) \rbrace\\
V_1 ~=~& (v_2, v_3)\\
V_2 ~=~& (v_1, v_2)
\end{aligned}$

&
$\begin{aligned}
r' ~=~& (p_1,p_2'')\\
p_1 ~=~& \lbrace (v_1, t_1, v_2), (v_1, t_2, v_3) \rbrace\\
p_2'' ~=~& \lbrace (v_5, t_3, v_3), (v_3, t_4, v_2), (v_2, t_5, v_4) \rbrace\\
\phantom{} &\\
\phantom{} &\\
\end{aligned}$ \\[1.5cm]

\begin{tikzpicture}[
    node style/.style={circle, draw, minimum size=6pt, inner sep=0pt, fill=white},
    every path/.style={->, >=Stealth}
]

\node[node style] (n22) at (2, -2) {};
\node[node style] (n31) at (1, -2.5) {};
\node[node style] (n32) at (2, -3) {};

\node[node style] (n13) at (4.5, -1) {};
\node[node style] (n23) at (4.5, -2) {};
\node[node style] (n33) at (4.5, -3) {};
\node[node style] (n43) at (4.5, -4) {};

\node[left=4pt] at (n31) {$v_1$};
\node[above=4pt] at (n22) {$v_2$};
\node[below=4pt] at (n32) {$v_3$};

\node[right=4pt] at (n13) {$v_4$};
\node[right=4pt] at (n23) {$v_3$};
\node[right=4pt] at (n33) {$v_2$};
\node[right=4pt] at (n43) {$v_1$};

\draw (n31) -- (n22)  node[midway, above=2pt] {$t_1$};
\draw (n31) -- (n32)  node[midway, below=2pt] {$t_2$};

\draw (n43) -- (n33) node[midway, right] {$t_3$};
\draw (n33) -- (n23) node[midway, right] {$t_4$};
\draw (n23) -- (n13) node[midway, right] {$t_5$};

\draw[-, dotted] (n22) -- (n23) node[midway, above] {$(v_2,v_3)$};
\draw[-, dotted] (n32) -- (n33) node[midway, above] {$(v_3,v_2)$};

\node[anchor=west] at (1.2, -4.8) {$p_1$};
\node[anchor=west] at (4.3, -4.8) {$p_2$};
\end{tikzpicture}

&
\begin{tikzpicture}[
    node style/.style={circle, draw, minimum size=6pt, inner sep=0pt, fill=white},
    every path/.style={->, >=Stealth}
]

\node[node style] (n22) at (2, -2) {};
\node[node style] (n31) at (1, -2.5) {};
\node[node style] (n32) at (2, -3) {};

\node[node style] (n13) at (4.5, -1) {};
\node[node style] (n23) at (4.5, -2) {};
\node[node style] (n33) at (4.5, -3) {};
\node[node style] (n43) at (4.5, -4) {};

\node[left=4pt] at (n31) {$v_1$};
\node[above=4pt] at (n22) {$v_2$};
\node[below=4pt] at (n32) {$v_3$};

\node[right=4pt] at (n13) {$v_4$};
\node[right=4pt] at (n23) {$v_2$};
\node[right=4pt] at (n33) {$v_3$};
\node[right=4pt] at (n43) {$v_5$};

\draw (n31) -- (n22)  node[midway, above=2pt] {$t_1$};
\draw (n31) -- (n32)  node[midway, below=2pt] {$t_2$};

\draw (n43) -- (n33) node[midway, right] {$t_3$};
\draw (n33) -- (n23) node[midway, right] {$t_4$};
\draw (n23) -- (n13) node[midway, right] {$t_5$};

\draw[-, dotted] (n22) -- (n23);
\draw[-, dotted] (n32) -- (n33);

\node[anchor=west] at (1.2, -4.8) {$p_1$};
\node[anchor=west] at (4.3, -4.8) {$p_2''$};
\end{tikzpicture}

\end{tabular}
};

\node[fit=(content), minimum width=\textwidth, draw=black, dashed, inner sep=8pt, rounded corners=4pt] {};
\end{tikzpicture}
\caption{Example of a rule of the form $r=(p_1,p_2,V_1,V_2)$ and the result of rewriting it to $r'=(p_1,p_2'')$.}
\label{fig:simplifiednotation}
\end{figure}
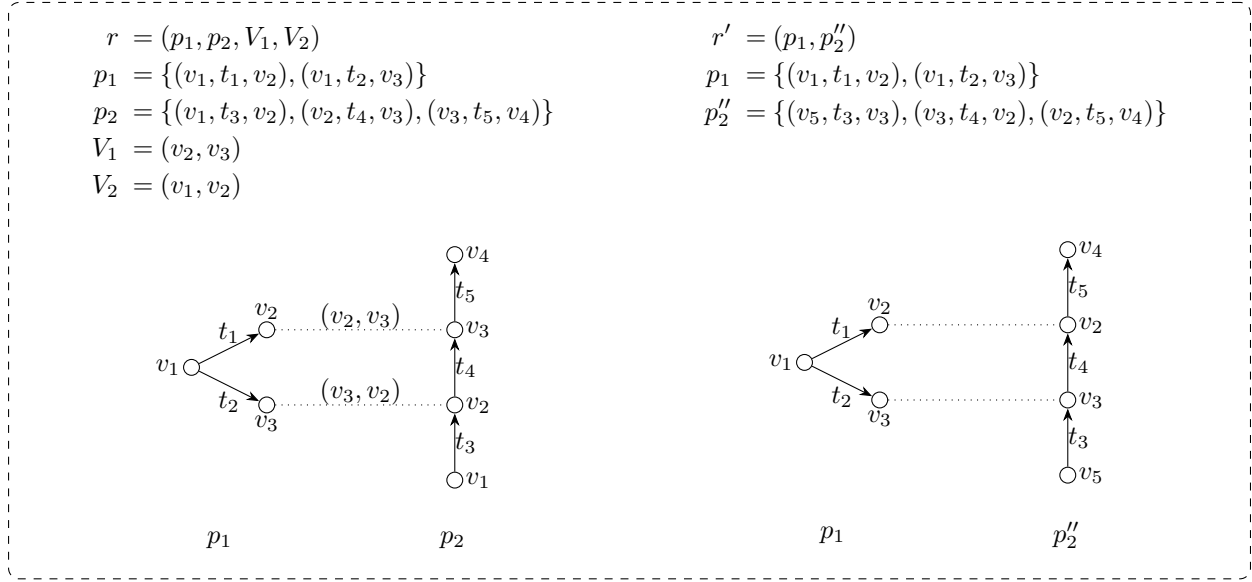

We formalize the meaning of a graph pattern-based association rule of the form $(p_1, p_2)$ in the context of some graph $g$ as follows:
\begin{equation*}
\Omega^{nra}_{p_1,g} \neq \varnothing  \vdash_p \forall \mu_1 \in \Omega^{nra}_{p_1,g} : \exists \mu_2 \in \Omega^{nra}_{p_2,g} :
\forall v \in \mathcal{V}_{p_1} \cap \mathcal{V}_{p_2} : \mu_1(v) = \mu_2(v)
\end{equation*}

Note that for the example rules shown in Figures \ref{fig:Ex1}--\ref{fig:Ex3}, the variables in the consequent pattern have been named in a way to ease understanding of the rule, thus corresponding variables are named equally. In these examples, from a rule $(p_1,p_2,V_1,V_2)$ the simplified form $(p_1,p_2)$ can simply be obtained by omitting the two sequences of variables. 

\subsection{Applications}\label{sec:applications}

We sketch how graph pattern-based association rules can be used in two types of tasks: generative tasks and evaluative tasks. 
In a generative task one produces likely extensions of a graph, whereas in an evaluative tasks one scores a graph by its probability.

Note the analogy to language models \cite{jm3}: language models allows both to extend a sequence of tokens by predicting the most likely continuation and to predict the probability of a sequence of tokens.

\subsubsection{Generative Tasks}
Given a graph $g$ and a rule of the form $(p_1,p_2)$ where $p_1$ matches the graph (i.e., $\Omega^{nra}_{p_1,g} \neq \varnothing$). %
Let's consider the case where all variables in $p_2$ are joining variables (i.e., 
$\mathcal{V}_{p_2} \subseteq \mathcal{V}_{p_1}$). 
Then, for each $\mu \in \Omega_{p_1,g}$ it is the case that $\mu(p_2)$ does not contain variables, because all variables were replaced via $\mu$. Thus, we obtain a (possibly) extended graph $g' = e_r(g)$ (with $|g'| \geq |g|$) via the following operation $e_r$:
\[ e_r(g) = g \cup \bigcup\limits_{\mathclap{\mu \in \Omega^{nra}_{p_1,g}}} \mu(p_2)\]
Once a graph has been extended via that operation it can be the case that it can be extended further using the same operation, as shown in Figure \ref{fig:iterativeextension}. 
Given a rule $r$, let $g^0_r$ be $g$ and let $g^n_r$, with $n > 1$, be the result of extending the graph $g^{n-1}_r$, i.e.:
\[ g^n_r = e_r(g^{n-1}_r)\]

We can then define the closure $Cl_r(g)$ of the graph $g$ under the rule $r$ as follows: 
 \[ Cl_r(g) = \bigcup\limits_{n \in \mathbb{N}} g_r^n\]

\begin{figure}
\begin{tikzpicture}
\node (content) [
    inner sep=0pt
] {
\begin{tabular}[t]{@{}c@{\hspace{0cm}}c@{}}
$\begin{aligned}
g ~=~ & \lbrace (t_1, t_2, t_3), (t_3, t_2, t_4), (t_4, t_2, t_5), (t_5, t_2, t_6) \rbrace\\
p_1 ~=~ & \lbrace (v_1, t_2, v_2), (v_2, t_2, v_3) \rbrace\\
p_2 ~=~ & \lbrace (v_1, t_2, v_3) \rbrace \\
r ~=~ & (p_1,p_2)
\end{aligned}$

\\[1.5cm]

\begin{tikzpicture}[
    node style/.style={circle, draw, minimum size=6pt, inner sep=0pt, fill=white},
    every path/.style={->, >=Stealth}
]

\node[node style] (g1t1) at (-1, -1) {};
\node[above=5pt] at (g1t1) {$t_1$};
\node[node style] (g1t3) at (1, -1) {};
\node[above=5pt] at (g1t3) {$t_3$};
\node[node style] (g1t5) at (3, -1) {};
\node[above=5pt] at (g1t5) {$t_5$};
\node[node style] (g1t6) at (5, -1) {};
\node[above=5pt] at (g1t6) {$t_6$};
\draw (g1t1) -- (g1t3)  node[midway, above=2pt] {$t_2$};
\draw (g1t3) -- (g1t5)  node[midway, above=2pt] {$t_2$};
\draw (g1t5) -- (g1t6)  node[midway, above=2pt] {$t_2$};
\node[anchor=west] at (-2, -1) {$g$};

\node[node style] (g2t1) at (-1, -2) {};
\node[above=5pt] at (g2t1) {$t_1$};
\node[node style] (g2t3) at (1, -2) {};
\node[above=5pt] at (g2t3) {$t_3$};
\node[node style] (g2t5) at (3, -2) {};
\node[above=5pt] at (g2t5) {$t_5$};
\node[node style] (g2t6) at (5, -2) {};
\node[above=5pt] at (g2t6) {$t_6$};
\draw (g2t1) -- (g2t3)  node[midway, above=2pt] {$t_2$};
\draw (g2t3) -- (g2t5)  node[midway, above=2pt] {$t_2$};
\draw (g2t5) -- (g2t6)  node[midway, above=2pt] {$t_2$};
\draw[->] (g2t1) .. controls +(0,-1) and +(0,-1) .. (g2t5) node[midway, above=2pt] {$t_2$};
\draw[->] (g2t3) .. controls +(0,-1) and +(0,-1) .. (g2t6) node[midway, above=2pt] {$t_2$};
\node[anchor=west] at (-2, -2) {$g'$};

\node[node style] (g3t1) at (-1, -3.5) {};
\node[above=5pt] at (g3t1) {$t_1$};
\node[node style] (g3t3) at (1, -3.5) {};
\node[above=5pt] at (g3t3) {$t_3$};
\node[node style] (g3t5) at (3, -3.5) {};
\node[above=5pt] at (g3t5) {$t_5$};
\node[node style] (g3t6) at (5, -3.5) {};
\node[above=5pt] at (g3t6) {$t_6$};
\draw (g3t1) -- (g3t3)  node[midway, above=2pt] {$t_2$};
\draw (g3t3) -- (g3t5)  node[midway, above=2pt] {$t_2$};
\draw (g3t5) -- (g3t6)  node[midway, above=2pt] {$t_2$};
\draw[->] (g3t1) .. controls +(0,-1) and +(0,-1) .. (g3t5) node[midway, above=2pt] {$t_2$};
\draw[->] (g3t3) .. controls +(0,-1) and +(0,-1) .. (g3t6) node[midway, above=2pt] {$t_2$};
\draw[->] (g3t1) .. controls +(0,-2) and +(0,-2) .. (g3t6) node[midway, above=2pt] {$t_2$};
\node[anchor=west] at (-2, -3.5) {$g''$};
\end{tikzpicture}
\end{tabular}
};

\node[fit=(content),
minimum width=\textwidth, 
draw=black, dashed, inner sep=8pt, rounded corners=4pt] {};
\end{tikzpicture}
\caption{An example where a graph can be extended iteratively for two steps using one rule. The rule is given in simplified form, i.e., $(p_1,p_2)$.}
\label{fig:iterativeextension}
\end{figure}
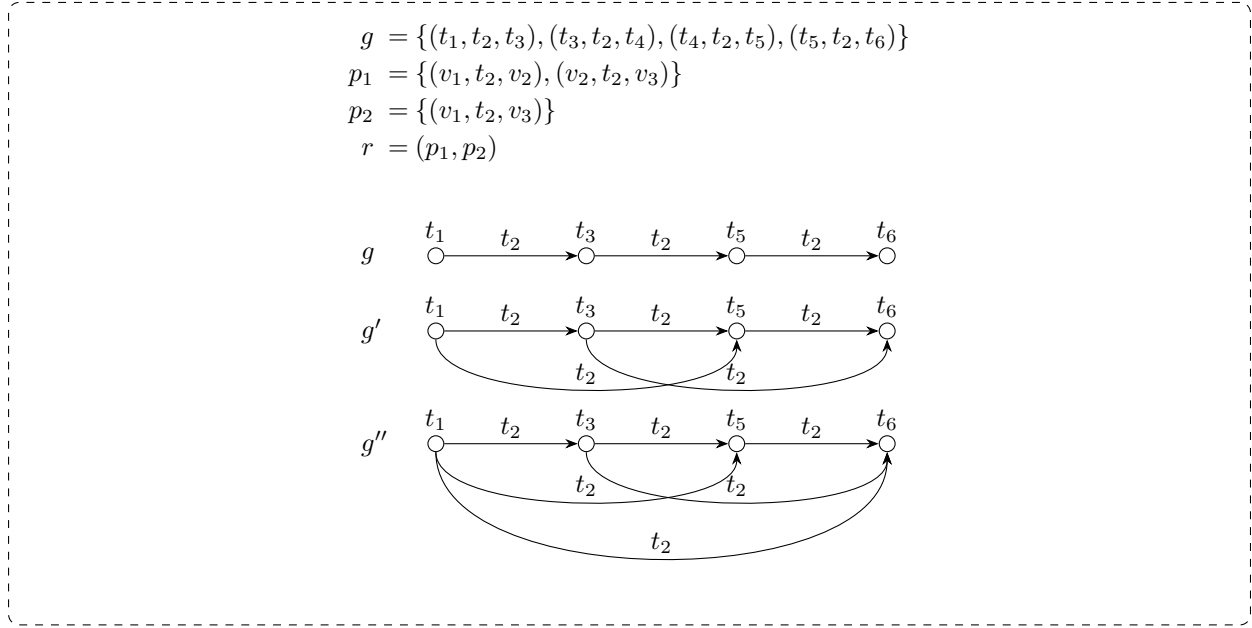

Since we only consider finite graphs and a rule based on finite graph patterns, a fixpoint is reached after a finite number of steps. Note that in practice, however, reaching the fixpoint might not be the goal. Instead, one might be interested in particular extensions in particular regions -- not in all possible extensions everywhere. 
 For the rules shown in Figure \ref{fig:Ex1} and Figure \ref{fig:Ex2} it is the case that $\mathcal{V}_{p_2} \subseteq \mathcal{V}_{p_1}$. Thus, for every $\mu \in \Omega^{nra}_{p_1,g}$ it is the case that $\mu(p_2)$ is a graph. Thus, we can use these rules to extend $g$ to $g \cup \mu(p_2)$.

If $p_2$ contains further variables, then for every $\mu \in \Omega^{nra}_{p_1,g}$ it is the case that $\mu(p_2)$ is a graph pattern that contains variables and the extension $g'$ would not be a graph but a graph pattern. However, this does not mean that those rules would not be interesting for generative tasks. However, it would just be necessary to introduce another step that replaces remaining variables with terms to obtain a graph. Furthermore, when a graph is to be extended iteratively to a graph pattern, then one would need to extend the definition of graph pattern evaluation so that a graph pattern can be evaluated against a graph pattern. 
This is the case for the rule shown in Figure \ref{fig:Ex3}.

What we leave out in the sketch of the generative task is how to select the rules. Not only do we need to select rules where the antecedent pattern matches the graph, we might also want to select rules that can extend the graph in a particular area (thus we'd have constraints about the antecedent pattern), that can extend the graph in a particular way (thus we'd have constraints about the consequent pattern), or we'd select rules based on quality metrics. In Chapter \ref{sec:metrics} we introduce quality metrics for graph pattern-based association rules. For a concrete use case one would need to investigate how to select rules according to the scores obtained with these metrics.

A well-known task is link prediction in knowledge graphs, which is the task of inferring missing or likely relationships between entities in a knowledge graph. 
Here, three kinds of link prediction tasks are distinguished: head term prediction (HTP), relation term prediction (RTP), and tail term prediction (TTP), which can be described as follows:

\begin{itemize}
\item (HTP) Given a knowledge graph $g$ and an incomplete triple $(?,r,t)$, where $r$ is the relation and $t$ is the tail (target) term, the task is to identify the most likely head (source) term $s \in V$ where $V$ is the vocabulary of terms.
\item (RTP) Given a knowledge graph $g$ and an incomplete triple $(s,?,t)$, where $s$ is the source term and $t$ is the tail (target) term, the task is to identify the most likely relation term $r \in V$ where $V$ is the vocabulary of terms.
\item (TTP) Given a knowledge graph $g$ and an incomplete triple $(s,r,?)$, where $s$ is the source term and $r$ is the relation term, the task is to identify the most likely tail (target) term $t \in V$ where $V$ is the vocabulary of terms.
\end{itemize}

For each of these tasks, adding the predicted triple $(s,r,t)$ to $g$ would provide plausible new information that was missing from the graph.

Note that recent rule-based approaches such as SAFRAN \cite{safran} and AnyBURL \cite{meilicke2019introduction} have achieved results for the link prediction task that are competitive with those achieved by approaches based on neural embedding-based machine learning models.

In the context of link prediction based on graph pattern-based association rules we can distinguish two types of situations: either the rule contains the term to be inserted into the incomplete triple (e.g., the rule contain the term $t$ so that we can complete $(s,r,?)$ to $(s,r,t)$), or it does not. The consequence is that when we have a specific vocabulary given and for each term in that vocabulary want to select, from a large set of rules, those rules that could add that term, then using only those rules where the consequent pattern contains that term might be insufficient.

Figure \ref{fig:TTPnotcontained} shows an example with a graph $g$ and a rule applied for tail term prediction. Given $\mu = \lbrace (v_1,t_1),(v_2,t_3),(v_3,t_5) \rbrace \in \Omega^{nra}_{p_1,g}$, by adding $\mu(p_2) = \lbrace (t_3,t_8,t_5)\rbrace$ to $g$ we extend $g$ with the triple $(t_3,t_8,t_5)$, thus, the predicted tail (target) term is $t_5$. In this example, $t_5$ is not contained in $p_2$ (i.e., $t_5 \not \in \mathcal{T}_{p_2}$). This shows that rules where the target term is not contained in the consequent pattern can be used to extend the graph with that term. 
Figure~\ref{fig:RTPcontained} shows an example with a graph $g$ and a rule applied for relation term prediction.
Given $\mu = \lbrace (v1,t2),(v2,t3),(v3,t5),(v4,t6),(v5,t7) \rbrace \in \Omega^{nra}_{p_1,g}$, by adding $\mu(p_2) = \lbrace (t_3,t_8,t_7) \rbrace$ to $g$ we extend $g$ with the triple $(t_3,t_8,t_7)$, thus, the predicted relation term is $t_8$. In this example, $t_8$ is contained in $p_2$ (i.e., $t_8 \in \mathcal{T}_{p_2}$). 
For the link prediction tasks HTP, RTP and TTP it is sufficient that the consequent pattern consists of a single triple pattern.
However, even in the context of generative tasks it can be useful to allow more complex consequents, with which a graph can be extended in larger steps. 
For example, in the context of graphs such as the one shown in Figure \ref{fig:examplegraph}, one can imagine a antecedent that identifies between which identified entities a relation is expressed and then, in one step, extends the graph with a new relation instance, describing the head and tail entities of the relation instance.

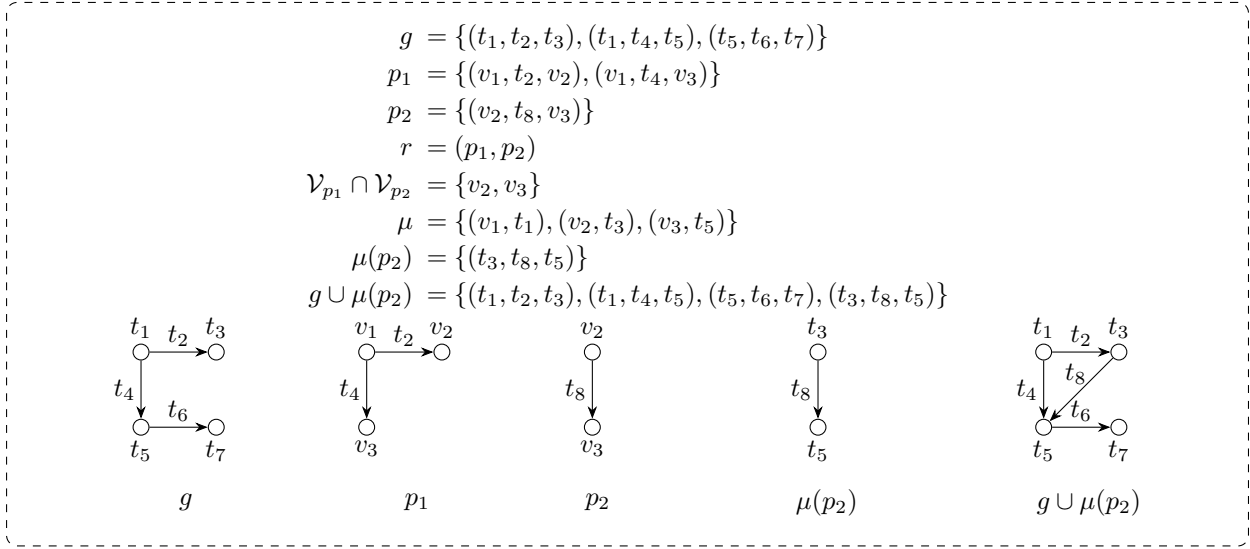
\begin{figure}
\begin{tikzpicture}
\node (content) [
    inner sep=0pt
] {
\begin{tabular}[t]{@{}c@{\hspace{0cm}}c@{}}
$\begin{aligned}
g ~=~ & \lbrace (t_1, t_2, t_3), (t_1,t_4,t_5), (t_5,t_6,t_7) \rbrace\\

p_1 ~=~ & \lbrace (v_1, t_2, v_2), (v_1, t_4, v_3) \rbrace\\

p_2 ~=~ & \lbrace (v_2, t_8, v_3) \rbrace\\
r ~=~ & (p_1,p_2)\\
\mathcal{V}_{p_1} \cap \mathcal{V}_{p_2} ~=~ & \lbrace v_2, v_3 \rbrace\\

\mu ~=~ & \lbrace (v_1, t_1), (v_2, t_3), (v_3, t_5)\rbrace\\

\mu(p_2) ~=~ & \lbrace (t_3, t_8, t_5) \rbrace \\

g \cup \mu(p_2) ~=~ & \lbrace (t_1, t_2, t_3), (t_1,t_4,t_5), (t_5,t_6,t_7), (t_3, t_8, t_5) \rbrace
\end{aligned}$
&

\\[1.8cm]

\begin{tikzpicture}[
    node style/.style={circle, draw, minimum size=6pt, inner sep=0pt, fill=white},
    every path/.style={->, >=Stealth}
]

\node[node style] (gt1) at (1, -1) {};
\node[above=5pt] at (gt1) {$t_1$};
\node[node style] (gt3) at (2, -1) {};
\node[above=5pt] at (gt3) {$t_3$};
\node[node style] (gt5) at (1, -2) {};
\node[below=5pt] at (gt5) {$t_5$};
\node[node style] (gt7) at (2, -2) {};
\node[below=5pt] at (gt7) {$t_7$};
\draw (gt1) -- (gt3)  node[midway, above=2pt] {$t_2$};
\draw (gt1) -- (gt5)  node[midway, left=2pt] {$t_4$};
\draw (gt5) -- (gt7)  node[midway, above=2pt] {$t_6$};

\node[node style] (p1v1) at (4, -1) {};
\node[above=5pt] at (p1v1) {$v_1$};
\node[node style] (p1v2) at (5, -1) {};
\node[above=5pt] at (p1v2) {$v_2$};
\node[node style] (p1v3) at (4, -2) {};
\node[below=5pt] at (p1v3) {$v_3$};
\draw (p1v1) -- (p1v2)  node[midway, above=2pt] {$t_2$};
\draw (p1v1) -- (p1v3)  node[midway, left=2pt] {$t_4$};

\node[node style] (p2v2) at (7, -1) {};
\node[above=5pt] at (p2v2) {$v_2$};
\node[node style] (p2v3) at (7, -2) {};
\node[below=5pt] at (p2v3) {$v_3$};
\draw (p2v2) -- (p2v3)  node[midway, left=2pt] {$t_8$};

\node[node style] (muv2) at (10, -1) {};
\node[above=5pt] at (muv2) {$t_3$};
\node[node style] (muv3) at (10, -2) {};
\node[below=5pt] at (muv3) {$t_5$};
\draw (muv2) -- (muv3)  node[midway, left=2pt] {$t_8$};

\node[node style] (gnewt1) at (13, -1) {};
\node[above=5pt] at (gnewt1) {$t_1$};
\node[node style] (gnewt3) at (14, -1) {};
\node[above=5pt] at (gnewt3) {$t_3$};
\node[node style] (gnewt5) at (13, -2) {};
\node[below=5pt] at (gnewt5) {$t_5$};
\node[node style] (gnewt7) at (14, -2) {};
\node[below=5pt] at (gnewt7) {$t_7$};
\draw (gnewt1) -- (gnewt3)  node[midway, above=2pt] {$t_2$};
\draw (gnewt1) -- (gnewt5)  node[midway, left=2pt] {$t_4$};
\draw (gnewt5) -- (gnewt7)  node[midway, above=2pt] {$t_6$};
\draw (gnewt3) -- (gnewt5)  node[midway, left=2pt, above=2pt] {$t_8$};

\node[anchor=west] at (1.5, -3) {$g$};
\node[anchor=west] at (4.5, -3) {$p_1$};
\node[anchor=west] at (6.9, -3) {$p_2$};
\node[anchor=west] at (9.7, -3) {$\mu(p_2)$};
\node[anchor=west] at (12.9, -3) {$g \cup \mu(p_2)$};

\end{tikzpicture}
&
\end{tabular}
};

\node[fit=(content),
minimum width=\textwidth, 
draw=black, dashed, inner sep=8pt, rounded corners=4pt] {};
\end{tikzpicture}
\caption{Tail Term Prediction (TTP) example where the predicted tail node term (i.e., $t_5$) is not contained in $\mathcal{T}_{p_2}$. The rule is given in simplified form, i.e., $(p_1,p_2)$.}
\label{fig:TTPnotcontained}
\end{figure}

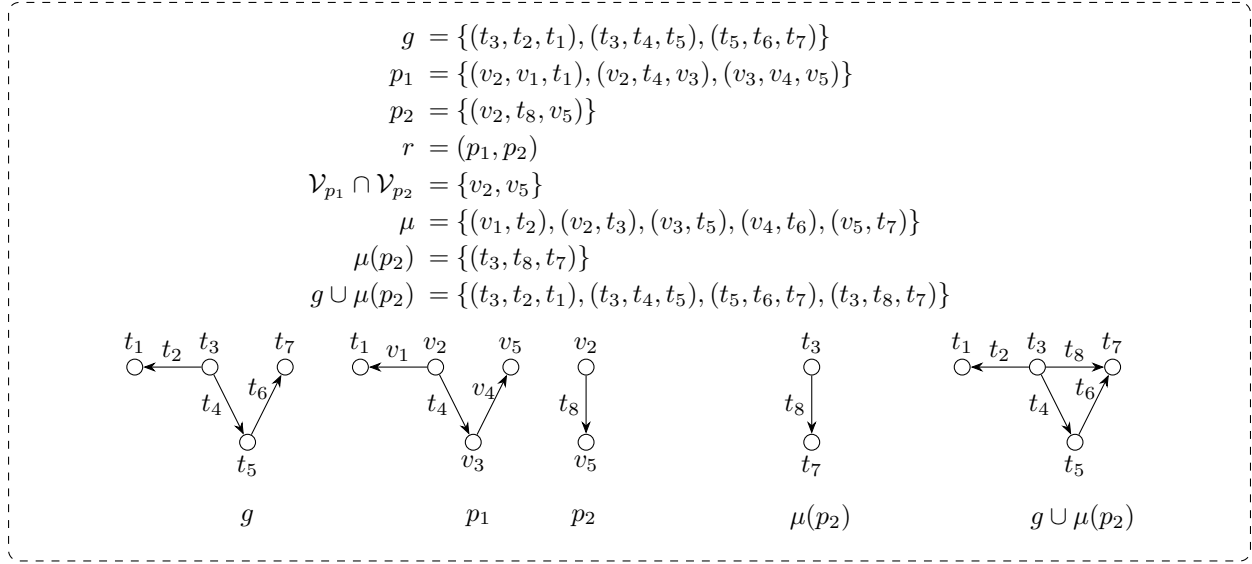
\begin{figure}
\begin{tikzpicture}
\node (content) [
    inner sep=0pt
] {
\begin{tabular}[t]{@{}c@{\hspace{0cm}}c@{}}
$\begin{aligned}
g ~=~ & \lbrace (t_3, t_2, t_1), (t_3, t_4, t_5), (t_5, t_6, t_7) \rbrace\\
p_1 ~=~ & \lbrace (v_2, v_1, t_1), (v_2, t_4, v_3), (v_3, v_4, v_5) \rbrace\\
p_2 ~=~ & \lbrace (v_2, t_8, v_5) \rbrace \\
r ~=~ & (p_1,p_2)\\ 
\mathcal{V}_{p_1} \cap \mathcal{V}_{p_2} ~=~ & \lbrace v_2, v_5 \rbrace\\
\mu ~=~ & \lbrace (v_1,t_2),(v_2,t_3),(v_3,t_5),(v_4,t_6),(v_5,t_7)\rbrace\\
\mu(p_2) ~=~ & \lbrace (t_3, t_8, t_7) \rbrace\\
g \cup \mu(p_2) ~=~ & \lbrace (t_3, t_2, t_1), (t_3, t_4, t_5), (t_5, t_6, t_7), (t_3, t_8, t_7) \rbrace
\end{aligned}$
&

\\[2cm]

\begin{tikzpicture}[
    node style/.style={circle, draw, minimum size=6pt, inner sep=0pt, fill=white},
    every path/.style={->, >=Stealth}
]

\node[node style] (gt1) at (-1, -1) {};
\node[above=5pt] at (gt1) {$t_1$};
\node[node style] (gt3) at (0, -1) {};
\node[above=5pt] at (gt3) {$t_3$};
\node[node style] (gt5) at (0.5, -2) {};
\node[below=5pt] at (gt5) {$t_5$};
\node[node style] (gt7) at (1, -1) {};
\node[above=5pt] at (gt7) {$t_7$};
\draw (gt3) -- (gt1)  node[midway, above=2pt] {$t_2$};
\draw (gt3) -- (gt5)  node[midway, left=2pt] {$t_4$};
\draw (gt5) -- (gt7)  node[midway, left=3pt, above=2pt] {$t_6$};

\node[node style] (p1t1) at (2, -1) {};
\node[above=5pt] at (p1t1) {$t_1$};
\node[node style] (p1v2) at (3, -1) {};
\node[above=5pt] at (p1v2) {$v_2$};
\node[node style] (p1v3) at (3.5, -2) {};
\node[below=5pt] at (p1v3) {$v_3$};
\node[node style] (p1v5) at (4, -1) {};
\node[above=5pt] at (p1v5) {$v_5$};
\draw (p1v2) -- (p1t1)  node[midway, above=2pt] {$v_1$};
\draw (p1v2) -- (p1v3)  node[midway, left=2pt] {$t_4$};
\draw (p1v3) -- (p1v5)  node[midway, left=3pt, above=2pt] {$v_4$};

\node[node style] (p2v2) at (5, -1) {};
\node[above=5pt] at (p2v2) {$v_2$};
\node[node style] (p2v5) at (5, -2) {};
\node[below=5pt] at (p2v5) {$v_5$};
\draw (p2v2) -- (p2v5)  node[midway, left=2pt] {$t_8$};

\node[node style] (muv2) at (8, -1) {};
\node[above=5pt] at (muv2) {$t_3$};
\node[node style] (muv3) at (8, -2) {};
\node[below=5pt] at (muv3) {$t_7$};
\draw (muv2) -- (muv3)  node[midway, left=2pt] {$t_8$};

\node[node style] (gnewt1) at (10, -1) {};
\node[above=5pt] at (gnewt1) {$t_1$};
\node[node style] (gnewt3) at (11, -1) {};
\node[above=5pt] at (gnewt3) {$t_3$};
\node[node style] (gnewt5) at (11.5, -2) {};
\node[below=5pt] at (gnewt5) {$t_5$};
\node[node style] (gnewt7) at (12, -1) {};
\node[above=5pt] at (gnewt7) {$t_7$};
\draw (gnewt3) -- (gnewt1)  node[midway, above=2pt] {$t_2$};
\draw (gnewt3) -- (gnewt5)  node[midway, left=2pt] {$t_4$};
\draw (gnewt5) -- (gnewt7)  node[midway, left=3pt, above=2pt] {$t_6$};
\draw (gnewt3) -- (gnewt7)  node[midway, above=2pt] {$t_8$};

\node[anchor=west] at (0.4, -3) {$g$};
\node[anchor=west] at (3.4, -3) {$p_1$};
\node[anchor=west] at (4.8, -3) {$p_2$};
\node[anchor=west] at (7.7, -3) {$\mu(p_2)$};
\node[anchor=west] at (10.9, -3) {$g \cup \mu(p_2)$};
\end{tikzpicture}

&
\end{tabular}
};

\node[fit=(content),
minimum width=\textwidth, 
draw=black, dashed, inner sep=8pt, rounded corners=4pt] {};
\end{tikzpicture}
\caption{Relation Term Prediction (RTP) example where the predicted relation term ($t_8$) is contained in $\mathcal{T}_{p_2}$. The rule is given in simplified form, i.e., $(p_1,p_2)$.}
\label{fig:RTPcontained}
\end{figure}

\subsubsection{Evaluative Tasks}
We provide some ideas on how the probability of a given graph $g$ can be predicted given a set $R$ of rules of the form $(p_1,p_2)$ that were derived from a set $G$ of graphs.

We make use of the following properties:
\begin{itemize}
\item $p_1$ matches many/few graphs in $G$ (i.e., taking into account $|\lbrace \, g' \in G \mid \Omega^{nra}_{p_1,g'} \neq \varnothing \,\rbrace|$).
\item If $p_1$ matches a graph $g' \in G$, then it matches often/rarely (i.e., taking into account $|\Omega^{nra}_{p_1,g'}|$).
\item The predictions made by a rule are often/rarely already contained in a graph $g' \in G$. If the prediction $\mu(p_2)$ contains no variables, then the prediction is contained in a graph $g' \in G$ if $\mu(p_2) \subseteq g'$; 
if the prediction $\mu(p_2)$ contains variables, then the prediction is "contained" in a graph $g' \in G$ if $\Omega^{nra}_{\mu(p_2),g'} \neq \varnothing$.
\end{itemize}

These properties can then be involved in the prediction of a graph's probability:
\begin{enumerate}
\item A rule where $p_1$ matches many graphs in $G$ but does not match $g$ may indicate that the probability of $g$ is low.

\item A rule where $p_1$ matches many graphs in $G$ and does match $g$ may indicate that the probability of $g$ is high.

\item A rule where $p_1$ matches many graphs in $G$ often but does not or rarely match $g$ may indicate that the probability of $g$ is low.

\item A rule where $p_1$ matches many graphs in $G$ often and also matches $g$ often may indicate that the probability of $g$ is high.

\item If a rule that often matches in $G$ and the predictions made by the rule are often present in these graphs, then if $p_1$ matches the graph but the predictions are not contained in the graph, then that may indicate that the probability of $g$ is low.

\item If a rule that often matches in $G$ and the predictions made by the rule are often present in these graphs, then if $p_1$ matches the graph and the predictions are contained in the graph, then tat may indicate that the probability of $g$ is high.
\end{enumerate}

\section{Metrics for Graph Pattern-based Association Rules}\label{sec:metrics}

In this section we introduce the metrics confidence, lift, leverage, and conviction for graph pattern-based association rules and show to what extent these metrics have the same characteristics as their counterparts in the context of itemset-based association rules.

\subsection{Sets and Probability Space}\label{sec:probspace}

Before the probability space is defined and the metrics are introduced, some core ideas need to be discussed and sets need to be defined.%

Given a graph $g$ and a graph pattern-based association rule $(p_1, p_2, V_1, V_2)$ (with $|V_1| = |V_2| = n$), a mapping function $\mu_1 \in \Omega^{nra}_{p_1,g}$ corresponds to a mapping function $\mu_2 \in \Omega^{nra}_{p_2,g}$ if there exists a non-repetitive term sequence $T$ of length $n$ such that $\mu_1(V_1) = \mu_2(V_2) = T$. The specific correspondence between $\mu_1$ and $\mu_2$ is characterized by $T$. To quantify the strength of correspondence between two patterns we need to define i) the set of term sequences that \textit{could} characterize a correspondence between two mappings of the patterns and ii) the set of term sequences that \textit{actually do} characterize a correspondence between two mappings of the patterns.

Consider the set of non-repetitive term sequences of length $n$ created from the set $\mathcal{T}_g$, i.e., the set of terms that occur in the graph $g$. 
In the case that the pattern $p_1$ contains terms (i.e., $\mathcal{T}_{p_1} \neq \varnothing$), then for a term sequence that contains terms that are member of $\mathcal{T}_{p_1}$ we know a priori that there cannot exist a mapping function $\mu \in \Omega^{nra}_{p_1,g}$ such that $\mu(V_1) = T$ due to no-repeated-anything semantics. 
Likewise, for a term sequence for which in principle a mapping function $\mu_1 \in \Omega^{nra}_{p_1,g}$ exists we can know a priori that no mapping function $\mu_2 \in \Omega^{nra}_{p_2}$ can exist such that $T$ characterizes the correspondence between $\mu_1$ and $\mu_2$ in the case that $T$ contains terms that are member of $\mathcal{T}_{p_2}$.  

Thus, we introduce the set $\fulltau$, which is the set of non-repetitive term sequences of length $n$ that are created from terms in $\mathcal{T}_{g} \setminus (\mathcal{T}_{p_1} \cup \mathcal{T}_{p_2})$. These term sequences could in principle characterize a correspondence between two mappings $\mu_1 \in \Omega^{nra}_{p_1,g}$ and $\mu_2 \in \Omega^{nra}_{p_2,g}$. We define this set as follows:
\[ \mathcal{T}^n_{g,p_1,p_2} := \lbrace \,(t_1, \ldots, t_n) \in \bigtimes_{i=1}^n (\mathcal{T}_g \setminus (\mathcal{T}_{p_1} \cup \mathcal{T}_{p_2})) \mid \forall i,j \in [1,n] : i \neq j \Rightarrow t_i \neq t_j \,\rbrace\]

Note that $\mathcal{T}^n_{g,p_1,p_2} = \mathcal{T}^n_{g,p_2,p_1}$. We compute the cardinality of $\mathcal{T}^n_{g,p_1,p_2}$ as follows:
\begin{equation} |\mathcal{T}^n_{g,p_1,p_2}| = \prod_{i=0}^{n-1} \Big( |\mathcal{T}_g \setminus (\mathcal{T}_{p_1} \cup \mathcal{T}_{p_2})| -i \Big)\label{cardinality_fulltau}
 \end{equation}
 
We introduce a predicate that, given a graph $g$, a sequence $T$ of terms, a pattern $p$, and a sequence $V$ of variables, for which its truth value is evaluated as follows:
\[ m_g(T,p,V) \equiv \exists \mu \in \Omega^{nra}_{p,g} : \mu(V) = T \]

We say that, in the context of a graph $g$, a pattern $p$, and a sequence $V$ of variables, a mapping function $\mu \in \Omega^{nra}_{p,g}$ corresponds to a term sequence $T$ (or, likewise, a term sequence $T$ corresponds to a mapping function $\mu \in \Omega^{nra}_{p,g}$), if $m_g(T,p,V)$ is true. 
Thus, the predicate evaluates to true if for a given graph $g$, pattern $p$, variable sequence $V$, and term sequence $T$ there exists a mapping function that corresponds to $T$.\\

Given a graph $g$, two graph patterns $p_1$ and $p_2$ such that $\Omega^{nra}_{p_1,g} \neq \varnothing$, $\Omega^{nra}_{p_2,g} \neq \varnothing$, and $\fulltau \neq \varnothing$, and given a graph pattern-based association rule $(p_1,p_2,V_1,V_2)$, we can formulate a probability space $(\Omega,\mathcal{F},\hat{P}_g)$ as follows:
\begin{itemize}
\item The sample space $\Omega$ is the set $\fulltau$ -- we pick a term sequence $T$ uniformly at random from $\fulltau$, i.e., the set of term sequences that could in principle characterize a correspondence between two matches $\mu_1 \in \Omega^{nra}_{p_1,g}$ and $\mu_2 \in \Omega^{nra}_{p_2,g}$.

\item The sigma-algebra $\mathcal{F}$ is the powerset $\mathcal{P}(\Omega)$, because $\Omega$ is finite and thus every subset of $\Omega$ is a measurable event.

\item The empirical probability measure $\hat{P}_g : \mathcal{F} \to [0,1]$ is defined for every event $X \subseteq \Omega$ by $\hat{P}_g(X) = \tfrac{|X|}{|\fulltau|}$.
Given a set $E \subseteq \fulltau$ of term sequences, we associate it to the event $E_{g,p_1,V_1;p_2} \coloneq \lbrace \, T \in \fulltau \mid m_g(T,p_1,V_1) \, \rbrace$ (and to the event $E_{g,p_2,V_2;p_1} \coloneq \lbrace \, T \in \fulltau \mid m_g(T,p_2,V_2) \, \rbrace$). 
Thus, the probability $\hat{P}(E_{g,p_1,V_1;p_2})$
that a term sequence $T$, selected uniformly at random from $\fulltau$, corresponds to a match $\mu$ of $p_1$ (i.e., $\mu \in \Omega^{nra}_{p_1,g}$) such that $\mu(V_1) = T$, is measured as follows:
\[ \hat{P}_g(E_{g,p_1,V_1;p_2}) = \frac{|\lbrace \,T \in \fulltau \mid m_g(T,p_1,V_1) \,\rbrace|}{|\fulltau|}\]

The event $E_{g,p_p,V_2;p_1}$ is defined equivalently and the empirical probability $\hat{P}_g(E_{g,p_2,V_2;p_1})$ is measured equivalently.

\end{itemize}

The requirement $\fulltau \neq \varnothing$ means that it needs to be the case that in principle each of the two patterns could match the graph. As a trivial case, no non-empty pattern could match an empty graph. But we only consider the case where all graphs are non-empty. However, with $n=|V_1| = |V_2|$, if $n > |\mathcal{T}_g \setminus (\mathcal{T}_{p_2} \cup \mathcal{T}_{p_2})|$, then $\fulltau = \varnothing$, because there do not exist sufficiently many terms the variables could be bound to to create a non-repetitive term sequence of length $n$.\\

Assuming $\fulltau$ to be non-empty, $(\Omega,\mathcal{F},\hat{P}_g)$ is a probability space, because it satisfies the Kolmogorov axioms:
\begin{enumerate}
\item Non-negativity (Axiom 1).

For all $X \subseteq \Omega$, $|X| \geq 0$, and $|T| > 0$. Thus, $\hat{P}_g(X) = \tfrac{|X|}{|\fulltau|}$ is a rational and thus a real number. Furthermore, since $|X| \geq 0$ we have $\hat{P}_g(X) \geq 0$ and $\tfrac{|X|}{|\fulltau|} \geq 0$.

\item Normalization (Axiom 2).

Because $\Omega = \fulltau$, it follows that $\hat{P}_g(\Omega) = \tfrac{|\Omega|}{|\fulltau|} = \tfrac{|\Omega|}{|\Omega|} = 1$.

\item Countable additivity (Axiom 3).

Let $X = (X_1, \ldots,X_m)$ be a sequence of pairwise disjoint non-empty subsets of $\Omega$. Then
\[ \hat{P}_g\left(\bigcup_{X_i \in X} X_i\right) =
\frac{\left|\displaystyle \bigcup_{X_i \in X}(X_i)\right|}{|\fulltau|} = \frac{\displaystyle \sum_{X_i \in X}|X_i|}{|\fulltau|} = \sum_{X_i \in X} \frac{|X_i|}{|\fulltau|} = \sum_{X_i \in X} \hat{P}_g(X_i)\]
\end{enumerate}

We can describe the empirical joint and conditional probabilities as follows:

\[
\begin{array}{rl}
\hat{P}_g(E_{g,p_1,V_1;p_2} \cap E_{g,p_2,V_2;p_1}) =& \dfrac{|\lbrace \, T \in \fulltau \mid m_g(T,p_1,V_1) \land m_g(T,p_2,V_2)\,\rbrace|}{|\fulltau|}\\[12pt]
\hat{P}_g(E_{g,p_2,V_2;p_1} \mid E_{g,p_1,V_1;p_2}) =& \dfrac{|\lbrace \, T \in \fulltau \mid m_g(T,p_1,V_1) \land m_g(T,p_2,V_2) \,\rbrace|}{|\lbrace \, T \in \fulltau \mid m_g(T,p_1,V_1) \,\rbrace|}
\end{array}
\]

$\hat{P}_g(E_{g,p_1,V_1;p_2})$ denotes the empirical probability, estimated from the graph $g$, that a randomly selected term sequence $T \in \fulltau$ satisfies $m_g(T,p_1,V_1)$; which means that a match of $p_1$ on $g$ (i.e., $\mu \in \Omega^{nra}_{p_1,g}$) exists that is characterized by $T$ (i.e.,  $\mu(V_1) = T$). 
$\hat{P}_g(E_{g,p_2,V_2;p_1})$ can be interpreted equivalently. 
$\hat{P}_g(E_{g,p_1,V_1;p_2} \cap E_{g,p_2,V_2;p_1})$ denotes the empirical probability, estimated from the graph $g$, that a randomly selected term sequence $T \in \fulltau$ satisfies both $m_g(T,p_1,V_1)$ and $m_g(T,p_2,V_2)$. 
Finally, $\hat{P}_g(E_{g,p_2,V_2;p_1} \mid E_{g,p_1,V_1;p_2})$ denotes the empirical conditional probability, estimated from the graph $g$, that a randomly selected term sequence $T \in \fulltau$ satisfies $m_g(T,p_2,V_2)$ under the condition that $T$ satisfies $m_g(T,p_1,V_1)$.

Note the necessary difference in how events are defined here in contrast to how they are defined in the context of itemset-based association rules. In ISAR context, for a rule $A \Rightarrow B$ we were able to define the
relative support of the antecedent without knowledge about the consequent (or vice versa). 
Here, however, since we make use of the set $\fulltau$, we need to know both the antecedent pattern as well as the consequent pattern, so that we can discard those term sequences for which we know a priori that they cannot characterize correspondences. Thus, we call the event $E_{g,p_1,V_1;p_2}$ instead of $E_{g,p_1,V_1}$. To denote that the pattern $p_2$ is a bit secondary (because we mainly care about matches of $p_1$), but cannot be ignored (because without $p_2$ we cannot specify $\fulltau$), we put the secondary $p_2$ after the semicolon in $E_{g,p_1,V_1;p_2}$.

\subsection{Metrics for the Single-graph Setting}\label{sec:metricssingle}

Given a graph $g$ and a rule $(p_1,p_2,V_1,V_2)$ with $|V_1|=|V_2|=n$. Let $\mathcal{P}^n$, with $|\mathcal{V}_p| \geq n$ for each $p \in \mathcal{P}^n$, be the set of graph patterns with at least $n$ variables, and let $\mathcal{V}^n$ denote the set of non-repetitive variable sequences of length $n$. 

In the context of a rule $(p_1,p_2,V_1,V_2)$ and a graph $g$ with $\fulltau \neq \varnothing$, when we define absolute and relative support of one of the two graph patterns, e.g., $p_1$, then we also need to take into account the other graph pattern, e.g., $p_2$. Furthermore, we need to take into account the sequence of joining variables, e.g., $V_1$. Thus, we evaluate the support of a tuple of the form $(p_1,V_1;p_2)$. Here, $p_1$ is the pattern we are primarily interested in, but we also need to take into account the second graph pattern, i.e., $p_2$.

We define absolute support and relative support (for a single graph $g$ or a bag $G$ of graphs) of a pair such as $(p_i,V_i;p_j)$ (where $p_i = p_1$ and $V_i = V_1$, or, $p_i = p_2$ and $V_i = V_2$) as follows:
\[
\begin{array}{rl}
\operatorname{absolute-support}_g : & (\mathcal{P}^n \times \mathcal{V}^n \times \mathcal{P}^n) 
 \longrightarrow \mathbb{N}, \\[6pt]
& (p_i,V_i;p_j) \longmapsto |\lbrace ~T \in \mathcal{T}^n_{g,p_1,p_2} \mid m_g(T,p_i,V_i)~\rbrace|\\[6pt]
& = \sum_{T \in \mathcal{T}^n_{g,p_1,p_2}}\mathbb{1}_{\lbrace m_g(T,p_i,V_i)\rbrace}\\[6pt]
\operatorname{relative-support}_g : & (\mathcal{P}^n \times \mathcal{V}^n \times \mathcal{P}^n) 
 \longrightarrow [0,1] \cap \mathbb{Q}, \\[6pt]
 & (p_i,V_i;p_j) \longmapsto \dfrac{|\lbrace ~T \in \mathcal{T}^n_{g,p_i,p_j} \mid m_g(T,p_i,V_i)~\rbrace|}{|\mathcal{T}^n_{g,p_i,p_j}|}\\[8pt]
 & = \dfrac{\sum_{T \in \mathcal{T}^n_{g,p_i,p_j}}\mathbb{1}_{\lbrace m_g(T,p_i,V_i)\rbrace}}{|\mathcal{T}^n_{g,p_i,p_j}|}\\[8pt]
 & = \hat{P}_g(E_{g,p_i,V_i;p_j})
\end{array}
\]
Given the set $\fulltau$ of term sequences that, in principle, could characterize correspondences between two mappings $\mu_1 \in \Omega^{nra}_{p_1,g}$ and $\mu_2 \in \Omega^{nra}_{p_2,g}$, for a (pattern, variable sequence) tuple $(p_i,V_i)$, absolute support for a single graph $g$ is defined as the number of these  term sequences for which $m_g(T,p_i,V_i)$ is true.

Given the set $\fulltau$ of term sequences that, in principle, could characterize correspondences between two mappings $\mu_1 \in \Omega^{nra}_{p_1,g}$ and $\mu_2 \in \Omega^{nra}_{p_2,g}$, for a (pattern, variable sequence) tuple $(p_i,V_i)$, 
relative support for a single graph $g$ is defined as the number of these  term sequences for which $m_g(T,p_i,V_i)$ is true to the  number of term sequences that could in principle characterize correspondences between two mappings $\mu_1$ and $\mu_2$.

The support metrics are anti-monotonic. Given a (pattern, variable sequence) tuple $(p_i,V_i)$, when we extend $p_i$ to $p_i'$ by adding one or more triple patterns, then for each variable $v \in V_i$ the additional triples might add (but never remove) constraints about $v$ -- thus, it cannot be the case that there exists a term sequence $T$ for which a mapping $\mu' \in \Omega^{nra}_{p_i',g}$ exists such that $\mu'(V_i) = T$ but no mapping $\mu \in \Omega^{nra}_{p_i,g}$ exists such that $\mu(V_i) = T$. However, it can be the case that $|\Omega^{nra}_{p_i}| \leq |\Omega^{nra}_{p_i',g}|$, but that does not make the support metrics not anti-monotonic.\\

The metrics confidence, lift, leverage, and conviction for the single-graph setting are defined analogously to their ISAR counterparts. We define confidence as follows:
\begin{align}\allowdisplaybreaks[4]
\operatorname{confidence}_g ~:~ & (\mathcal{P}^n \times \mathcal{P}^n \times \mathcal{V}^n \times \mathcal{V}^n) \longrightarrow [0,1] \cap \mathbb{Q}, \nonumber\\[6pt]
& (p_1,p_2,V_1,V_2) \longmapsto \dfrac{|\lbrace ~T \in \mathcal{T}^n_{g,p_1,p_2} \mid m_g(T,p_1,V_1) \land m_g(T,p_2,V_2)~\rbrace|}{|\lbrace ~T \in \mathcal{T}^n_{g,p_1,p_2} \mid m_g(T,p_1,V_1)~\rbrace|}\nonumber\\
& = \dfrac{\displaystyle\sum_{T \in \mathcal{T}^n_{g,p_1,p_2}}\mathbb{1}_{\lbrace m_g(T,p_1,V_1) \land m_g(T,p_2,V_2)\rbrace}}{\displaystyle\sum_{T \in \mathcal{T}^n_{g,p_1,p_2}}\mathbb{1}_{\lbrace m_g(T,p_1,V_1) \rbrace}}\nonumber\\
& = \hat{P}_g(E_{g,p_2,V_2;p_1} \mid E_{g,p_1,V_1;p_2})\nonumber
\end{align}
The score of $\operatorname{confidence}_g$ is undefined if $p_1$ does not match the graph $g$ (i.e., $\Omega^{nra}_{p_1,g} = \varnothing$, or, equivalently, $\hat{P}_g(E_{g,p_1,V_1;p_2})= 0$). We define lift as follows:
\begin{align}\allowdisplaybreaks[4]
\operatorname{lift}_g ~:~ & (\mathcal{P}^n \times \mathcal{P}^n \times \mathcal{V}^n \times \mathcal{V}^n) \longrightarrow [0,\infty) \cap \mathbb{Q}, \nonumber\\[2pt]
& (p_1,p_2,V_1,V_2) \longmapsto \dfrac{\operatorname{confidence}_g(p_1,p_2,V_1,V_2)}{\operatorname{relative-support}_g(p_2,V_2;p_1)} = \dfrac{\hat{P}_g(E_{g,p_2,V_2;p_1} \mid E_{g,p_1,V_1;p_2})}{\hat{P}_g(E_{g,p_2,V_2;p_1})}\nonumber\\[4pt]
& = \dfrac{|\lbrace \,T \in \fulltau \mid m_g(T,p_1,V_1) \land m_g(T,p_2,V_2)\,\rbrace|\cdot |\fulltau|}{|\lbrace \,T \in \fulltau \mid m_g(T,p_1,V_1)\,\rbrace|\cdot |\lbrace \, T \in \fulltau \mid m_g(T,p_2,V_2) \, \rbrace|}\nonumber\\[4pt]
& = \dfrac{\left[\indicatorboth\right] \cdot |\fulltau|}{ \left[ \indicatorleft \right] \cdot \left[ \indicatorright \right] }\nonumber
\end{align}
The score of $\operatorname{lift}_g$ is undefined if $p_1$ or $p_2$ does not match the graph $g$ (i.e., $\Omega^{nra}_{p_1,g} = \varnothing$ or $\Omega^{nra}_{p_2,g} = \varnothing$, or, equivalently, $\hat{P}_g(E_{g,p_1,V_1;p_2})= 0$ or $\hat{P}_g(E_{g,p_2,V_2;p_1})= 0$). We define leverage as follows:
\begin{align}\allowdisplaybreaks[4]
\operatorname{leverage}_g ~:~ & (\mathcal{P}^n \times \mathcal{P}^n \times \mathcal{V}^n \times \mathcal{V}^n) \longrightarrow [-\tfrac{1}{4},\tfrac{1}{4}] \cap \mathbb{Q}, \nonumber\\[6pt]
& (p_1,p_2,V_1,V_2) \longmapsto \dfrac{|\lbrace T \in \fulltau \mid M_g(T,p_1,V_1 \land m_g(T,p_2,V_2)\rbrace|}{|\fulltau|} - \nonumber\\[8pt]
& ~~~~\dfrac{|\lbrace \,T \in \fulltau ~|~ m_g(T,p_1,V_1)\,\rbrace|}{|\fulltau|}\cdot \dfrac{|\lbrace \, T \in \fulltau \mid m_g(T,p_2,V_2)\,\rbrace|}{|\fulltau|}\nonumber\\[8pt]
& = \dfrac{\indicatorboth}{|\fulltau|} - \nonumber\\[6pt]
&~~~~\frac{\indicatorleft}{|\fulltau|} \cdot \frac{\indicatorright}{|\fulltau|}\nonumber\\[6pt]
& = \hat{P}_g(E_{g,p_1,V_1;p_2} \cap E_{g,p_2,V_2;p_1}) - \hat{P}_g(E_{g,p_2,V_2;p_1}) \cdot \hat{P}_g(E_{g,p_2,V_2;p_1})\nonumber
\end{align}
The score of $\operatorname{leverage}_g$ is undefined if $\fulltau = \varnothing$, which means that there cannot be a match of both patterns on $g$. We define conviction as follows:
\begin{align}\allowdisplaybreaks[4]
\operatorname{conviction}_g ~:~ & (\mathcal{P}^n \times \mathcal{P}^n \times \mathcal{V}^n \times \mathcal{V}^n) \longrightarrow \mathbb{Q}_{> 0} \cup \lbrace \infty \rbrace, \nonumber\\[6pt]
& (p_1,p_2,V_1,V_2) \longmapsto \dfrac{\hat{P}_g(E_{g,p_1,V_1;p_2}) \cdot(1-\hat{P}_g(E_{g,p_2,V_2;p_1}))}{1 - \hat{P}_g(E_{g,p_2,V_2;p_1} \mid E_{g,p_1,V_1;p_2})}\nonumber\\[6pt]
& \hspace{-0.9cm}= \dfrac{\left(\sum_{T \in \fulltau} \mathbb{1}_{\left\lbrace m_g(T,p_1,V_1) \right\rbrace}\right)^2 {\cdot} \left(|\fulltau| {-} \sum_{T \in \fulltau} \mathbb{1}_{\left\lbrace m_g(T,p_2,V_2) \right\rbrace}\right)}{|\fulltau|^2{\cdot}\! \left(\left[\sum_{T \in \fulltau} \mathbb{1}_{\left\lbrace m_g(T,p_1,V_1) \right\rbrace}\right]\!{-}\!\left[\sum_{T \in \fulltau} \mathbb{1}_{\left\lbrace m_g(T,p_1,V_1) \land m_g(T,p_2,V_2) \right\rbrace}\right]\right)}\nonumber
\end{align}
In the case that $\hat{P}_g(E_{g,p_2,V_2;p_1} \mid E_{g,p_2,V_2;p_1}) = 1$ (i.e., $\operatorname{confidence}_g(p_1,p_2,V_1,V_2) = 1$), then $\operatorname{conviction}_g(p_1,p_2,V_1,V_2) \to \infty$: since the denominator approaches zero from the positive side, the fraction diverges to infinity. We define $\operatorname{conviction}_g(p_1,p_2,V_1,V_2)$ to \textit{be} $\infty$ in the case that $\operatorname{confidence}_g(p_1,p_2,V_1,V_2) = 1$. Conviction is undefined if $\hat{P}_g(E_{g,p_2,V_2;p_1}) = 1$.\\

In practice, one derives graph patterns from a non-empty graph and constructs rules from graph patterns with a non-zero support. Thus, 
the situation does not occur that 
the probability of the antecedent (i.e., $\hat{P}_g(E_{g,p_1,V_1;p_2})$) or the probability of the consequent
(i.e., $\hat{P}_g(E_{g,p_2,V_2;p_1})$) being zero 
leading to the confidence score, the lift score, or the leverage score to be undefined does not occur.

\subsection{Metrics for the Graph-transactional Setting}\label{sec:metricstransactional}

We have two equivalent ways to define metrics for the graph transactional setting. As an example we take the confidence metric. The first option is to sum over the individual probabilities per graph, for example:
\[ \operatorname{confidence}_G(p_1,p_2,V_1,V_2) = \dfrac{\sum_{g \in G} \hat{P}_g(E_{g,p_1,V_1;p_2} \cap E_{g,p_2,V_2;p_1})}{\sum_{g \in G} \hat{P}_g(E_{g,p_1,V_1;p_2})}\]
The second option consists in adapting the definitions of the events and define the probabilities $\hat{P}_G(E_{G,p_2,V_2;p_1})$ and $\hat{P}_G(E_{G,p_1,V_1;p_2} \cap E_{G,p_2,V_2;p_1})$ and define confidence as follows:
\[ \operatorname{confidence}_G(p_1,p_2,V_1,V_2) = \dfrac{\hat{P}_G(E_{G,p_1,V_1;p_2} \cap E_{G,p_2,V_2;p_1})}{\hat{P}_G(E_{G,p_1,V_1;p_2})}\]
To define the event $E_{G,p_1,V_1;p_2}$, we need 
to adapt the sample space from $\fulltau$ to
$\mathcal{T}^n_{G,p_1,p_2}$, which we define as follows:
\[ \mathcal{T}^n_{G,p_1,p_2} := \lbrace \,(i,T) \mid g_i \in G \land T \in \mathcal{T}^n_{g_i,p_1,p_2} \,\rbrace\]
Here, $i$ is the identifier of a graph $g_i \in G$. Then, we can define the event $E_{G,p_1,V_1;p_2}$:
\[ E_{G,p_1,V_1;p_2} = \lbrace \, (i,T) \in \mathcal{T}^n_{G,p_1,p_2} \mid m_{g_i}(T,p_1,V_1) \, \rbrace\]
The event $E_{G,p_2,V_2;p_1}$ is defined equivalently. We define the probability of the event as follows:
\[ \hat{P}(E_{G,p_1,V_1;p_2}) = \dfrac{|\lbrace \, (i,T) \in \mathcal{T}^n_{G,p_1,p_2} \mid m_{g_i}(T,p_1,V_1) \, \rbrace|}{|\mathcal{T}^n_{G,p_1,p_2}|}\]
The probability of the event $E_{G,p_2,V_2;p_1}$ is defined equivalently. Then, we can define relative support as follows:
\[
\begin{array}{rl}
\operatorname{relative-support}_G : & (\mathcal{P}^n \times \mathcal{V}^n \times \mathcal{P}^n) 
 \longrightarrow [0,1] \cap \mathbb{Q}, \\[6pt]
 & (p_i,V_i;p_j) \longmapsto \hat{P}(E_{G,p_i,V_i;p_j})
\end{array}
\]
We define the empirical joint and conditional probabilities as follows: 
\[
\begin{array}{rl}
\hat{P}_G(E_{G,p_1,V_1;p_2} \cap E_{G,p_2,V_2;p_1}) =& \dfrac{|\lbrace \, (i,T) \in \fulltauG \mid m_{g_i}(T,p_1,V_1) \land m_{g_i}(T,p_2,V_2)\,\rbrace|}{|\fulltauG|}\\[10pt]
\hat{P}_G(E_{G,p_2,V_2;p_1} \mid E_{G,p_1,V_1;p_1}) =& \dfrac{|\lbrace \, (i,T) \in \fulltauG \mid m_{g_i}(T,p_1,V_1) \land m_{g_i}(T,p_2,V_2) \,\rbrace|}{|\lbrace \, (i,T) \in \fulltauG \mid m_{g_i}(T,p_1,V_1) \,\rbrace|}
\end{array}
\]
$\hat{P}_G(E_{G,p_1,V_1;p_2})$ denotes the empirical probability, estimated from the bag $G$ of graphs, that for a randomly selected 
pair consisting of the identifier $i$ of a graph and a term sequence $T$ (that could in principle correspond to a match of $p_1$ and to a match of $p_2$ on the graph with the identifier $i$) it is the case that the term sequence satisfies $m_{g_i}(T,p_1,V_1)$; which means that a match of $p_1$ on $g_i$ (i.e., $\mu \in \Omega^{nra}_{p_1,g_i}$) exists that is characterized by $T$ (i.e.,  $\mu(V_1) = T$). $\hat{P}_G(E_{G,p_2,V_2})$ can be interpreted equivalently. 
$\hat{P}_G(E_{G,p_1,V_1;p_2} \cap E_{G,p_2,V_2;p_1})$ denotes the empirical probability, estimated from the bag $G$ of graphs, that for a randomly selected pair consisting of the identifier $i$ of a graph and a term sequence $T$ (that could in principle correspond to a match of $p_1$ and to a match of $p_2$ on the graph with the identifier $i$) it is the case that the term sequence satisfies both $m_{g_i}(T,p_1,V_1)$ and $m_{g_i}(T,p_2,V_2)$. 
Finally, $\hat{P}_G(E_{G,p_2,V_2;p_1} \mid E_{G,p_1,V_1;p_2})$ denotes the empirical conditional probability, estimated from the bag $G$ of graphs, 
that for a randomly selected pair consisting of the identifier $i$ of a graph and a term sequence $T$ (that could in principle correspond to a match of $p_1$ and to a match of $p_2$ on the graph with the identifier $i$) it is the case that the term sequence satisfies $m_{g_i}(T,p_2,V_2)$ under the condition that it satisfies $m_{g_i}(T,p_1,V_1)$.\\

We formulate a probability space $(\Omega, \mathcal{F}, \hat{P}_G)$ with $\Omega = \fulltauG$, $\mathcal{F} = \mathcal{P}(\Omega)$, and the empirical probability measure $\hat{P}_G$ defined above. It is trivial to show that the Kolmogorov axioms are satisfied in this case, too. 
Based on the events $E_{G,p_1,V_1;p_2}$ and $E_{G,p_2,V_2;p_1}$ we can define the following metrics for the graph-transactional setting:
\begin{align*}
\operatorname{confidence}_G ~:~ & (\mathcal{P}^n \times \mathcal{P}^n \times \mathcal{V}^n \times \mathcal{V}^n) \longrightarrow [0,1] \cap \mathbb{Q}, \\[6pt]
& (p_1,p_2,V_1,V_2) \longmapsto \hat{P}_G(E_{G,p_2,V_2;p_1} \mid E_{G,p_1,V_1;p_2})\\
& = \dfrac{|\lbrace \, (i,T) \in \fulltauG \mid m_{g_i}(T,p_1,V_1) \land m_{g_i}(T,p_2,V_2) \,\rbrace|}{|\lbrace \, (i,T) \in \fulltauG \mid m_{g_i}(T,p_1,V_1) \,\rbrace|}
\end{align*}
The score of $\operatorname{confidence}_G$ is undefined if $p_1$ matches none of the graphs in $G$ (i.e., $\forall g \in G : \Omega^{nra}_{p_2,g} = \varnothing$, or, equivalently, $\forall g \in G : \hat{P}_g(E_{g,p_1,V_1;p_2}) = 0$).
\begin{align*}
\operatorname{lift}_G ~:~ & (\mathcal{P}^n \times \mathcal{P}^n \times \mathcal{V}^n \times \mathcal{V}^n) \longrightarrow [0,\infty) \cap \mathbb{Q}, \\[2pt]
& (p_1,p_2,V_1,V_2) \longmapsto \dfrac{\hat{P}_G(E_{G,p_2,V_2;p_1} \mid E_{G,p_1,V_1;p_2})}{\hat{P}_G(E_{G,p_2,V_2;p_1})}\\
& = \dfrac{|\lbrace \, (i,T) \in \fulltauG \mid m_{g_i}(T,p_1,V_1) \land m_{g_i}(T,p_2,V_2) \,\rbrace|}{|\lbrace \, (i,T) \in \fulltauG \mid m_{g_i}(T,p_1,V_1) \,\rbrace| \cdot |\lbrace \, (i,T) \in \fulltauG \mid m_{g_i}(T,p_2,V_2) \,\rbrace|}
\end{align*}
The score of $\operatorname{lift}_G$ is undefined if $p_1$ matches none of the graphs in $G$ or if $p_2$ matches none of the graphs in $G$ (i.e., $\forall g \in G : \Omega^{nra}_{p_1,g} = \varnothing \lor \forall g \in G : \Omega^{nra}_{p_2,g} = \varnothing$, or, equivalently, $\forall g \in G : \hat{P}_g(E_{g,p_1,V_1;p_2}) = \hat{P}_g(E_{g,p_2,V_2;p_1}) = 0$).
\begin{align*}
\operatorname{leverage}_G ~:~ & (\mathcal{P}^n \times \mathcal{P}^n \times \mathcal{V}^n \times \mathcal{V}^n) \longrightarrow [-\tfrac{1}{4},\tfrac{1}{4}] \cap \mathbb{Q}, \\[6pt]
& \hspace{-0.2cm}(p_1,p_2,V_1,V_2) \longmapsto \hat{P}_G(E_{G,p_1,V_1;p_2} \cap E_{G,p_2,V_2;p_1}){-}\hat{P}_G(E_{G,p_2,V_2;p_1}) {\cdot} \hat{P}_G(E_{G,p_2,V_2;p_1})\\
& = \dfrac{|\lbrace \, (i,T) \in \fulltauG \mid m_{g_i}(T,p_1,V_1) \land m_{g_i}(T,p_2,V_2)\,\rbrace|}{|\fulltauG|} -\\
& \hspace{-0cm}\dfrac{|\lbrace \, (i,T) \in \fulltauG \mid m_{g_i}(T,p_1,V_1) \,\rbrace|}{|\fulltauG|} \cdot \dfrac{|\lbrace \, (i,T) \in \fulltauG \mid m_{g_i}(T,p_2,V_2) \,\rbrace|}{|\fulltauG|}
\end{align*}
The score of $\operatorname{leverage}_G$ is undefined if for all graphs $g \in G$ it is the case that $\fulltau = \varnothing$, which means that there cannot be a match of both patterns for any graph.
\begin{align*}
\operatorname{conviction}_G ~:~ & (\mathcal{P}^n \times \mathcal{P}^n \times \mathcal{V}^n \times \mathcal{V}^n) \longrightarrow \mathbb{Q}_{> 0} \cup \lbrace \infty \rbrace, \\[6pt]
& (p_1,p_2,V_1,V_2) \longmapsto \dfrac{\hat{P}_G(E_{G,p_1,V_1;p_2}) \cdot(1-\hat{P}_G(E_{G,p_2,V_2;p_1}))}{1 - \hat{P}_G(E_{G,p_2,V_2;p_1} \mid E_{G,p_1,V_1;p_2})}\\[6pt]
& \hspace{-2.5cm}= \dfrac{\left(\sum_{(i,T) \in \fulltauG} \mathbb{1}_{\left\lbrace m_{g_i}(T,p_1,V_1) \right\rbrace}\right)^2 {\cdot} \left(|\fulltauG| {-} \sum_{(i,T) \in \fulltauG} \mathbb{1}_{\left\lbrace m_{g_i}(T,p_2,V_2) \right\rbrace}\right)}{|\fulltauG|^2{\cdot}\! \left(\left[\sum_{(i,T) \in \fulltauG} \mathbb{1}_{\left\lbrace m_{g_i}(T,p_1,V_1) \right\rbrace}\right]\!{-}\!\left[\sum_{(i,T) \in \fulltauG} \mathbb{1}_{\left\lbrace m_{g_i}(T,p_1,V_1) \land m_{g_i}(T,p_2,V_2) \right\rbrace}\right]\right)}
\end{align*}
In the case that $\hat{P}_G(E_{G,p_2,V_2;p_1} \mid E_{G,p_2,V_2;p_1}) = 1$ (i.e., $\operatorname{confidence}_G(p_1,p_2,V_1,V_2) = 1$), then $\operatorname{conviction}_G(p_1,p_2,V_1,V_2) \to \infty$: since the denominator approaches zero from the positive side, the fraction diverges to infinity. We define $\operatorname{conviction}_G(p_1,p_2,V_1,V_2)$ to \textit{be} $\infty$ in the case that $\operatorname{confidence}_G(p_1,p_2,V_1,V_2) = 1$. Conviction is undefined if $\hat{P}_G(E_{G,p_2,V_2;p_1}) = 1$\\

Note that these metrics are micro-averages. Thus, larger graphs can have a stronger impact on a score than small graphs. For situations where this is undesired, we introduce macro-averaged metrics, beginning with macro-averaged relative support and macro-averaged confidence:
\[
\begin{array}{rl}
\operatorname{macro-relative-support}_G : & (\mathcal{P}^n \times \mathcal{V}^n \times \mathcal{P}^n) 
 \longrightarrow [0,1] \cap \mathbb{Q}, \\[6pt]
 & (p_i,V_i;p_j) \longmapsto \dfrac{1}{|G|}\displaystyle \sum_{g \in G}\hat{P}(E_{g,p_i,V_i;p_j})
\\
\operatorname{macro-confidence}_G ~:~ & (\mathcal{P}^n \times \mathcal{P}^n \times \mathcal{V}^n \times \mathcal{V}^n) \longrightarrow [0,1] \cap \mathbb{Q}, \\[6pt]
& (p_1,p_2,V_1,V_2) \longmapsto \frac{1}{|G|}\sum_{g \in G}\operatorname{confidence}_g(p_1,p_2,V_1,V_2)
\end{array}
\]
The score of $\operatorname{macro-confidence}_G$ is undefined if there is a graph $g \in G$ that is not matched by $p_1$ (i.e., $\exists g \in G : \Omega^{nra}_{p_1,g} = \varnothing$).
Thus, given a bag $G$ of graphs where some graph is not matched by $p_1$, we need to reduce $G$ so that every graph in the reduced set satisfies that condition. If we have derived two rules $r_a=(p_{a,1},p_{a,2},V_{a,1},V_{a,2})$ and $r_b=(p_{b,1},p_{b,2},V_{b,1},V_{b,2})$ from a bag $G$ of graphs, then 
the result of reducing $G$ to those graphs $g$ that satisfy 
$ \hat{P}(E_{g,p_{a,1}, V_{a,1},p_{a,2}}) > 0$ 
can be different in size from 
the result of reducing $G$ to those graphs $g$ that satisfy 
$ \hat{P}(E_{g,p_{b,1}, V_{b,1},p_{b,2}}) > 0$. 
The consequence is that if $r_1$ has a higher lift score than $r_2$, the rule $r_2$ might still be more interesting because $G$ had to be reduced less for $r_2$ than for $r_1$. Thus, ranking rules by lift score only might be insufficient in practical applications.  

To quantify the fraction of the bag $G$ of graphs where a metric with requirements specific to the graphs is applicable we define the metric $\operatorname{degree-of-applicability}_G$ as follows:
\begin{equation*} \operatorname{degree-of-applicability}_G(c) = \frac{\sum_{g \in G} \mathbb{1}_{\lbrace c \rbrace}}{|G|}\end{equation*}
where $c$ is a boolean expression. In the case of confidence, we use 
$c = \hat{P}(E_{g,p_1, V_1,p_2}) > 0$ 
and evaluate macro-confidence on the reduced bag $G_c = \lbrace \, g \in G \mid c \, \rbrace$.

A simple method when ranking a set of graph pattern-based association rules is to carry out two-stage sorting where first it is sorted (descending) based on the $\operatorname{degree-of-applicability}$ score. Ties are broken based on, e.g., the $\operatorname{lift}$ score. 
This two-stage sorting approach is also known as lexicographic ordering.
The field of multi-criteria decision making (MCDM) provides further options beyond lexicographic ordering.
We define macro-averaged lift as follows:
\begin{align*}
\operatorname{macro-lift}_G ~:~ & (\mathcal{P}^n \times \mathcal{P}^n \times \mathcal{V}^n \times \mathcal{V}^n) \longrightarrow [0,\infty) \cap \mathbb{Q}, \\[2pt]
& (p_1,p_2,V_1,V_2) \longmapsto \frac{1}{|G|}\sum_{g \in G}\operatorname{lift}_g(p_1,p_2,V_1,V_2)
\end{align*}
The score of $\operatorname{macro-lift}_G$ is undefined if there exists a graph $g \in G$ where either $p_1$ or $p_2$ does not match (i.e., $\exists g \in G : \Omega^{nra}_{p_1,g} = \varnothing \lor \Omega^{nra}_{p_2,g} = \varnothing$). Thus, we might need to reduce $G$ with $c = \hat{P}(E_{g,p_1, V_1,p_2}) > 0 \land \hat{P}(E_{g,p_2, V_2,p_1}) > 0$ and evaluate macro-lift on the reduced bag $G_c$.
We define macro-averaged leverage as follows:
\begin{align*}
\operatorname{macro-leverage}_G ~:~ & (\mathcal{P}^n \times \mathcal{P}^n \times \mathcal{V}^n \times \mathcal{V}^n) \longrightarrow [-\tfrac{1}{4},\tfrac{1}{4}] \cap \mathbb{Q}, \\[6pt]
& (p_1,p_2,V_1,V_2) \longmapsto \frac{1}{|G|}\sum_{g \in G}\operatorname{leverage}_g(p_1,p_2,V_1,V_2)
\end{align*}
The score of $\operatorname{macro-leverage}_G$ is undefined if there exists a graph $g \in G$ where $\fulltau = \varnothing$, which means that there cannot be a match of both patterns on $g$ (i.e., $\exists g \in G : \fulltau = \varnothing$). Thus, we might need to reduce $G$ with $c = \fulltau \neq \varnothing$ and evaluate macro-leverage on the reduced bag~$G_c$.
We define macro-averaged conviction as follows:
\begin{align*}
\operatorname{macro-conviction}_G ~:~ & (\mathcal{P}^n \times \mathcal{P}^n \times \mathcal{V}^n \times \mathcal{V}^n) \longrightarrow \mathbb{Q}_{> 0} \cup \lbrace \infty \rbrace, \\[6pt]
& (p_1,p_2,V_1,V_2) \longmapsto \frac{1}{|G|}\sum_{g \in G}\operatorname{conviction}_g(p_1,p_2,V_1,V_2)
\end{align*}
We define the score of $\operatorname{conviction}_g$ to be $\infty$ in the case that $\hat{P}_G(E_{G,p_2,V_2;p_1} \mid E_{G,p_1,V_1;p_2}) = 1$ and, therefore, we need to define what the score of $\operatorname{macro-conviction}_G$ is in the case that conviction is $\infty$ for at least one $g \in G$. We treat $\infty$ as a value and let it behave as follows: $\infty + \infty = \infty$, and, for any $x \in \mathbb{Q}$, $\infty + x = \infty$. 
Alternatively, it can make sense to count the number of graphs in G for which $\operatorname{conviction}_g = \infty$ and then evaluate macro-conviction on the reduced bag $G_c$ where $c = \hat{P}_g(E_{g,p_2,V_2;p_1} \mid E_{g,p_1,V_1;p_2}) < 1 \land \hat{P}_g(E_{g,p_2,V_2;p_1}) < 1$.\\

It is important to note that when we macro-average the scores then there is no single probability space -- we are working with a family of probability spaces and the macro-average is an external aggregation of metrics, not something that corresponds to probabilities in one unified probability space. Thus, the the metrics are no longer pure functions of one measure $P$. This has consequences on which characteristics a macro-averaged metric shares with its corresponding micro-averaged metric and with its corresponding itemset-based metrics, as we will see in Section \ref{sec:characteristics}.

\subsection{Reframing to Classical Association Rule Metrics}\label{sec:reframing}

In principle we can reframe the problem of measuring properties of graph pattern-based association rules into the problem of measuring properties of itemset-based association rules -- both in the case where we have given a single graph and where we have given a bag of graphs and micro-averaged metrics. The method consists in generating a transaction database $T$ for a given graph pattern-based association rule $(p_1,p_2,V_1,V_2)$ and a bag $G$ of graphs (that could contain a single graph). We generate the bag $T$ (using bag union operations $\uplus$) as follows:
\[ T = \biguplus_{g \in G}~\biguplus_{X \in \fulltau} \left[ \lbrace \, A \mid m_g(X,p_1,V_1) \, \rbrace \cup \lbrace \, B \mid m_g(X,p_2,V_2) \, \rbrace \right]\]
Thus, for each graph $g \in G$, for each term sequence $X \in \fulltau$ we add a new transaction $t$ to $T$ %
where $t \in \lbrace \varnothing, \lbrace A \rbrace, \lbrace B \rbrace, \lbrace A,B \rbrace\rbrace$.
$A$ and $B$ are special symbols. In the context of a graph $g$ and a term sequence $X$, $A$ expresses that $m_g(X,p_1,V_1)$ is true, which means that there exists a match $\mu \in \Omega^{nra}_{p_1,g}$ such that $\mu(V_1)=X$. $B$ expresses that $m_g(X,p_2,V_2)$ is true, which means that there exists a match $\mu \in \Omega^{nra}_{p_2,g}$ such that $\mu(V_2)=X$.

Because for each $g \in G$ we add exactly $|\fulltau|$ transactions to the bag $T$, it follows that $|T| = |\fulltauG|$. 
Furthermore, the number of transactions in $T$ that contain $A$ is equal to $\sum_{g \in G} |\lbrace \, X \in \fulltau \mid m_g(X,p_1,V_1) \,\rbrace|$, the number of transactions in $T$ that contain $B$ is equal to $\sum_{g \in G} |\lbrace \, X \in \fulltau \mid m_g(X,p_2,V_2) \,\rbrace|$, and
the number of transactions in $T$ that contain $A$ and $B$ is equal to $\sum_{g \in G} |\lbrace \, X \in \fulltau \mid m_g(X,p_1,V_1) \land m_g(X,p_2,V_2) \,\rbrace|$. 
We have thus created a homomorphism, which allows us to evaluate classical metrics for itemset-based association rules on a classical transaction database where each transaction is an itemset instead of evaluating metrics for graph pattern-based association rules on a transaction database where each transaction is a graph.

Given the graph pattern-based association rule $(p_1,p_2,V_1,V_2)$, a bag $G$ of graphs, two special symbols $A$ and $B$ that correspond to the patterns $p_1$ and $p_2$, respectively, and the generated transaction database $T$, we define the following correspondences between events:
\begin{align*}
\hat{P}_G(E_{p_1,V_1;p_2}) ~=~ & \hat{P}_T(E_{\lbrace A \rbrace})\\
\hat{P}_G(E_{p_2,V_2;p_1}) ~=~ & \hat{P}_T(E_{\lbrace B \rbrace})\\
\hat{P}_G(E_{p_1,V_1;p_2} \cap E_{p_2,V_2;p_1}) ~=~ & \hat{P}_T(E_{\lbrace A, B \rbrace})
\end{align*}
From these correspondences follow further correspondences:
\begin{align*}
\operatorname{relative-support}_G(p_1,V_1;p_2) ~=~ & \hat{P}_T(E_{\lbrace A \rbrace}) \\
\operatorname{relative-support}_G(p_2,V_2;p_1) ~=~ & \hat{P}_T(E_{\lbrace B \rbrace}) \\
\operatorname{conf.}_G(p_1,p_2,V_1,V_2) = \hat{P}_G(E_{p_2,V_2;p_1} \mid E_{p_1,V_1;p_2}) ~=~ & \hat{P}_T(E_{\lbrace B \rbrace} \mid E_{\lbrace A \rbrace}) = \operatorname{conf.}_T(\lbrace A \rbrace \Rightarrow \lbrace B \rbrace)\\
\operatorname{lift}_G(p_1,p_2,V_1,V_2) ~=~ & \operatorname{lift}_T(\lbrace A \rbrace \Rightarrow \lbrace B \rbrace)\\
\operatorname{leverage}_G(p_1,p_2,V_1,V_2) ~=~ & \operatorname{leverage}_T(\lbrace A \rbrace \Rightarrow \lbrace B \rbrace)\\
\operatorname{conviction}_G(p_1,p_2,V_1,V_2) ~=~ & \operatorname{conviction}_T(\lbrace A \rbrace \Rightarrow \lbrace B \rbrace)\\
\end{align*}

From the existence of the homomorphism we can conclude that the micro-averaged metrics for graph pattern-based association rules are generalizations of the metrics for itemset-based association rules.

Although this homomorphism is interesting from a theoretical point of view, exploiting the homomorphism and actually generating itemset-based transaction databases is not practical for a couple of reasons: i) the set $\fulltau$ can be very large, see Eq. \ref{cardinality_fulltau}. It is expected that usually most term sequences do not correspond to any match. Thus, the generated database would be large and most of its transactions would be empty. Furthermore, ii) for each rule another transaction database needs to be generated, resulting in high computational costs and high memory costs when the metrics are evaluated for a set of rules.\\
 
The transaction database generation procedure described above can only be applied for the micro-averaged metrics. The situation is different for the macro-averaged metrics. Here, one would need to distinguish transactions generated for different graphs. A simple way would be to generate one transaction database per graph, evaluate a metric, and then create the average of the obtained scores.

\subsection{Characteristics of the Metrics}\label{sec:characteristics}
With characteristics of metrics we refer to the scores of a metric obtained in especially interesting situations. In Table \ref{tab:situation-names} we introduce names to refer to the situations that are relevant when comparing and interpreting the scores of these metrics. These situations are described differently, depending on whether we are dealing with a situation in the context of an itemset-based association rule (itemset setting), in the context of a graph pattern-based association rule and a single graph (single graph setting), in the context of a graph pattern-based association rule and a bag of graphs where scores are micro-averaged (graph transactional setting, micro-averaged), or in the context of a graph pattern-based association rule and a bag of graphs where scores are macro-averaged (graph transactional setting, macro-averaged).

\begin{longtable}{>{\raggedright\arraybackslash}p{0.93cm} p{13cm}}
\caption{Names and descriptions of situations that are relevant when describing and comparing the scores of ISAR and GPAR metrics.\label{tab:situation-names}}\\
\toprule
\textbf{Name} & \textbf{Situation Description} \\
\midrule
\endfirsthead
\bottomrule
\endfoot

\texttt{IDE} &
The events $E_A$ and $E_B$ are identical, i.e., $E_A = E_B$. That means, for every transaction $t \in T$, when $A \subseteq t$ is satisfied then also $B \subseteq t$ is satisfied and vice versa %
(itemset setting).%

The events $E_{g,p_1,V_1;p_2}$ and $E_{g,p_2,V_2;p_1}$ are identical, i.e., $E_{g,p_1,V_1;p_2} = E_{g,p_2,V_2;p_1}$. That means, for every term sequence $T \in \fulltau$, 
when $m_g(T,p_1,V_1)$ is satisfied then also $m_g(T,p_2,V_2)$ is satisfied 
and vice versa %
(single graph setting).%

The events $E_{G,p_1,V_1;p_2}$ and $E_{G,p_2,V_2;p_1}$ are identical, i.e., $E_{G,p_1,V_1;p_2} = E_{G,p_2,V_2;p_1}$. That means, for every pair $(i,T) \in \fulltauG$,
when $m_{g_i}(T,p_1,V_1)$ is satisfied then also $m_{g_i}(T,p_2,V_2)$ is satisfied 
and vice versa %
(graph transactional setting, micro-averaged).

For every graph $g \in G$ it is the case that $E_{g,p_1,V_1;p_2}$ and $E_{g,p_2,V_2;p_1}$ are identical (graph transactional setting, macro-averaged).
\\

\texttt{DIS} &
The events $E_A$ and $E_B$ are disjoint, i.e., $E_A \cap E_B = \varnothing$. That means, for every transaction $t \in T$, when $A \subseteq t$ is satisfied then $B \subseteq t$ is not satisfied and vice versa %
(itemset setting).%

The events $E_{g,p_1,V_1;p_2}$ and $E_{g,p_2,V_2;p_1}$ are disjoint, i.e., $E_{g,p_1,V_1;p_2} \cap E_{g,p_2,V_2;p_1} = \varnothing$. That means, for every term sequence $T \in \fulltau$, 
when $m_g(T,p_1,V_1)$ is satisfied then $m_g(T,p_2,V_2)$ is not satisfied 
and vice versa %
(single graph setting).%

The events $E_{G,p_1,V_1;p_2}$ and $E_{G,p_2,V_2;p_1}$ are disjoint, i.e., $E_{G,p_1,V_1;p_2} \cap E_{G,p_2,V_2;p_1} = \varnothing$. That means, for every pair $(i,T) \in \fulltauG$,
when $m_{g_i}(T,p_1,V_1)$ is satisfied then also $m_{g_i}(T,p_2,V_2)$ is satisfied 
and vice versa %
(graph transactional setting, micro-averaged).

For every graph $g \in G$ it is the case that $E_{g,p_1,V_1;p_2}$ and $E_{g,p_2,V_2;p_1}$ are disjoint (graph transactional setting, macro-averaged).
\\

\texttt{IND} & 
The events $E_A$ and $E_B$ are independent, i.e., $\hat{P}_T(E_A \cap E_B) = \hat{P}_T(E_A) \cdot \hat{P}_T(E_B)$. That means, for every transaction $t \in T$, when $A \subseteq t$ is satisfied then this does not change the probability that $B \subseteq t$ is satisfied 
and vice versa (itemset setting).%

The events $E_{g,p_1,V_1;p_2}$ and $E_{g,p_2,V_2;p_1}$ are independent, i.e., $\hat{P}_T(E_{g,p_1,V_1;p_2} \cap E_{g,p_2,V_2;p_1}) = \hat{P}_T(E_{g,p_1,V_1;p_2}) \cdot \hat{P}_T(E_{g,p_2,V_2;p_1})$. That means, for every term sequence $T \in \fulltau$, 
when $m_g(T,p_1,V_1)$ is satisfied then this does not change the probability that $m_g(T,p_2,V_2)$ is satisfied and vice versa (single graph setting).%

The events $E_{G,p_1,V_1;p_2}$ and $E_{G,p_2,V_2;p_1}$ are independent, i.e., $\hat{P}_T(E_{G,p_1,V_1;p_2} \cap E_{G,p_2,V_2;p_1}) = \hat{P}_T(E_{G,p_1,V_1;p_2}) \cdot \hat{P}_T(E_{G,p_2,V_2;p_1})$. That means, for every pair $(i,T) \in \fulltauG$, when $m_{g_i}(T,p_1,V_1)$ is satisfied then this does not change the probability that $m_{g_i}(T,p_2,V_2)$ is satisfied and vice versa (graph transactional setting, micro-averaged).

For every graph $g \in G$ it is the case that $E_{g,p_1,V_1;p_2}$ and $E_{g,p_2,V_2;p_1}$ are independent (graph transactional setting, macro-averaged).\\

\texttt{POS} & 
The events $E_A$ and $E_B$ are positively correlated, i.e., $\hat{P}_T(E_A \cap E_B) > \hat{P}(E_A)\cdot \hat{P}(E_B)$. That means, for every transaction $t \in T$, when $A \subseteq t$ is satisfied then this increases the probability that $B \subseteq t$ is satisfied and vice versa (itemset setting).%

The events $E_{g,p_1,V_1;p_2}$ and $E_{g,p_2,V_2;p_1}$ are positively correlated, i.e., $\hat{P}_g(E_{g,p_1,V_1;p_2} \cap E_{g,p_2,V_2;p_1}) > \hat{P}_g(E_{g,p_1,V_1;p_2}) \cdot \hat{P}_g(E_{g,p_2,V_2;p_1})$. That means, for every term sequence $T \in \fulltau$, when $m_g(T,p_1,V_1)$ is satisfied then this increases the probability that $m_g(T,p_2,V_2)$ is satisfied and vice versa (single graph setting).%

The events $E_{G,p_1,V_1;p_2}$ and $E_{G,p_2,V_2;p_1}$ are positively correlated, i.e., $\hat{P}_G(E_{G,p_1,V_1;p_2} \cap E_{G,p_2,V_2;p_1}) > \hat{P}_G(E_{G,p_1,V_1;p_2}) \cdot \hat{P}_G(E_{G,p_2,V_2;p_1})$. That means, for every pair $(i,T) \in \fulltauG$, when $m_{g_i}(T,p_1,V_1)$ is satisfied then this increases the probability that $m_{g_i}(T,p_2,V_2)$ is satisfied and vice versa (graph transactional setting, micro-averaged).

For every graph $g \in G$ it is the case that $E_{g,p_1,V_1;p_2}$ and $E_{g,p_2,V_2;p_1}$ are positively correlated (graph transactional setting, macro-averaged).
\\

\texttt{NEG} & 
The events $E_A$ and $E_B$ are negatively correlated, i.e., $\hat{P}_T(E_A \cap E_B) < \hat{P}(E_A)\cdot \hat{P}(E_B)$. That means, for every transaction $t \in T$, when $A \subseteq t$ is satisfied then this decreases the probability that $B \subseteq t$ is satisfied and vice versa (itemset setting).%

The events $E_{g,p_1,V_1;p_2}$ and $E_{g,p_2,V_2;p_1}$ are negatively correlated, i.e., $\hat{P}_g(E_{g,p_1,V_1;p_2} \cap E_{g,p_2,V_2;p_1}) < \hat{P}_g(E_{g,p_1,V_1;p_2}) \cdot \hat{P}_g(E_{g,p_2,V_2;p_1})$. That means, for every term sequence $T \in \fulltau$, when $m_g(T,p_1,V_1)$ is satisfied then this decreases the probability that $m_g(T,p_2,V_2)$ is satisfied and vice versa (single graph setting).%

The events $E_{G,p_1,V_1;p_2}$ and $E_{G,p_2,V_2;p_1}$ are negatively correlated, i.e., $\hat{P}_G(E_{G,p_1,V_1;p_2} \cap E_{G,p_2,V_2;p_1}) < \hat{P}_G(E_{G,p_1,V_1;p_2}) \cdot \hat{P}_G(E_{G,p_2,V_2;p_1})$. That means, for every pair $(i,T) \in \fulltauG$, when $m_{g_i}(T,p_1,V_1)$ is satisfied then this decreases the probability that $m_{g_i}(T,p_2,V_2)$ is satisfied and vice versa (graph transactional setting, micro-averaged).
For every graph $g \in G$ it is the case that $E_{g,p_1,V_1;p_2}$ and $E_{g,p_2,V_2;p_1}$ are negatively correlated (graph transactional setting, macro-averaged).
\\
\end{longtable}

Tables \ref{properties-confidence}--\ref{properties-conviction} show the key characteristics of the metrics confidence, lift, leverage, and conviction, respectively, for itemset-based association rules, denoted by ISAR, and for graph pattern-based association for the single graph setting, denoted by GPAR (single), for the graph transactional setting where scores are micro-averaged, denoted by GPAR (micro), and
for the graph transactional setting where scores are macro-averaged, denoted by GPAR (macro), for each of the situations \texttt{IDE}, \texttt{DIS}, \texttt{IND}, \texttt{POS}, and \texttt{NEG} (which are explained in Table \ref{tab:situation-names}). For each score, the number of the respective proposition is provided for which a proof is given in the appendix. 
For example, Table \ref{properties-confidence} expresses that the score of $\operatorname{confidence}_G(p_1,p_2,V_1,V_2)$ (i.e., GPAR (micro) context) is $0$ if and only if ($\Longleftrightarrow$) the \texttt{DIS} situation is given (i.e., $E_{g,p_1,V_1;p_2} \cap E_{g,p_2,V_2;p_1} = \varnothing$)
(proposition P\ref{proposition-confidence-DIS-ISAR-GPARsinglemicro}). 
Furthermore, Table \ref{properties-confidence} expresses that the score of $\operatorname{confidence}_g(p_1,p_2,V_1,V_2)$ (i.e., GPAR (single) context) is $1$ if ($\Longrightarrow$, but not \textit{if and only if}) the \texttt{IDE} situation is given (i.e., $E_{g,p_1,V_1;p_2} = E_{g,p_2,V_2;p_1}$)
(proposition P\ref{proposition-confidence-IDE-ISAR-GPARsinglemicro}).

\begin{table}
\caption{Key characteristics of the confidence metrics, assuming that $\hat{P}_T(E_B) > 0$, $\hat{P}_g(E_{g,p_1,V_1;p_2}) > 0$, $\hat{P}_G(E_{G,p_1,V_1;p_2}) > 0$, or $\forall g \in G : \hat{P}_g(E_{g,p_1,V_1;p_2}) > 0$, respectively. All four confidence metrics have the co-domain $[0,1] \cap \mathbb{Q}$.}
\label{properties-confidence}
\renewcommand{\arraystretch}{0.8}
\setlength{\tabcolsep}{4pt}

\begin{tabularx}{\textwidth}{
  L{0.04\textwidth} %
  L{0.38\textwidth} %
  L{0.05\textwidth} %
  L{0.40\textwidth} %
  L{0.06\textwidth} %
}
\toprule
\textbf{Sit.} & 
\textbf{ISAR} &
\textbf{} &
\textbf{GPAR (single)} &
\textbf{}\\
\cmidrule[1pt](lr){2-3} \cmidrule[1pt](lr){4-5}

\texttt{IDE} & 
$\arrowbox{\Longrightarrow} =1$ &
(P\ref{proposition-confidence-IDE-ISAR-GPARsinglemicro})
&
$\arrowbox{\Longrightarrow} =1$ &
(P\ref{proposition-confidence-IDE-ISAR-GPARsinglemicro})\\

\texttt{DIS} & 
$\Longleftrightarrow~ =0$ &
(P\ref{proposition-confidence-DIS-ISAR-GPARsinglemicro})
&
$\Longleftrightarrow~ =0$ &
(P\ref{proposition-confidence-DIS-ISAR-GPARsinglemicro})\\

\texttt{IND} & 
$\Longleftrightarrow~ = \hat{P}_T(E_B)$ &
(P\ref{proposition-confidence-IND-ISAR-GPARsinglemicro})
&
$\Longleftrightarrow~ = \hat{P}_g(E_{g,p_2,V_2;p_1})$ &
(P\ref{proposition-confidence-IND-ISAR-GPARsinglemicro})\\

\texttt{POS} & 
$\Longleftrightarrow~ > \hat{P}_T(E_B)$ &
(P\ref{proposition-confidence-POS-ISAR-GPARsinglemicro})
&
$\Longleftrightarrow~ > \hat{P}_g(E_{g,p_2,V_2;p_1})$ &
(P\ref{proposition-confidence-POS-ISAR-GPARsinglemicro})\\

\texttt{NEG} & 
$\Longleftrightarrow~ <  \hat{P}_T(E_B)$ &
(P\ref{proposition-confidence-NEG-ISAR-GPARsinglemicro})
&
$\Longleftrightarrow~ < \hat{P}_g(E_{g,p_2,V_2;p_1})$ &
(P\ref{proposition-confidence-NEG-ISAR-GPARsinglemicro})\\

\cmidrule[1pt](lr){1-5}
\textbf{Sit.} & 
\textbf{GPAR (micro)} &
\textbf{} &
\textbf{GPAR (macro)} &
\textbf{}\\
\cmidrule[1pt](lr){2-3} \cmidrule[1pt](lr){4-5}

\texttt{IDE} & 
$\arrowbox{\Longrightarrow} =1$ &
(P\ref{proposition-confidence-IDE-ISAR-GPARsinglemicro})
&
$\arrowbox{\Longrightarrow} =1$ &
(P\ref{proposition-confidence-IDE-GPARmacro})\\

\texttt{DIS} & 
$\Longleftrightarrow~ =0$ &
(P\ref{proposition-confidence-DIS-ISAR-GPARsinglemicro})
&
$\Longleftrightarrow~ =0$ &
(P\ref{proposition-confidence-DIS-GPARmacro})\\

\texttt{IND} & 
$\Longleftrightarrow~ = \hat{P}_G(E_{G,p_2,V_2;p_1})$ &
(P\ref{proposition-confidence-IND-ISAR-GPARsinglemicro})
&
$\arrowbox{\Longrightarrow} = \tfrac{1}{|G|} \sum_{g \in G} \hat{P}_G(E_{G,p_2,V_2;p_1})$ &
(P\ref{proposition-confidence-IND-GPARmacro})\\

\texttt{POS} & 
$\Longleftrightarrow~ > \hat{P}_G(E_{G,p_2,V_2;p_1})$ &
(P\ref{proposition-confidence-POS-ISAR-GPARsinglemicro})
&
$\arrowbox{\Longrightarrow} > \tfrac{1}{|G|} \sum_{g \in G} \hat{P}_G(E_{G,p_2,V_2;p_1})$ &
(P\ref{proposition-confidence-POS-GPARmacro})\\

\texttt{NEG} & 
$\Longleftrightarrow~ < \hat{P}_G(E_{G,p_2,V_2;p_1})$ &
(P\ref{proposition-confidence-NEG-ISAR-GPARsinglemicro})
&
$\arrowbox{\Longrightarrow} < \tfrac{1}{|G|} \sum_{g \in G} \hat{P}_G(E_{G,p_2,V_2;p_1})$ &
(P\ref{proposition-confidence-NEG-GPARmacro})\\
\bottomrule
\end{tabularx}
\end{table}

\begin{table}
\caption{Key characteristics of the lift metrics, assuming that $\hat{P}_T(E_A) > 0 \land \hat{P}_T(E_B) > 0$,
$\hat{P}_g(E_{g,p_1,V_1;p_2}) > 0 \land \hat{P}_g(E_{g,p_2,V_2;p_1}) > 0$, 
$\hat{P}_G(E_{G,p_1,V_1;p_2}) > 0 \land \hat{P}_G(E_{G,p_2,V_2;p_1}) > 0$, or 
$\forall g \in G : \hat{P}_g(E_{g,p_1,V_1;p_2}) > 0 \land \hat{P}_g(E_{g,p_2,V_2;p_1}) > 0$,
 respectively. All four lift metrics have the co-domain $[0,\infty) \cap \mathbb{Q}$.}
\label{properties-lift}
\renewcommand{\arraystretch}{0.8}
\setlength{\tabcolsep}{4pt}

\begin{tabularx}{\textwidth}{
  L{0.04\textwidth} %
  L{0.38\textwidth} %
  L{0.06\textwidth} %
  L{0.40\textwidth} %
  L{0.06\textwidth} %
}
\toprule
\textbf{Sit.} & 
\textbf{ISAR} &
\textbf{} &
\textbf{GPAR (single)} &
\textbf{}\\
\cmidrule[1pt](lr){2-3} \cmidrule[1pt](lr){4-5}

\texttt{IDE} & 
$\arrowbox{\Longrightarrow} =1/\hat{P}_T(E_B)$ &
(P\ref{proposition-lift-IDE-ISAR-GPARsinglemicro})
&
$\arrowbox{\Longrightarrow} =1/\hat{P}_g(E_{g,p_1,V_1;p_2})$ &
(P\ref{proposition-lift-IDE-ISAR-GPARsinglemicro})\\

\texttt{DIS} & 
$\Longleftrightarrow~ =0$ &
(P\ref{proposition-lift-DIS-ISAR-GPARsinglemicro})
&
$\Longleftrightarrow~ =0$ &
(P\ref{proposition-lift-DIS-ISAR-GPARsinglemicro})\\

\texttt{IND} & 
$\Longleftrightarrow~ =1$ &
(P\ref{proposition-lift-IND-ISAR-GPARsinglemicro})
&
$\Longleftrightarrow~ =1$ &
(P\ref{proposition-lift-IND-ISAR-GPARsinglemicro})\\

\texttt{POS} & 
$\Longleftrightarrow~ >1$ &
(P\ref{proposition-lift-POS-ISAR-GPARsinglemicro})
&
$\Longleftrightarrow~ >1$ &
(P\ref{proposition-lift-POS-ISAR-GPARsinglemicro})\\

\texttt{NEG} & 
$\Longleftrightarrow~ <1$ &
(P\ref{proposition-lift-NEG-ISAR-GPARsinglemicro})
&
$\Longleftrightarrow~ <1$ &
(P\ref{proposition-lift-NEG-ISAR-GPARsinglemicro})\\

\cmidrule[1pt](lr){1-5}
\textbf{Sit.} & 
\textbf{GPAR (micro)} &
\textbf{} &
\textbf{GPAR (macro)} &
\textbf{}\\
\cmidrule[1pt](lr){2-3} \cmidrule[1pt](lr){4-5}

\texttt{IDE} & 
$\arrowbox{\Longrightarrow} =1/\hat{P}_G(E_{G,p_1,V_1;p_2})$ &
(P\ref{proposition-lift-IDE-ISAR-GPARsinglemicro})
&
$\arrowbox{\Longrightarrow} = \tfrac{1}{|G|}\sum_{g \in G} 1/\hat{P}_G(E_{G,p_1,V_1;p_2})$ &
(P\ref{proposition-lift-IDE-GPARmacro})\\

\texttt{DIS} & 
$\Longleftrightarrow~ =0$ &
(P\ref{proposition-lift-DIS-ISAR-GPARsinglemicro})
&
$\Longleftrightarrow~ =0$ &
(P\ref{proposition-lift-DIS-GPARmacro})\\

\texttt{IND} & 
$\Longleftrightarrow~ =1$ &
(P\ref{proposition-lift-IND-ISAR-GPARsinglemicro})
&
$\arrowbox{\Longrightarrow} =1$ &
(P\ref{proposition-lift-IND-GPARmacro})\\

\texttt{POS} & 
$\Longleftrightarrow~ > 1$ &
(P\ref{proposition-lift-POS-ISAR-GPARsinglemicro})
&
$\arrowbox{\Longrightarrow} >1$ &
(P\ref{proposition-lift-POS-GPARmacro})\\

\texttt{NEG} & 
$\Longleftrightarrow~ <1$ &
(P\ref{proposition-lift-NEG-ISAR-GPARsinglemicro})
&
$\arrowbox{\Longrightarrow} <1$ &
(P\ref{proposition-lift-NEG-GPARmacro})\\
\bottomrule
\end{tabularx}
\end{table}

\begin{table}
\caption{Key characteristics of the leverage metrics, assuming that $T \neq \varnothing$, $\fulltau \neq \varnothing$, $\fulltauG \neq \varnothing$, or $\forall g \in G : \fulltau \neq \varnothing$, respectively. All four leverage metrics have the co-domain $[-\tfrac{1}{4},\tfrac{1}{4}] \cap \mathbb{Q}$.}
\label{properties-leverage}
\renewcommand{\arraystretch}{0.8}
\setlength{\tabcolsep}{4pt}

\begin{tabularx}{\textwidth}{
  L{0.04\textwidth} %
  L{0.38\textwidth} %
  L{0.06\textwidth} %
  L{0.40\textwidth} %
  L{0.06\textwidth} %
}
\toprule
\textbf{Sit.} & 
\textbf{ISAR} &
\textbf{} &
\textbf{GPAR (single)} &
\textbf{}\\
\cmidrule[1pt](lr){2-3} \cmidrule[1pt](lr){4-5}

\texttt{IDE} & 
$\arrowbox{\Longrightarrow} = \hat{P}_T(E_A)(1-\hat{P}_T(E_B))$ &
(P\ref{proposition-leverage-IDE-ISAR-GPARsinglemicro})
&
$\arrowbox{\Longrightarrow} = \hat{P}_g(E_{g,p_1,V_1;p_2})(1-\hat{P}_g(E_{g,p_2,V_2;p_1}))$ &
(P\ref{proposition-leverage-IDE-ISAR-GPARsinglemicro})\\

\texttt{DIS} & 
$\Longleftrightarrow~ = -\hat{P}_T(E_A)\hat{P}_T(E_B)$ &
(P\ref{proposition-leverage-DIS-ISAR-GPARsinglemicro})
&
$\Longleftrightarrow~ =-\hat{P}_g(E_{g,p_1,V_1;p_2})\hat{P}_g(E_{g,p_2,V_2;p_1})$ &
(P\ref{proposition-leverage-DIS-ISAR-GPARsinglemicro})\\

\texttt{IND} & 
$\Longleftrightarrow~ =0$ &
(P\ref{proposition-leverage-IND-ISAR-GPARsinglemicro})
&
$\Longleftrightarrow~ =0$ &
(P\ref{proposition-leverage-IND-ISAR-GPARsinglemicro})\\

\texttt{POS} & 
$\Longleftrightarrow~ >0$ &
(P\ref{proposition-leverage-POS-ISAR-GPARsinglemicro})
&
$\Longleftrightarrow~ >0$ &
(P\ref{proposition-leverage-POS-ISAR-GPARsinglemicro})\\

\texttt{NEG} & 
$\Longleftrightarrow~ <0$ &
(P\ref{proposition-leverage-NEG-ISAR-GPARsinglemicro})
&
$\Longleftrightarrow~ <0$ &
(P\ref{proposition-leverage-NEG-ISAR-GPARsinglemicro})\\

\cmidrule[1pt](lr){1-5}
\textbf{Sit.} & 
\textbf{GPAR (micro)} &
\textbf{} &
\textbf{GPAR (macro)} &
\textbf{}\\
\cmidrule[1pt](lr){2-3} \cmidrule[1pt](lr){4-5}

\texttt{IDE} & 
$\arrowbox{\Longrightarrow} = \hat{P}_G(E_{G,p_1,V_1;p_2})(1-\hat{P}_G(E_{G,p_2,V_2;p_1}))$ &
(P\ref{proposition-leverage-IDE-ISAR-GPARsinglemicro})
&
$\arrowbox{\Longrightarrow} = \tfrac{1}{|G|}\sum_{g \in G}\hat{P}_g(E_{g,p_1,V_1;p_2})(1-\hat{P}_g(E_{g,p_2,V_2;p_1}))$ &
(P\ref{proposition-leverage-IDE-GPARmacro})\\

\texttt{DIS} & 
$\Longleftrightarrow~ =-\hat{P}_G(E_{G,p_1,V_1;p_2})\hat{P}_G(E_{G,p_2,V_2;p_1})$ &
(P\ref{proposition-leverage-DIS-ISAR-GPARsinglemicro})
&
$\arrowbox{\Longrightarrow} =-\tfrac{1}{|G|}\sum_{g \in G}\hat{P}_g(E_{g,p_1,V_1;p_2}) \hat{P}_g(E_{g,p_2,V_2;p_1})$ &
(P\ref{proposition-leverage-DIS-GPARmacro})\\

\texttt{IND} & 
$\Longleftrightarrow~ =0$ &
(P\ref{proposition-leverage-IND-ISAR-GPARsinglemicro})
&
$\arrowbox{\Longrightarrow} =0$ &
(P\ref{proposition-leverage-IND-GPARmacro})\\

\texttt{POS} & 
$\Longleftrightarrow~ >0$ &
(P\ref{proposition-leverage-POS-ISAR-GPARsinglemicro})
&
$\arrowbox{\Longrightarrow} >0$ &
(P\ref{proposition-leverage-POS-GPARmacro})\\

\texttt{NEG} & 
$\Longleftrightarrow~ <0$ &
(P\ref{proposition-leverage-NEG-ISAR-GPARsinglemicro})
&
$\arrowbox{\Longrightarrow} <0$ &
(P\ref{proposition-leverage-NEG-GPARmacro})\\
\bottomrule
\end{tabularx}
\end{table}

\begin{table}
\caption{Key characteristics of the conviction metrics, assuming that $\hat{P}_T(E_A) > 0 \land \hat{P}_T(E_B) < 1$, 
$\hat{P}_g(E_{g,p_1,V_1;p_2}) > 0 \land \hat{P}_g(E_{g,p_2,V_2;p_1}) < 1$, 
$\hat{P}_G(E_{G,p_1,V_1;p_2}) > 0 \land \hat{P}_G(E_{G,p_2,V_2;p_1}) < 1$, 
or $\forall g \in G : \hat{P}_g(E_{g,p_1,V_1;p_2}) > 0 \land \hat{P}_g(E_{g,p_2,V_2;p_1}) < 1$, respectively.
All four conviction metrics have the co-domain $\mathbb{Q}_{>0} \cup \lbrace \infty \rbrace$.}
\label{properties-conviction}
\renewcommand{\arraystretch}{0.8}
\setlength{\tabcolsep}{4pt}

\begin{tabularx}{\textwidth}{
  L{0.04\textwidth} %
  L{0.38\textwidth} %
  L{0.06\textwidth} %
  L{0.40\textwidth} %
  L{0.06\textwidth} %
}
\toprule
\textbf{Sit.} & 
\textbf{ISAR} &
\textbf{} &
\textbf{GPAR (single)} &
\textbf{}\\
\cmidrule[1pt](lr){2-3} \cmidrule[1pt](lr){4-5}

\texttt{IDE} & 
$\arrowbox{\Longrightarrow} = \infty$ &
(P\ref{proposition-conviction-IDE-ISAR-GPARsinglemicro})
&
$\arrowbox{\Longrightarrow} = \infty$ &
(P\ref{proposition-conviction-IDE-ISAR-GPARsinglemicro})\\

\texttt{DIS} & 
$\Longleftrightarrow~ =1-\hat{P}_T(E_B)$ &
(P\ref{proposition-conviction-DIS-ISAR-GPARsinglemicro})
&
$\Longleftrightarrow~ =1-\hat{P}_g(E_{g,p_2,V_2;p_1})$ &
(P\ref{proposition-conviction-DIS-ISAR-GPARsinglemicro})\\

\texttt{IND} & 
$\Longleftrightarrow~ =1$ &
(P\ref{proposition-conviction-IND-ISAR-GPARsinglemicro})
&
$\Longleftrightarrow~ =1$ &
(P\ref{proposition-conviction-IND-ISAR-GPARsinglemicro})\\

\texttt{POS} & 
$\Longleftrightarrow~ >1$ &
(P\ref{proposition-conviction-POS-ISAR-GPARsinglemicro})
&
$\Longleftrightarrow~ >1$ &
(P\ref{proposition-conviction-POS-ISAR-GPARsinglemicro})\\

\texttt{NEG} & 
$\Longleftrightarrow~ <1$ &
(P\ref{proposition-conviction-NEG-ISAR-GPARsinglemicro})
&
$\Longleftrightarrow~ <1$ &
(P\ref{proposition-conviction-NEG-ISAR-GPARsinglemicro})\\

\cmidrule[1pt](lr){1-5}
\textbf{Sit.} & 
\textbf{GPAR (micro)} &
\textbf{} &
\textbf{GPAR (macro)} &
\textbf{}\\
\cmidrule[1pt](lr){2-3} \cmidrule[1pt](lr){4-5}

\texttt{IDE} & 
$\arrowbox{\Longrightarrow} =\infty$ &
(P\ref{proposition-conviction-IDE-ISAR-GPARsinglemicro})
&
$\arrowbox{\Longrightarrow} = \infty$ &
(P\ref{proposition-conviction-IDE-GPARmacro})\\

\texttt{DIS} & 
$\Longleftrightarrow~ =1-\hat{P}_G(E_{G,p_2,V_2;p_1})$ &
(P\ref{proposition-conviction-DIS-ISAR-GPARsinglemicro})
&
$\arrowbox{\Longrightarrow} = \tfrac{1}{|G|}\sum_{g \in G}(1-\hat{P}_g(E_{g,p_2,V_2;p_1}))$ &
(P\ref{proposition-conviction-DIS-GPARmacro})\\

\texttt{IND} & 
$\Longleftrightarrow~ =1$ &
(P\ref{proposition-conviction-IND-ISAR-GPARsinglemicro})
&
$\arrowbox{\Longrightarrow} =1$ &
(P\ref{proposition-conviction-IND-GPARmacro})\\

\texttt{POS} & 
$\Longleftrightarrow~ >1$ &
(P\ref{proposition-conviction-POS-ISAR-GPARsinglemicro})
&
$\arrowbox{\Longrightarrow} >1$ &
(P\ref{proposition-conviction-POS-GPARmacro})\\

\texttt{NEG} & 
$\Longleftrightarrow~ <1$ &
(P\ref{proposition-conviction-NEG-ISAR-GPARsinglemicro})
&
$\arrowbox{\Longrightarrow} <1$ &
(P\ref{proposition-conviction-NEG-GPARmacro})\\
\bottomrule
\end{tabularx}
\end{table}

Table \ref{tab:whendefined} summarizes the conditions under which the metrics are defined. We remove redundant terms from these conditions. For example, instead of $T \neq \varnothing \land \hat{P}_T(E_A) >0$ we write $\hat{P}_T(E_A) >0$, because from $\hat{P}_T(E_A) > 0$ it follows that $T \neq \varnothing$.

\begin{table}
\caption{Overview over the necessary conditions for the scores of the metrics to be defined.}
\label{tab:whendefined}
\centering
\begin{tabularx}{0.7\linewidth}{lX}
\cmidrule[1pt](lr){1-2}
\textbf{Metric} & 
\textbf{Condition} \\
\cmidrule[1pt](lr){1-2}
$\operatorname{confidence}_T$ & 
$\hat{P}_T(E_A) > 0$
\\
$\operatorname{confidence}_g$ & 
$\hat{P}_g(E_{g,p_1,V_1;p_2}) > 0$
\\
$\operatorname{confidence}_G$ & 
$\hat{P}_G(E_{G,p_1,V_1;p_2}) > 0$
\\
$\operatorname{macro-confidence}_G$ & 
$\forall g \in G : \hat{P}_g(E_{g,p_1,V_1;p_2}) > 0$\\

\cmidrule[1pt](lr){1-2}

$\operatorname{lift}_T$ & 
$\hat{P}_T(E_A) > 0 \land \hat{P}_T(E_B) > 0$
\\
$\operatorname{lift}_g$ & 
$\hat{P}_g(E_{g,p_1,V_1;p_2}) > 0 \land \hat{P}_g(E_{g,p_2,V_2;p_1}) > 0$
\\
$\operatorname{lift}_G$ & 
$\hat{P}_G(E_{G,p_1,V_1;p_2}) > 0 \land \hat{P}_G(E_{G,p_2,V_2;p_1}) > 0$
\\
$\operatorname{macro-lift}_G$ & 
$\forall g \in G : \hat{P}_g(E_{g,p_1,V_1;p_2}) > 0 \land \hat{P}_g(E_{g,p_2,V_2;p_1}) > 0$\\

\cmidrule[1pt](lr){1-2}

$\operatorname{leverage}_T$ & 
$T \neq \varnothing$
\\
$\operatorname{leverage}_g$ & 
$\fulltau \neq \varnothing$
\\
$\operatorname{leverage}_G$ & 
$\fulltauG \neq \varnothing$
\\
$\operatorname{macro-leverage}_G$ & 
$\forall g \in G : \fulltau \neq \varnothing$\\

\cmidrule[1pt](lr){1-2}

$\operatorname{conviction}_T$ & 
$\hat{P}_T(E_A) > 0 \land \hat{P}_T(E_B) < 1$
\\
$\operatorname{conviction}_g$ & 
$\hat{P}_g(E_{g,p_1,V_1;p_2}) > 0 \land \hat{P}_g(E_{g,p_2,V_2;p_1}) < 1$
\\
$\operatorname{conviction}_G$ & 
$\hat{P}_G(E_{G,p_1,V_1;p_2}) > 0 \land \hat{P}_G(E_{G,p_2,V_2;p_1}) < 1$
\\
$\operatorname{macro-conviction}_G$ & 
$\forall g \in G : \hat{P}_g(E_{g,p_1,V_1;p_2}) > 0 \land \hat{P}_g(E_{g,p_2,V_2;p_1}) < 1$\\
\cmidrule[1pt](lr){1-2}
\end{tabularx}
\end{table}

As we can see from Tables \ref{properties-confidence}--\ref{properties-conviction} and from Table \ref{tab:whendefined}, the metrics behave similarly across contexts. %
The metrics for itemset-based association rules and the metrics for graph pattern-based association rules in the single graph setting and in the graph transactional setting where scores are micro-averaged have similar characteristics, because these characteristics are a consequence of the fact that we have a probability space and the metrics are defined i.t.o. probabilities. However, there are some differences between the macro-averaged metrics and the other metrics.
With macro-averaged metrics we often (in $14$ of $20$ cases) loose the possibility to deduce the situation from the score. For example, whereas in ISAR context if the score of lift is $1$ then we can deduce that $\hat{P}_T(E_A)$ and $\hat{P}_T(E_B)$ are independent. However, in GPAR (macro) context if the score of lift is $1$ we cannot deduce that $\hat{P}_G(E_{G,p_1,V_1;p2})$ and $\hat{P}_G(E_{G,p_2,V_2;p_1})$ are independent in each $g \in G$.

A further disadvantage of the macro-averaged metrics is, as can be seen in Table \ref{tab:whendefined},
that the conditions when a score is defined are much stronger, as these conditions need to apply for each graph in $G$. These conditions require $G$ to be reduced and the metrics to be evaluated on reduced bags, as discussed in Section \ref{sec:metricstransactional}.

\section{Related Work}\label{sec:relatedwork}
Related work comes from three areas: i) extensions of the itemset-based association rule concept, ii) extensions of functional dependencies from relational databases, and iii) Horn rules.
We first provide an overview over general extensions of itemset-based association rules, before we in detail discuss the most related concepts.

\subsection{Extensions of Association Rules}
The classical concept of association rules -- where rules have the form $X \Rightarrow Y$ where X and Y are itemsets and where rules are mined from a database of transactions where a transaction is an itemset -- has been extended in multiple directions. %

\begin{itemize}
\item \textit{Temporal association rules} (see, e.g., Zhou and Hirasawa \cite{zhou2019evolving}) include temporal information and can then, e.g., be used to predict a customer's future demand. For example, the rule
$shoes~ (t=0) \land cloth~ (t=0) \Rightarrow shoe~ oil ~(t=2)$ expresses that customers usually buy shoes and cloth together (at $t=0$) and buy shoe oil two months later (at $t=2$).

\item A related concept is \textit{sequential association rules} (see, e.g., Srikant et al. and Fournier-Viger et al \cite{srikant1996mining,fournier2012cmrules}) where a transaction is a sequence of itemsets. 
Sequential association rules have, for example, been applied for the task of relation detection in text \cite{cellier2015sequential}. 
For example, we can transform a sentence such as "animals such as dogs" into a sequence of sets where each set contains a token and the token's part-of-speech tag, i.e., $(\{t{-}animals,p{-}NOUN\},\, \{t{-}such, p{-}JJ\},\, \\\{t{-}as, p{-}IN\},\, \{t{-}dogs, p{-}NOUN\})$.
The following rule detects that in the sentence "animals such as dogs" the hyponym relation is expressed:\\
$(\{p{-}NOUN\},\, \{t{-}such, p{-}JJ\},\, \{t{-}as, p{-}IN\},\, \{p{-}NOUN\}) \Rightarrow hyponym$.

\item \textit{Spacial association rules} (see, e.g., Koperski and Han \cite{koperski1995discovery}) are rules of the form $X \Rightarrow Y$ where $X$ and $Y$ are sets of predicates where some predicates are spacial predicates. For example, in the following rule, the spacial predicates $within$, $adjacent$, and $closeTo$ occur: $isA(X,\, city) \,\land\, within(X,\, BritishColumbia) \,\land\, adjacentTo(X,\, water) \Rightarrow closeTo(X,\, USA)$. Here, $X$ is a variable. The rule expresses that if there is a thing that is a city, that thing is located within British Columbia, and that thing is adjacent to water, then that thing is located close to the USA. 

\item \textit{Numerical association rules} (see, e.g., Fukuda et al. \cite{fukuda1996mining,kaushik2021systematic}) contain statements that, when rules are mined or applied, require numerical values to be compared. For example, the rule
$(Balance \in [15821, 26264]) \Rightarrow (CardLoan = yes)$, expresses that if a customer's balance score is a value in the numerical interval $[15821, 26264]$, then the customer is likely to use card loan. 

\item \textit{Quantitative association rules} (see, e.g., Srikant and Agrawal \cite{srikant96mining}) are related, as can be seen in the following example rule:
$(Age: 30..39) \land (Married: Yes) \Rightarrow (NumCars = 2)$
where $Age$ and $NumCars$ are numerical attributes and $Married$ is a categorical attribute. In the context of a database about persons, the rule expresses that a person between the age of 30 and 39 that is not married usually owns two cars. 
Related is the work by
Taboada et al. \cite{taboada2008association} where rules of the following form are mined:
$ (A_1 > a_1) \land \ldots \land (A_k > a_k) \Rightarrow (A_m > a_m) \land \cdots \land (A_n > a_n)$. A challenge lies in finding good threshold values.

\item Whereas numerical and quantitative association rules require that crisp threshold values are defined or found, for example by partitioning a scale into intervals,
\textit{fuzzy association rules} (see, e.g., Hu et al. \cite{hu2003discovering}) allow to define linguistic variables such as "high" and "low".
See the following example of a fuzzy association rule, taken from \cite{au1999farm}:
$Investment{-}in{-}construction = Very{-}high
\Rightarrow Capital{-}of{-}transportation = Very{-}high$. The rule expresses that very high investment in construction tends to introduce high capital of transportation. The concept of fuzziness has been combined with other types of association rules, resulting, e.g., in fuzzy temporal association rules (see, e.g., Carinena \cite{carinena2014fuzzy}).

\item Similar to fuzzy association rules, but not based on fuzzy set membership, are \textit{gradual association rules} (see, e.g., Di-Jorio et al. \cite{di2009mining}). Going beyond
 interval based rules, such as \textit{a room located less than one kilometer from the centre will cost between 60 and 200 dollars}, 
a gradual association rule such as $\lbrace Pop^{\geq}, Dist^{\leq} \rbrace \Rightarrow \lbrace  Price^{\geq} \rbrace$ expresses \textit{the bigger the town and the nearer from
the town centre, then the higher the price}.  

\item Whereas the previously mentioned types of association rules deal with increased expressiveness of rules or with a modified interpretation of rules, \textit{privacy-preserving association rules} (see, e.g., Saygin et al. \cite{saygin2002privacy}) deal with the problem to prevent to learn association rules that contain sensitive data that is contained in the transaction database.

\end{itemize}

\subsection{Relational Association Rules}

Extensions towards more complex transactions deal with the scenario where a relational database is given that is viewed as a Datalog database. %
\textit{Relational association rules}, introduced by Dehaspe and Toivonen \cite{10.5555/567222.567232,dehaspe1999discovery}) and later extended by Goethals and Van den Bussche \cite{goethals2002relational},
 contain conjunctive queries as antecedent and consequent. 
For example, consider the following two conjunctive queries shown in Prolog notation:
\begin{align*}
Q_1(x,y) & \mathrel{:\!-} likes(x,\texttt{"Duvel"}), visits(x,y)\\
Q_2(x,y) & \mathrel{:\!-} likes(x,\texttt{"Duvel"}), visits(x,y), serves(y,\texttt{"Duvel"})
\end{align*}

In the first conjunctive query, $Q_1(x,y)$ is the head of the rule and $likes(x,\texttt{"Duvel"}), visits(x,y)$ is the body of the rule. 
Let $A$, $B$, $C$, and $D$ be constants, 
let $likes(A,B)$ be a predicate that expresses that the person $A$ likes the beverage $B$,
let $visits(A,C)$ be a predicate that expresses that the person $A$ visits the bar $C$, and
let $serves(C,D)$ be a predicate that expresses that the bar $C$ serves the beverage $D$. 
The rule $Q_1 \Rightarrow Q_2$ can be interpreted as follows: if a person (x) that likes the beverage Duvel visits a bar (y), then that bar serves the beverage Duvel.

We can express the relational association rule as a graph pattern-based association rule of the form $(p_1,p_2)$ as follows:
\begin{align*}
r ~=~& (p_1,p_2,V)\\
p_1 ~=~& \lbrace (v_1, likes, "Duvel"), (v_1,visits,v_2) \rbrace\\
p_2 ~=~& \lbrace (v_1, likes, "Duvel"), (v_1,visits,v_2), (v_2,serves,"Duvel") \rbrace \end{align*}

Both the antecedent conjunctive query and the consequent conjunctive query can be seen as graph patterns, with the difference that variables cannot occur in predicate position. Note that the body of the antecedent is contained in the body of the consequent, which is a requirement for this type of rules, but which is not required for the GPAR concept we introduced.
Further differences relate to the evaluation semantics of relational association rules. Standard conjunctive query evaluation semantics is comparable to homomorphism semantics and conjunctive queries cannot contain statements of the form $x \neq y$ or $x \neq A$, where $x$ and $y$ are variables, $A$ is a constant, which one would need to add to enforce an evaluation that resembles no-repeated-anything semantics.

The support of a conjunctive query is defined as the number of tuples in the query's answer relation. When evaluating a conjunctive query such as $Q_1(x,y) \mathrel{:\!-} likes(x,\texttt{"Duvel"}), visits(x,y)$, the variables $x$ and $y$ serve as projection variables and the resulting relation is a set of pairs. When the conjunctive query is extended, then the number of pairs cannot grow -- thus, the support metric is anti-monotonic.

Confidence of a rule $Q_1 \Rightarrow Q_2$ is defined as the number of tuples in the answer relation of $Q_2$ divided by the number of tuples in the answer relation of $Q_1$.\\

We can represent the relational database as a graph $g$. Let $\mathcal{T}_g^n$ be the set of (not necessarily non-repetitive) tuples of length $n=|V|$ consisting of constants that occur in the relational database. Using our notation, we express their metrics of absolute support and confidence as follows:
\[\operatorname{absolute-support}_g(p,V) := |\lbrace \, T \in \mathcal{T}_g^n \mid \exists \mu \in \Omega^{hom}_{p_,g} : \mu(V) = T \,\rbrace|\]
\[ \operatorname{confidence}_g(p_1,p_2,V) = \dfrac{|\lbrace \, T \in \mathcal{T}_g^n \mid 
\exists \mu_1 \in \Omega^{hom}_{p_1,g} : \mu_1(V) = T
\land \exists \mu_2 \in \Omega^{hom}_{p_2,g} : \mu_2(V) = T
\,\rbrace|}{|\lbrace \, T \in \mathcal{T}_g^n \mid \exists \mu \in \Omega^{hom}_{p_1,g} : \mu(V) = T \,\rbrace|}\]

This shows some degree of similarity between their metrics and confidence and our metrics of confidence, as the strength of correspondence between two graph patterns, given a sequence of joining variables, is scored based on common term sequences.

\subsection{Graph Entity Depencencies}\label{sec:GED}
The concept of \textit{functional dependency (FD)}, introduced by Codd \cite{codd1972further}, is well-known in the context of relational databases, where functional dependencies ensure data consistency and guide schema normalization.

Given a set $A = \lbrace A_1, \ldots, A_n \rbrace$ of attributes with domains $D_1, \ldots D_n$, respectively, let $X,Y \subseteq A$ be two sets of attributes, and let $R \subseteq D_1 \times \cdots \times D_n$ be a relation. 
A functional dependency $X \to Y$ expresses that for any two tuples $t_1,t_2 \in R$, if the tuples agree on the attributes in X, they must also agree on attributes in Y, i.e.,
\[\forall t_1,t_2 \in R : (t_1[X] = t_2[X] \Rightarrow t_1[Y] = t_2[Y])\]

For example, let a customer be described 
in terms of the phone number country code (CC),
phone number area code (AC),
phone number (PN),
name (NM),
street (STR),
city (CT), and ZIP code (ZIP). Thus, $R \subseteq D_{CC} \times D_{AC} \times \cdots \times D_{STR}$. 
The functional dependency $[CC,AC] \to [CT]$ expresses that
\[\forall t_1,t_2 \in R : (t_1[CC,AC] = t_2[CC,AC] \Rightarrow t_1[CT] = t_2[CT]).\]

Thus, customers with the same phone number country code and the same phone number area code live in the same city. 
Yao et al. \cite{yao2008mining} have investigated the problem of mining functional dependencies from relational data. 
To capture more fine-grained data relationships, the concept of functional dependency has been generalized to \textit{conditional functional dependency (CFD)} by Bohannon et al. \cite{bohannon2006conditional}, which allows the specification of conditions on attribute values:
\[\forall t_1,t_2 \in R : ((t_1[X] = t_2[X] \land C_1(t_1[X])) \Rightarrow (t_1[Y] = t_2[Y] \land C_2(t_1[Y])))\]
Here, $C_1$ and $C_2$ are predicates that evaluate to true if the condition is satisfied; false otherwise. 
For example, the CFD $[CC=\texttt{"01"},AC=\texttt{"908"},PN] \to [STR,CT=\texttt{"MH"},ZIP]$ expresses that
\begin{align*}
\forall t_1,t_2 \in R : ( & ( \\
 & t_1[CC,AC,PN] = t_2[CC,AC,PN] \land\\
 & t_1[CC]=\texttt{"01"} \land \\
 & t_1[AC] = \texttt{"908"}) \Rightarrow (\\
 & t_1[STR,CT,ZIP] = t_2[STR,CT,ZIP] \land\\
 & t_1[CT] = \texttt{"MH"})).
\end{align*}
Here, $C_1 \equiv t_1[CC] = \texttt{"01"} \land t_1[AC] = \texttt{"908"}$ and $C_2 \equiv t_1[CT] = \texttt{"MH"}$. Thus, customers with the same phone number country code (which is 01), the same phone number area code (which is 908), and the same phone number live in the same street, live in the same city (which is MH), and have the same ZIP number. 

Fan et al. \cite{fan2010discovering} have investigated the problem of mining conditional functional dependencies from relational data. 
While functional dependencies and conditional functional dependencies are defined in the relational context, more recent work has extended the notion of dependencies to graph-structured data. In this setting, the concept of a \textit{graph functional dependency (GFD)} has been introduced \cite{fan2016functional,fan2017dependencies,fan2019dependencies}, 
but for a type of graphs similar to property graphs \cite{bonifati2019schema}. %
GFD was extended to \textit{graph entity dependency} (GED) by Fan and Lu \cite{fan2019dependencies}, which subsumes GFD.%

GEDs are designed to be used for detecting inconsistencies, they are consistency constraints. Given an infinite set $\Gamma$ of labels, an infinite set $\Upsilon$ of attributes, and an infinite set $U$ of values,
GEDs are defined for property graphs of the form $G=(V,E,L,F_A)$ where
$V$ is a finite set of node identifiers, $E \subseteq V \times \Gamma \times V$ is a finite set of edges where $(v,\iota, v') \in E$ denotes an edge from node $v$ to node $v'$ labeled with $\iota$, 
each node $v \in V$ carries a label $L(v)$,
and for each node $v \in V$, $F_A(v)$ is a finite tuple of the form $(A_1=a_1, \ldots, A_n=a_n)$ where $A_1, \ldots,A_n$ are attributes in $\Upsilon$, $a_1, \ldots, a_n$ are constants in $U$, and $A_i \neq A_j$ if $i \neq j$. We write $v.A_i$ to denote the value of the attribute $A_i$ of the node $v$. Finally, every node $v \in V$ has a special attribute \textit{id} that carries the node's identifier.

The following example graph expresses that the person Alice is married to the person Bob and that the home address of Alice is Bielefeld:
\begin{align*}
\Gamma ~=~& \lbrace marriedTo \rbrace\\
\Upsilon ~=~& \lbrace id, name, homeAddress \rbrace\\
U ~=~& \lbrace v_1, v_2, \texttt{"Alice"}, \texttt{"Bob"}, \texttt{"Bielefeld"} \rbrace\\
G ~=~& (V,E,L,F_A)\\
V ~=~& \lbrace v_1, v_2 \rbrace\\
E ~=~& \lbrace (v_1, marriedTo, v_2) \rbrace\\
L ~=~& \lbrace (v_1, person),~(v_2, person) \rbrace\\
F_A(v_1) ~=~& \lbrace (id, v_1),~(name, \texttt{"Alice"}),~(homeAddress, \texttt{"Bielefeld"}) \rbrace\\
F_A(v_2) ~=~& \lbrace (id, v_2),~(name, \texttt{"Bob"}) \rbrace
\end{align*}

A graph pattern is a graph $Q[\bar{x}] = (V_Q, E_Q, L_Q)$ where $V_Q$ is a finite set of node identifiers, $E_Q \subseteq V_Q \times (\Gamma \cup \lbrace "\_"\rbrace) \times V_Q$ is a finite set of edges
where $(v,\iota, v') \in E_Q$ denotes an edge labeled with $\iota$ where $\iota = \texttt{"\_"}$ denotes a wildcard, and $L_Q$ is a function that assigns a label from $\Gamma \cup \lbrace \texttt{"\_"} \rbrace$ to each node. $\bar{x}$ is a list of distinct variables, each referring to another node in $V_Q$.

In the context of the graph above, the following shows an example of a graph pattern $Q[\bar{x}]$. Here, we interpret each identifier $x_i$ as a variable. Thus, in the graph pattern, all nodes are variables:
\begin{align*}
Q[\bar{x}] ~=~ & (V_Q,E_Q,L_Q)\\
\bar{x} ~=~& (x_1, x_2)\\
V_Q ~=~& \lbrace x_1, x_2 \rbrace\\
E_Q ~=~& \lbrace (x_1, marriedTo, x_2) \rbrace\\
L_Q ~=~& \lbrace (x_1, person),~(x_2, person) \rbrace
\end{align*}

Finally, we can define a GED as $Q[\bar{x}](X \Rightarrow Y)$ where $Q[\bar{x}]$ is a graph pattern and $X$ and $Y$ are (possibly empty) sets of literals of $\bar{x}$.
A literal of $\bar{x}$ is either a constant literal, a variable literal, or an id literal. A constant literal has the form $x.A = c$ where $x \in \bar{x}$ is a variable, $A \in \Upsilon$ (with $A \neq id$) is an attribute, and $c \in U$ is a constant.
A variable literal has the form $x.A = y.B$
where $x,y \in \bar{x}$ are variables and $A,B \in \Upsilon$ (with $A \neq id$ and $B \neq id$) are attributes.
An id literal has the form $x.id = y.id$ where $x,y \in \bar{x}$ are variables.

Bringing everything together, for our example graph we can define the following GED $\phi$:

\begin{align*}
\phi ~=~& Q[\bar{x}](X \Rightarrow Y)\\
Q[\bar{x}] ~=~& (V_Q,E_Q,L_Q)\\
\bar{x} ~=~& (x_1, x_2)\\
V_Q ~=~& \lbrace x_1, x_2 \rbrace\\
E_Q ~=~& \lbrace (x_1, marriedTo, x_2) \rbrace\\
L_Q ~=~& \lbrace (x_1, person),~(x_2, person) \rbrace\\
X ~=~& \varnothing\\
Y ~=~& \lbrace x_1.homeAddress = x_2.homeAddress \rbrace
\end{align*}

The GED expresses that if two persons are married, then they are expected to have the same home address. 
When this GED is evaluated over the graph then the inconsistency is found, i.e., that there is a pair $(v_1, v_2)$ of persons (where $v_1$ has the name Alice and $v_2$ has the name Bob) where $v_1$ is married to $v_2$ but the home addresses are not identical -- in fact, the home address of $v_2$ is missing.
Although GEDs are not designed to be used to extend graphs, nevertheless, they \textit{could} be used to extend property graphs in a limited way. Here, since the home address of $v_2$ is missing, it might be possible to correct the situation by copying the home address of $v_1$ to $v_2$. Although this procedure can extend a graph by adding missing attribute values, it cannot extend a graph by adding missing nodes or edges.

\subsection{Graph Association Rules}
Graph association rule (GAR) by Fan et al. \cite{fan2020capturing} is an extension of the concept of graph pattern association rule (GPAR) \cite{fan2015association} and of graph functional dependency (GFD) \cite{fan2016functional}. GARs can be applied to identify missing information and to predict links. Given an infinite set $\Gamma$ of labels, an infinite set $\Upsilon$ of attributes, and an infinite set $U$ of values, the GAR concept is defined for graphs of the form $(V,E,L,F_A)$, as in the context of graph entity dependencies (see Section \ref{sec:GED}), with the difference that no special $id$ attribute is used: $V$ is a finite set of node identifiers, $E \subseteq V \times \Gamma \times V$ is a finite set of edges where $(v,\iota, v') \in E$ denotes an edge from node $v$ to node $v'$ labeled with $\iota$, 
each node $v \in V$ carries a label $L(v)$,
and for each node $v \in V$, $F_A(v)$ is a finite tuple of the form $(A_1=a_1, \ldots, A_n=a_n)$ where $A_1, \ldots,A_n$ are attributes in $\Upsilon$, $a_1, \ldots, a_n$ are constants in $U$, and $A_i \neq A_j$ if $i \neq j$. We write $v.A_i$ to denote the value of the attribute $A_i$ of the node $v$.

A graph pattern is a graph $Q[\bar{x}] = (V_Q, E_Q, L_Q, \mu)$ where $V_Q$ is a finite set of node identifiers, $E_Q \subseteq V_Q \times (\Gamma \cup \lbrace \texttt{"\_"}\rbrace) \times V_Q$ is a finite set of edges
where $(v,\iota, v') \in E_Q$ denotes an edge labeled with $\iota$ where $\iota = \texttt{"\_"}$ denotes a wildcard, and $L_Q$ is a function that assigns a label from $\Gamma \cup \lbrace \texttt{"\_"} \rbrace$ to each node. $\bar{x}$ is a list of distinct variables, and $\mu$ is a bijective mapping that assigns a distinct variable to each node in $V_Q$.

A graph association rule is an expression of the form $Q[\bar{x}](X \to Y)$ where $Q[\bar{x}]$ is a graph pattern, $\bar{x}$ is a non-repetitive list of variables that occur in $Q[\bar{x}]$, and X and Y are sets of literals over $Q[\bar{x}]$. 
A literal over $Q[\bar{x}]$ is either 
an attribute literal, an edge literal, an ML literal, a variable literal, or a constant literal. We here ignore ML literals.
An attribute literal has the form $x.A$, where $x$ is a variable in $\bar{x}$ and $A$ is an attribute in $\Upsilon$.
An edge literal has the form $\iota(x,y)$, where $\iota$ is a label in $\Gamma$ and $x$ and $y$ are variables in $\bar{x}$.
A variable literal has the form $x.A = y.B$ where $x$ and $y$ are variables in $\bar{x}$ and $A$ and $B$ are attributes in $\Upsilon$.
A constant literal has the form $x.A=c$ where $x$ is a variable in $\bar{x}$, $A$ is an attribute in $\Upsilon$, and $c$ is a constant in $U$.

Graph patterns are evaluated under homomorphism semantics.
A match of the pattern $Q[\bar{x}]$ against the graph $G$ is a mapping $h$ from $Q=(V_Q,E_Q,L_Q,\mu)$ to $G=(V,E,L,F_A)$ such that for each node $u \in V_Q$, $L_Q(u) = L(h(u))$,
for each edge $(u, \iota, u') \in E_Q$ with $\iota \neq \texttt{"\_"}$, $(h(u), \iota, h(u')) \in E$. For each edge $(u, \texttt{"\_"}, u') \in E_Q$ there needs to be some $\iota \in \Gamma$ such that $(h(u), \iota, h(u')) \in E$.

Edge literals in $Y$ enable GARs to be used to extend graphs. The following graph expresses that a person with the name Alice is married to a person with the name Bob, and that Alice lives in a city with the name Bielefeld.
\begin{align*}
\Gamma ~=~& \lbrace person,~city,~marriedTo,~livesIn \rbrace\\
\Upsilon ~=~& \lbrace name \rbrace\\
U ~=~& \lbrace \texttt{"Alice"},~\texttt{"Bob"},~\texttt{"Bielefeld"} \rbrace\\
G ~=~& (V,~E,~L,~F_A)\\
V ~=~& \lbrace A,~B,~C \rbrace\\
E ~=~& \lbrace 
(A,~marriedTo,~B),~
(A,~livesIn,~C)
\rbrace\\
L ~=~& \lbrace 
(A,~person),~
(B,~person),~
(C,~city)
\rbrace\\
F_A(A) ~=~& \lbrace (name,~\texttt{"Alice"}) \rbrace\\
F_A(B) ~=~& \lbrace (name,~\texttt{"Bob"}) \rbrace\\
F_A(C) ~=~& \lbrace (name,~\texttt{"Bielefeld"}) \rbrace
\end{align*}

The following GAR expresses that if we have a person $x_1$ that is married to a person $x_2$ and if the person $x_1$ lives in the city $x_3$, then it is the case that $x_2$ lives in the city $x_3$.
\begin{align*}
Q[\bar{x}] ~=~& (V_Q,~E_Q,~L_Q,~\mu)\\
\bar{x} ~=~& (x_1,~x_2,~x_3)\\
V_Q ~=~& \lbrace A,~B,~C \rbrace\\
\mu ~=~& \lbrace (x_1,~A),~(x_2,~B),~(x_3,~C) \rbrace\\
E_Q ~=~& \lbrace (x_1,~marriedTo,~x_2),~(x_1,~livesIn,~x_3) \rbrace\\
L_Q ~=~& \lbrace 
(x_1,~person),~
(x_2,~person),~
(x_3,~city)
\rbrace\\
X ~=~& \varnothing\\
Y ~=~& \lbrace (x_2,~livesIn,~x_3) \rbrace
\end{align*}

Thus, we can extend the graph so that it also expresses that Bob lives in Bielefeld.

Although the GAR concept is an extension of concepts developed for identifying inconsistencies in data, GARs can be used to extend a graph. One of the contributions of this paper is the concept of graph pattern-based association rules, where these rules can be applied to extend a graph. Whereas our GAR concept is defined for property graphs, the GPAR concept is defined for directed labeled multigraphs. Whereas the pattern in a GAR is evaluated under homomorphism semantics, the patterns in a GPAR are evaluated under no-repeated-anything semantics, which allows patterns to pay more attention to a graph's topology. The graph pattern in a GAR cannot contain variables in predicate position, whereas the graph patterns in GPARs can, which not only allows to write more general patterns (e.g., allow some edge to have any label), but also to express constraints such as that two edges can have any label as long as both labels are identical or different, or to make use of information about the relation that is represented within the graph, such as that a predicate is declared as being symmetric, as shown in Figure \ref{fig:Ex2}. Beyond the structural differences, this paper considers both the single graph setting, as well as the graph transactional setting and introduces and analyzes metrics to score a pattern and to score the strength of correlation between two graph patterns. To the best of our knowledge, no equivalent metrics have been introduced for the GAR concept. However, for the concept underlying to GAR, the metrics support and confidence were defined in \cite{fan2015association}, but it is unclear whether they are applicable to GARs, too.

\subsection{Graph Pattern Association Rules}
Wang et al. introduced the Graph Pattern Association Rules concept \cite{wang2018mining,WANG2020112897}, which is different from the concept with the identical name by Fan et al. \cite{fan2015association}, for social network analysis, where these rules can be applied to recommend friends, but they are not limited to social networks.  They define a kind of property graph as a directed node-labeled graph by a tuple $(V,E,L)$ where $V$ is a set of node identifiers, $E \subseteq V \times V$ is a set of directed edges, and $L$ is a function that maps each node $v \in V$ to a tuple $L(v)$ of the form $(A_1 = a_1, \ldots, A_n=a_n)$ where $A_1, \ldots, A_n$ are distinct attribute names and $a_1, \ldots, a_n$ are attribute values. A graph pattern $p$ is defined as a tuple $(V_p, E_p,f)$ where $V_p$ is a set of node identifiers, $E_p \subseteq V_p \times V_p$ is a set of edges, and $f$ is a function that maps each node $v \in V_p$ to a tuple $f(v)$ of the form $(A_1 = a_1, \ldots, A_m=a_m)$ where $A_1, \ldots, A_m$ are distinct attribute names and $a_1, \ldots, a_m$ are attribute values.
A graph pattern $(V_p,E_p,f)$ matches a graph $(V,E,L)$ if an injection $h$ from $V_p$ to $V$ exists such that for each $(v,v') \in h$, for each statement $A_i=a_i \in L(v)$ it is the case that $A_i=a_i \in L(v')$ and that for each $(v,v') \in E_p$ it is the case that $(h(v),h(v')) \in E$. Thus, the graph pattern is evaluated under no-repeated-node semantics and we can refer to $h$ as the match. However, since edges are not labeled, no-repeated-node semantics and no-repeated-anything semantics coincide. 
A graph pattern association rule is a statement of the form $Q_l \Rightarrow Q_r$ where $Q_l$ and $Q_r$ are graph patterns. The rule expresses that if $Q_l$ matches a graph $G$, then $Q_r$ also matches $G$. Specifically, if there exists a match $h_l$ from $Q_l$ to $G$, then there exists a match $h_r$ from $Q_r$ to $G$ such that for each $u \in V_l \cap V_r$ it is the case that $h_l(v) = h_r(v)$ -- thus, via shared variables correspondences between the matches can be expressed.

The support of a graph pattern is defined
with the \textit{minimum image based support} metric \cite{bringmann2008frequent}.
Let the set of all matches of the graph pattern $Q=(V_Q,E_Q,f)$ on the graph $G=(V,E,L)$ be denoted by $\Omega_{Q,G}$. Let $img(u)$ denote the set of nodes in $V$ to which the node $u \in V_Q$ can be mapped to via some match $h \in \Omega_{Q,G}$, i.e., $img(u) = \lbrace u' \in V \mid \exists h \in \Omega_{Q,G} : h(u) = u' \rbrace$. Support of a graph pattern $Q$ given a graph $G$ is then defined as the minimum cardinality of these sets for nodes in $V_Q$, i.e., $supp(Q,G) = \min_{u \in V_Q}{|img(u)|}$. 
Given a GPAR $Q_l \Rightarrow Q_r$, with $Q_l = (V_l, E_l,f_l)$ and $Q_r = (V_r, E_r,f_r)$, the graph pattern $Q_R$ is constructed as $(V_l, E_l \cup E_r,f_l)$. Confidence of the rule $R = Q_l \Rightarrow Q_r$ is then defined as $supp(Q_R,G) / supp(Q_l,G)$.

Different to our approach is that another type of graph is considered, where nodes have attribute-value pairs, multiple nodes can have the identical set of attribute-value pairs, edges have no labels, and there can be at most one edge between two nodes. Comparable to our approach is the evaluation semantics used, because \textit{nrn} and \textit{nra} semantics here coincide  due to edges being unlabeled. 
A further similarity is that antecedent and consequent are graph patterns (without restricting these patterns, for example, to be path-shaped) and correspondences between matches can be expressed via shared variables.
There is a significant difference in how support is measured and, besides being based on a significantly different support metric, confidence does not measure the strength of correspondence between matches based on shared variables. Finally, the metrics are developed for the single graph setting, and not for the graph transactional setting and they are not defined using probabilities.

\subsection{Path Association Rules}
Sasaki and Karras \cite{sasaki2024mining} introduce the concept of \textit{path association rule} and adapt the metrics support, confidence, and lift. We slightly adapt their graph definition for clarity. Given a set $\mathcal{L}$ of labels and a set $\mathcal{A}$ of attributes, the concept is defined for graphs of the form $(V,E,src,tgt,\lambda_V,\lambda_E)$ where $V$ is a set of node identifiers, $E$ is a set of edge identifiers, $src$ is a function that maps an edge identifier to the identifier of the source node of the edge, $tgt$ is a function that maps an edge identifier to the identifier of the target node of the edge, $\lambda_V$ is a function that maps each node identifier $v$ to a set $A(v) \subseteq \mathcal{A}$ of attributes, and $\lambda_E$ is a function that maps each edge identifier $e$ to a single label $\lambda_E(e) \in \mathcal{L}$.

A path is a sequence of the form $(v_0, e_0, v_1, \ldots, e_{n-1},v_n)$ where $v_i \in V$ are node identifiers and $e_i \in E$ are edge identifiers. Two types of path patterns are introduced: simple path pattern and reachability path pattern -- we here focus on simple path patterns and call them path patterns. A path pattern is a sequence of the form $(A_0, l_0, A_1, \ldots,l_{n-1}, A_n)$ where $A_i \subseteq \mathcal{A}$ are sets of attributes and $l_i \in \mathcal{L}$ are edge labels.

A path pattern $(A_0, l_0, A_1, \ldots,l_{n-1}, A_n)$ matches a path $(v_0, e_0, v_1, \ldots, e_{n-1},v_n)$ if $A_i \subseteq A(v_i)$ and $l_i = l(e_i)$ for all $i$.
Given a graph $G$, a path $p$, and a node identifier $v$, $\mathcal{V}(p)$ is the set of node identifiers for which a path exist where the node is the source node of the path and the path is matched by the path pattern.

A path association rule is an expression of the form $p_x \Rightarrow p_y$ where $p_x$ and $p_y$ are path patterns.
Absolute support of a rule $p_x \Rightarrow p_y$ is defined as $|\mathcal{V}(p_x) \cap \mathcal{V}(p_y)|$, relative support is defined as $|\mathcal{V}(p_x) \cap \mathcal{V}(p_y)|/|V|$, confidence is defined as $|\mathcal{V}(p_x) \cap \mathcal{V}(p_y)|/|\mathcal{V}(p_x)|$, and lift is defined as
$(|\mathcal{V}(p_x) \cap \mathcal{V}(p_y)|\cdot |V|) / (|\mathcal{V}(p_x)| \cdot |\mathcal{V}(p_y)|)$.
In spite of the difference between their graph formalism and our DLM formalism, the situation is comparable to having a graph pattern-based association rule $(p_1,p_2,V_1,V_2)$ where $p_1$ and $p_2$ are line-shaped graph patterns where no variables occur in predicate position and $V_1=V_2=(v_1)$ contains only one variable $v_1$ that is the source of both patterns and the patterns are evaluated under homomorphism semantics.

\subsubsection{Multi-Relation Association Rules (MRAR)}

Ramezani, Saraee, and Nematbakhsh \cite{ramezani2014mrar,ramezani2013finding,ramezani2020mining} in 2014 introduced the concept of Multi-Relation Association Rules (MRAR). Rules are mined from a single directed labeled multigraph, which is an RDF graph or a graph that is created from a relational database.

The following example rule has the meaning \textit{"those who live in a place which is near to a city with humid climate type and are also younger than 20, they have a good health condition:"}
\begin{verbatim}
LiveIn(NearTo(ClimateType(Humid))), AgeLessThan(20) -> HealthCondition(Good)
\end{verbatim}

The antecedent is a set of chains of relations (here: \texttt{LiveIn(NearTo(ClimateType(Humid)))} and \texttt{AgeLessThan(20)} are chains of relations) and the consequent is a single chain of relations (here the chain is \texttt{HealthCondition(Good)}). Each chain ends with an entity (here: \texttt{Humid}, \texttt{20}, and \texttt{Good}) which the authors call endpoint entity.
We can represent such a chain of relations as a graph pattern and can represent the rule above in the form of a graph pattern-based association rule of the form $(p_1,p_2,V_1,V_2)$ as follows:
\begin{align*}
p_1 ~=~& \lbrace (v_j, LiveIn, v_1), (v_1, NearTo, v_2), (v_2, ClimateType, Humid),(v_j,AgeLessThan,20) \rbrace\\
p_2 ~=~& \lbrace (v_j, HealthCondition, Good) \rbrace\\
V_1 ~=~& (v_j)\\
V_2 ~=~& (v_j)
\end{align*}

Note the usage of the special variable $v_j$ ($j$ to indicate that $v_j$ is the joining variable) that occurs at the source of each chain of relations.  
The authors define the metrics relative support of a graph pattern (which they refer to as L-Large ItemChain) and confidence for rules. We can formalize these metrics in our notation as follows, where $\mathcal{T}_g$ denotes the set of nodes that occur in the graph $g$:
support is the number of entities that are connected to the endpoints via the chains of relations, divided by the number of entities in the graph.
\begin{align*}
support_g(p_1,p_2,(v_j)) ~=~& \dfrac{|\lbrace \, t \in \mathcal{T}_g \mid \exists \mu_1 \in \Omega^{hom}_{p_1,g} : (v_j,t) \in \mu_1 \,\rbrace|}{|\mathcal{T}_g|}\\
confidence_g(p_1,p_2,(v_j)) ~=~& \dfrac{|\lbrace \, t \in \mathcal{T}_g \mid \exists \mu_1 \in \Omega^{hom}_{p_1,g} : (v_j,t) \in \mu_1 \land \exists \mu_2 \in \Omega^{hom}_{p_2,g} : (v_j,t) \in \mu_2 \,\rbrace|}{|\lbrace \, t \in \mathcal{T}_g \mid \exists \mu_1 \in \Omega^{hom}_{p_1,g} : (v_j,t) \in \mu_1 \,\rbrace|}
\end{align*}

The main differences or the MRAR formalism and the GPAR formalism introduced in this paper are:
graph patterns in MRAR are much more restricted than in GPAR:
in MRAR, 
variables are not allowed in predicate position,
entities can only occur in object positions in triple patterns,
chains of relations need to end with an endpoint entity and cannot end with a variable. Because in MRAR chains of relations are used where the node variables are implicit, no constraints such as "resides in the father's country of birth" can be expressed. Thus, the graph patterns are star-like trees.
Furthermore,
there can only be one joining variable in MRAR and
MRAR patterns are evaluated under homomorphism semantics. The consequent of a MRAR rule is a single chain of relations, although the MRAR mining approach could most likely be easily extended so that consequents can consist of multiple chains of relations.

Scoring of MRAR rules and GPA rules is different in that MRA rules are scored against a single graph, whereas GPA rules can be scored against a single graph or a bag of graphs. Their definitions of support and confidence would admit a probability space formulation. Thus, metrics such as lift, leverage, and conviction could be defined. %

\subsection{Horn Rules}
According to \cite{lajus2020fast}, a Horn rule $R$ is a statement of the form $ R = B_1 \land \cdots \land B_n \Rightarrow H$ where $B_1 \land \cdots \land B_n$ is a conjunction of body atoms, $H$ is a head atom, and $\Rightarrow$ denotes logical implication. An atom is a statement of the form $r(X,Y)$ where $r$ is a relation and $X$ and $Y$ are either constants or variables.
Typically, in the context of Horn rule mining, further constraints are imposed, i.e., regarding the structure of a rule, namely that that rules need to be closed and connected.
A Horn rule is \textit{closed} if every variable occurs in at least two atoms. Two atoms are connected if they share a variable or a constant. A rule is \textit{connected} if every atom is connected transitively to every other atom of the rule. Further constraints can be imposed for when a rule is instantiated,
where instantiation is a process in which variables in body atoms of a rule are substituted with entities from the graph. Here, the \textit{object identity} \cite{semeraro1994avoiding,meilicke2020reinforced} constraint requires each variable to be replaced with a different variable.

Existing works on Horn rule mining \cite{meilicke2018fine,meilicke2019introduction,meilicke2020reinforced,lajus2020fast,gu2020towards} focus on mining Horn rules, with a subset of the aforementioned constraints, from a single knowledge graph in the form of an RDF graph -- thus, only the single graph setting is considered. 
Although Horn rules can be seen as built from a body graph pattern and a head graph pattern, there are some differences: i) in both patterns variables are not allowed in predicate position, ii) the body graph pattern needs to be connected, and iii) the head graph pattern can only consist of a single triple pattern. Furthermore, although object identity is more constraining than homomorphism semantics, it is still less constraining than \textit{nra} semantics.

Graph pattern-based association rules, as we have defined them, are implication rules. However, GPARs are not Horn rules, due to the aforementioned differences. 
However, every Horn rule can be expressed as a GPAR. Finally, because Horn rules are mined from a single graph, metrics for scoring Horn rules only consider the single graph setting. %

\subsection{Semantic Web Rule Formalisms}
Plenty of rule formalisms have been proposed in the past, for example in the context of Semantic Web Technologies. In the context of this paper it is convenient to represent a graph pattern-based association rule as a tuple of the form $(p_1,p_2,V_1,V_2)$ or of the form $(p_1,p_2)$.

However, 
we can rely on existing standards to describe and evaluate graph patterns for graph pattern-based association rules of the form $(p_1,p_2)$. 
In this section we discuss how graph pattern-based association rules can be expressed, ensuring no-repeated-anything semantics, using the \textit{SPARQL 1.1 Query Language}\footnote{\url{https://www.w3.org/TR/sparql11-query/}} and using the \textit{Semantic Web Rule Language (SWRL)}.\footnote{\url{https://www.w3.org/Submission/SWRL/}}

Building on the Sociology of Education example discussed in Section \ref{sec:graphsintro}, we define a rule that makes it explicit in a graph that the person, for which the pattern $p_2$ (shown in Figure \ref{fig:graphsandpatterns}: a person resides in the father's country of birth which is different from the mother's country of birth) matches, belongs to some class $ClassX$. The rule has the following simplified form:
\begin{align*}
p_1 ~=~& \lbrace 
  (v_1, hF, v_2),
  (v_1, hM, v_3),
  (v_1, cor, v_4),
  (v_2, cob, v_4),
  (v_3, cob, v_5) \rbrace\\
p_2 ~=~& \lbrace (v_1, type, ClassX) \rbrace%
\end{align*}
When we express the rule as a SPARQL CONSTRUCT query we need to transform the labels of the nodes and edges into URIs. Here, for readability, we use compact URIs and, therefore, prepend the example namespace \texttt{ex}. Thus, $cor$ (country of residence) becomes \texttt{ex:cor}.
The rule can then be represented as a SPARQL CONSTRUCT query:
\begin{verbatim}
PREFIX ex: <http://example.org#>
CONSTRUCT {
  ?v1 ex:type ex:ClassX .
} WHERE {
  ?v1 ex:hF ?v2 .
  ?v1 ex:hM ?v3 .
  ?v1 ex:cor ?v4 .
  ?v2 ex:cob ?v4 .
  ?v3 ex:cob ?v5 .
  FILTER (
    ?v1 != ?v2 && ?v1 != ?v3 && ?v1 != ?v4 && ?v1 != ?v5 && ?v2 != ?v3 &&
    ?v2 != ?v4 && ?v2 != ?v5 && ?v3 != ?v4 && ?v3 != ?v5 && ?v4 != ?v5 &&
    ?v1 != ex:hM && ?v1 != ex:hF && ?v1 != ex:cob && ?v1 != ex:cor &&
    ?v2 != ex:hM && ?v2 != ex:hF && ?v2 != ex:cob && ?v2 != ex:cor &&
    ?v3 != ex:hM && ?v3 != ex:hF && ?v3 != ex:cob && ?v3 != ex:cor &&
    ?v4 != ex:hM && ?v4 != ex:hF && ?v4 != ex:cob && ?v4 != ex:cor &&
    ?v5 != ex:hM && ?v5 != ex:hF && ?v5 != ex:cob && ?v5 != ex:cor  
  )
}
\end{verbatim}

The \texttt{FILTER} expression ensure that the graph pattern is evaluated under \textit{nra} semantics. Without the filter, the graph pattern is evaluated under \textit{hom} semantics. 
Furthermore, the rule can be represented as a SWRL rule:
\begin{verbatim}
ex:hF(?v1, ?v2) ^ 
ex:hM(?v1, ?v3) ^ 
ex:cor(?v1, ?v4) ^ 
ex:cob(?v2, ?v4) ^ 
ex:cob(?v3, ?v5) ^ 
swrlb:notEqual(?v1, ?v2) ^ swrlb:notEqual(?v1, ?v3) ^ 
swrlb:notEqual(?v1, ?v4) ^ swrlb:notEqual(?v1, ?v5) ^ 
swrlb:notEqual(?v2, ?v3) ^ swrlb:notEqual(?v2, ?v4) ^ 
swrlb:notEqual(?v2, ?v5) ^ swrlb:notEqual(?v3, ?v4) ^ 
swrlb:notEqual(?v3, ?v5) ^ swrlb:notEqual(?v4, ?v5) ^ 
swrlb:notEqual(?v1, ex:hM)  ^ swrlb:notEqual(?v1, ex:hF)  ^ 
swrlb:notEqual(?v1, ex:cob) ^ swrlb:notEqual(?v1, ex:cor) ^ 
swrlb:notEqual(?v2, ex:hM)  ^ swrlb:notEqual(?v2, ex:hF)  ^ 
swrlb:notEqual(?v2, ex:cob) ^ swrlb:notEqual(?v2, ex:cor) ^ 
swrlb:notEqual(?v3, ex:hM)  ^ swrlb:notEqual(?v3, ex:hF)  ^ 
swrlb:notEqual(?v3, ex:cob) ^ swrlb:notEqual(?v3, ex:cor) ^ 
swrlb:notEqual(?v4, ex:hM)  ^ swrlb:notEqual(?v4, ex:hF)  ^ 
swrlb:notEqual(?v4, ex:cob) ^ swrlb:notEqual(?v4, ex:cor) ^ 
swrlb:notEqual(?v5, ex:hM)  ^ swrlb:notEqual(?v5, ex:hF)  ^ 
swrlb:notEqual(?v5, ex:cob) ^ swrlb:notEqual(?v5, ex:cor)
-> rdf:type(?v1, ex:ClassX)
\end{verbatim}

The \texttt{swrlb:notEqual} statements ensure that the graph pattern is evaluated under \textit{nra} semantics -- it here serves a similar purpose as the filter in the SPARQL CONSTRUCT query.

Note, however, that in a graph pattern-based association rule we allow variables in predicate position. This is allowed in SPARQL queries and in an N3 logic rule,\footnote{\url{https://www.w3.org/TeamSubmission/n3/}} but this is not allowed in SWRL, in Rule Interchange Format,\footnote{\url{https://www.w3.org/TR/rif-core/}} in Datalog \cite{ceri1990logic,abiteboul1995foundations}, and in Horn rules. 
Thus, 
every graph pattern-based association rule can be expressed in SPARQL, but not every rule can be expressed in SWRL, RIF, or Datalog.

\section{Summary and Conclusions}\label{sec:summaryandconclusion}

Sets of directed labeled multigraphs (DLMs) are relevant in many areas such as
social networks, communication networks, citation networks, transportation networks, trade networks, biological interaction networks, molecular structures, and chemical reaction networks and so are approaches that are applicable for generative tasks where a given graph is extended or evaluative tasks where the plausibility of a given graph is evaluated.
These tasks can be realized with rule based approaches which are particularly interesting due to their interpretable nature in comparison to approaches based on artificial neural networks. 
After providing an overview over the concepts of classical, itemset-based association rules and relevant metrics (confidence, lift, leverage, and conviction), directed labeled multigraphs, and graph pattern evaluation, with a particular focus on the no-repeated-anything semantics, we formally introduced our concept of graph pattern-based association rules (GPAR), discussed their interpretation and showed how they can be applied for both generative tasks (such as link prediction and graph extension in general) and evaluative tasks (such as plausibility analysis).

We defined a probability space  and formulated probabilistic versions of confidence, lift, leverage, and conviction where we distinguish between the single graph setting and the graph transactional setting where for the latter we introduced both micro-averaged measures and macro-averaged measures. 
We showed how the problem of evaluating metrics for GPARs can be reframed as evaluation of corresponding metrics for itemset-based association rules. 
We examined these metrics in detail and described how they relate to their classical itemset-based counterparts, identifying the conditions under which key properties are preserved or diverge. 

Our analysis showed that the capabilities of our framework goes beyond a broad range of related formalisms, including graph functional dependencies, graph entity dependencies, relational association rules, graph association rules, multi-relation association rules, path association rules, and Horn rules, but showed that GPARs can be expressed as SPARQL CONSTRUCT queries.

Future work will focus on efficient mining of graph patterns that are frequent under no-repeated-anything semantics and on efficient mining of graph pattern-based association rules.

\section*{Acknowledgments}
This work was supported by the Research Council of Norway, centre of excellence Integreat -- Norwegian Centre for knowledge-driven machine learning, project number 332645.

\bibliographystyle{plain}
\bibliography{main}

\appendix
\section{Proofs}

We revisit the main definitions defined in the four contexts considered that we need for the proofs of the metrics' characteristics. Note that for the context of macro-averaged metrics we simply build on the definitions that were defined for the GPAR single graph context, because when the scores are macro-averaged we are working with a family of probability spaces.

\subsection*{Context ISAR: Itemset-based Association Rules}
Collection of definitions:
\begin{align*}
\allowdisplaybreaks
r ~=~& A \Rightarrow B\\
\Omega ~=~&  T\\
E_A ~=~& \lbrace \, t \in T \mid A \subseteq t \,\rbrace\\
E_B ~=~& \lbrace \, t \in T \mid B \subseteq t \, \rbrace\\
E_A \cap E_B ~=~& \lbrace \, t \in T \mid (A \cup B) \subseteq t \,\rbrace\\
\hat{P}_T(E_A) ~=~& \frac{|E_A|}{|T|}\\
\hat{P}_T(E_B) ~=~& \frac{|E_B|}{|T|}\\
\hat{P}_T(E_A \cap E_B) ~=~& \frac{|E_A \cap E_B|}{|T|}\\
\hat{P}_T(E_B \mid E_A) ~=~& \frac{|E_A \cap E_B|}{|E_A|}
\end{align*}

\subsection*{Context GPAR (single): Graph Pattern-based Association Rules, Single Graph Setting}
Collection of definitions:
\begin{align*}
\allowdisplaybreaks
r ~=~& (p_1,p_2,V_1,V_2)\\
\Omega ~=~& \fulltau\\
E_{g,p_1,V_1;p_2} ~=~& \lbrace \, T \in \fulltau \mid m_g(T,p_1,V_1) \, \rbrace\\
E_{g,p_2,V_2;p_1} ~=~& \lbrace \, T \in \fulltau \mid m_g(T,p_2,V_2) \, \rbrace\\
\hat{P}_g(E_{g,p_1,V_1;p_2}) ~=~& \frac{|E_{g,p_1,V_1;p_2}|}{|\fulltau|}\\
\hat{P}_T(E_{g,p_2,V_2;p_1}) ~=~& \frac{|E_{g,p_2,V_2;p_1}|}{|\fulltau|}\\
\hat{P}_g(E_{g,p_1,V_1;p_2} \cap E_{g,p_2,V_2;p_1}) ~=~& \frac{|E_{g,p_1,V_1;p_2} \cap E_{g,p_2,V_2;p_1}|}{|\fulltau|}\\
\hat{P}_g(E_{g,p_2,V_2;p_1} \mid E_{g,p_1,V_1;p_2}) ~=~& \frac{|E_{g,p_1,V_1;p_2} \cap E_{g,p_2,V_2;p_1}|}{|E_{g,p_1,V_1;p_2}|}
\end{align*}

\subsection*{Context GPAR (micro): Graph Pattern-based Association Rules, Graph Transactional Setting, Micro-averaged}
Collection of definitions:
\begin{align*}
\allowdisplaybreaks
r ~=~& (p_1,p_2,V_1,V_2)\\
\Omega ~=~& \fulltauG\\
E_{G,p_1,V_1;p_2} ~=~& \lbrace \, (i,T) \in \fulltauG \mid m_{g_i}(T,p_1,V_1) \, \rbrace\\
E_{G,p_2,V_2;p_1} ~=~& \lbrace \, (i,T) \in \fulltauG \mid m_{g_i}(T,p_2,V_2) \, \rbrace\\
\hat{P}_G(E_{G,p_1,V_1;p_2}) ~=~& \frac{|E_{G,p_1,V_1;p_2}|}{|\fulltauG|}\\
\hat{P}_G(E_{G,p_2,V_2;p_1}) ~=~& \frac{|E_{G,p_2,V_2;p_1}|}{|\fulltauG|}
\end{align*}

\subsection*{Generalization}
In order to abstract to a general setting for any of the three sample spaces, we make use of the following definitions:
\begin{align*}
\allowdisplaybreaks
E_1 ~=~& \lbrace \, x \in \Omega \mid cond_1(x) \,\rbrace\\
E_2 ~=~& \lbrace \, x \in \Omega \mid cond_2(x) \, \rbrace\\
\hat{P}_{\Omega}(E_1) ~=~& \frac{|E_1|}{|\Omega|}\\
\hat{P}_{\Omega}(E_2) ~=~& \frac{|E_2|}{|\Omega|}
\end{align*}
where $cond_1$ and $cond_2$ are boolean expressions, i.e., 
$A \subseteq t$ and $B \subseteq t$,
$m_g(T,p_1,V_1)$ and $m_g(T,p_2,V_2)$, or
$m_{g_i}(T,p_1,V_1)$ and $m_{g_i}(T,p_2,V_2)$, respectively.
We define the situations of interest as follows:
\begin{align*}
\texttt{IDE} ~\Leftrightarrow~& E_1 = E_2\\
\texttt{DIS} ~\Leftrightarrow~& E_1 \cap E_2 = \varnothing\\
\texttt{IND} ~\Leftrightarrow~& \hat{P}_{\Omega}(E_1 \cap E_2) = \hat{P}_{\Omega}(E_1) \cdot \hat{P}_{\Omega}(E_2)\\
\texttt{POS} ~\Leftrightarrow~& \hat{P}_\Omega(E_1 \cap E_2) > \hat{P}_\Omega(E_1)\cdot \hat{P}_\Omega(E_2)\\
\texttt{NEG} ~\Leftrightarrow~& \hat{P}_\Omega(E_1 \cap E_2) < \hat{P}_\Omega(E_1)\cdot \hat{P}_\Omega(E_2)
\end{align*}

The metrics can now be defined in their generalized form: 
\begin{align*}
\operatorname{confidence}_\Omega(r) ~=~& \hat{P}_\Omega(E_2 \mid E_1) = \frac{\hat{P}_\Omega(E_1 \cap E_2)}{\hat{P}_\Omega(E_1)} = \frac{|E_1 \cap E_2|}{|E_1|}\\
\operatorname{lift}_\Omega(r) ~=~& \frac{\hat{P}_\Omega(E_1 \cap E_2)}{\hat{P}_\Omega(E_2)} = \frac{|E_1 \cap E_2| \cdot |\Omega|}{|E_1| \cdot |E_2|}\\
\operatorname{leverage}_\Omega(r) ~=~& \hat{P}_\Omega(E_1 \cap E_2) - \hat{P}_\Omega(E_1) \cdot \hat{P}_\Omega(E_2)\\
\operatorname{conviction}_\Omega(r) ~=~& \frac{1 - \hat{P}_\Omega(E_2)}{ 1 - \hat{P}_\Omega(E_2 \mid E_1)} = \dfrac{|E_1| \cdot (|\Omega| - |E_2|)}{|\Omega| \cdot \left( |E_1| - |E_1 \cap E_2| \right)}
\end{align*}

We define $\operatorname{conviction}_\Omega(r)$ to be $\infty$ in the case that $\operatorname{confidence}_\Omega(r) = 1$.

\subsection{Proofs about Characteristics of Confidence}

\subsubsection{Contexts: ISAR, GPAR (single), and GPAR (micro)}
\begin{proposition}\label{proposition-confidence-IDE-ISAR-GPARsinglemicro}
Assuming that $\hat{P}_\Omega(E_1) > 0$, then $\texttt{IDE} \Longrightarrow \operatorname{confidence}_\Omega(r) = 1$.
\end{proposition}
\begin{proof}\medskip

\noindent ($\Rightarrow$) Suppose $E_1 = E_2$ holds. Then $\operatorname{confidence}_\Omega(r) = 1$ holds, because

$\operatorname{confidence}_\Omega(r) = \dfrac{|E_1 \cap E_2|}{|E_1|} \overset{\texttt{IDE}}{=} \dfrac{|E_1|}{|E_1|} = 1$.

\noindent ($\Leftarrow$)
Suppose $\operatorname{confidence}_\Omega(r) = 1$ holds. Then, $E_1 = E_2$ does not hold in general, because
$
\dfrac{|E_1 \cap E_2|}{|E_1|} = 1
\Longleftrightarrow~ |E_1 \cap E_2| = |E_1|
\Longleftrightarrow~ E_1 \subseteq E_2 \centernot \Longleftrightarrow E_1 = E_2.
$
\end{proof}
\begin{proposition}\label{proposition-confidence-DIS-ISAR-GPARsinglemicro}
Assuming that $\hat{P}_\Omega(E_1) > 0$, then $\texttt{DIS} \Longleftrightarrow \operatorname{confidence}_\Omega(r) = 0$.
\end{proposition}
\begin{proof}

\noindent ($\Rightarrow$) Suppose $E_1 \cap E_2 = \varnothing$ holds. Then, $\operatorname{confidence}_\Omega(r) = 0$ holds, because

$\operatorname{confidence}_\Omega(r) = \dfrac{|E_1 \cap E_2|}{|E_1|} \overset{\texttt{DIS}}{=} \dfrac{|\varnothing|}{|E_1|} = 0$.

\noindent $(\Leftarrow)$ Suppose $\operatorname{confidence}_\Omega(r) = 0$ holds, then $E_1 \cap E_2 = \varnothing$ holds, because
if $\frac{|E_1 \cap E_2|}{|E_1|} = 0$, then the fraction can only become $0$ if $E_1 \cap E_2 = \varnothing$.
\end{proof}
\begin{proposition}\label{proposition-confidence-IND-ISAR-GPARsinglemicro}
Assuming that $\hat{P}_\Omega(E_1) > 0$, then $\texttt{IND} \Longleftrightarrow \operatorname{confidence}_\Omega(r) = \hat{P}_\Omega(E_2)$.
\end{proposition}
\begin{proof}
Proof via a chain of bidirectional logical equivalences:
\begin{flushleft}
$\begin{aligned}
& \operatorname{confidence}_\Omega(r) = \hat{P}_\Omega(E_2)\\
\Longleftrightarrow~ & \frac{|E_1 \cap E_2|}{|E_1|} = \frac{|E_2|}{|\Omega|}\\
\Longleftrightarrow~ & |E_1 \cap E_2| = \dfrac{|E_1| \cdot |E_2|}{|\Omega|}~~~\text{(divide by } |\Omega| \text{ which is non-zero because } \hat{P}_\Omega(E_1) > 0 \text{)}\\
\Longleftrightarrow~ & \frac{|E_1 \cap E_2|}{|\Omega|} = \frac{|E_1|}{|\Omega|} \cdot \frac{|E_2|}{|\Omega|}\\
\Longleftrightarrow~ & \hat{P}_\Omega(E_1 \cap E_2) = \hat{P}_\Omega(E_1) \cdot \hat{P}_\Omega(E_2)\\ 
\Longleftrightarrow~ & \hat{P}_\Omega(E_1) \text { and } \hat{P}_\Omega(E_2) \text{ are independent.} 
\end{aligned}$
\end{flushleft}
\end{proof}
\begin{proposition}\label{proposition-confidence-POS-ISAR-GPARsinglemicro}
Assuming that $\hat{P}_\Omega(E_1) > 0$, then $\texttt{POS} \Longleftrightarrow \operatorname{confidence}_\Omega(r) > \hat{P}_\Omega(E_2)$.
\end{proposition}
\begin{proof}Omitted -- can be proven in the same way as Proposition \ref{proposition-confidence-IND-ISAR-GPARsinglemicro}. The only thing that needs to be changed is to replace $=$ with $>$.
\end{proof}
\begin{proposition}\label{proposition-confidence-NEG-ISAR-GPARsinglemicro}
Assuming that $\hat{P}_\Omega(E_1) > 0$, then $\texttt{NEG} \Longleftrightarrow \operatorname{confidence}_\Omega(r) < \hat{P}_\Omega(E_2)$.
\end{proposition}
\begin{proof}Omitted -- can be proven in the same way as Proposition \ref{proposition-confidence-IND-ISAR-GPARsinglemicro}. The only thing that needs to be changed is to replace $=$ with $<$.
\end{proof}

\subsubsection{Context: GPAR (macro)}

\[\operatorname{macro-confidence}_G(p_1,p_2,V_1,V_2) = \frac{1}{|G|} \sum_{g\in G} \operatorname{confidence}_g(p_1,p_2,V_1,V_2)\]

\begin{proposition}\label{proposition-confidence-IDE-GPARmacro}
Assuming that $\hat{P}_\Omega(E_{g,p_1,V_1;p_2}) > 0$ for each $g \in G$, then \\$\texttt{IDE} \Longrightarrow \operatorname{macro-confidence}_G(p_1,p_2,V_1,V_2) = 1$.
\end{proposition}
\begin{proof}
$(\Rightarrow$) Suppose that it holds that $E_1$ and $E_2$ are identical for each $g \in G$, then \\$\operatorname{macro-confidence}_G(p_1,p_2,V_1,V_2) = 1$ holds, because\\
$\operatorname{macro-confidence}_G(p_1,p_2,V_1,V_2) = \frac{1}{|G|} \sum_{g\in G} \operatorname{confidence}_g(p_1,p_2,V_1,V_2) \overset{\texttt{IDE}, Prop. \ref{proposition-confidence-IDE-ISAR-GPARsinglemicro}}{=} \frac{1}{|G|} \sum_{g\in G} 1 = 1$.

\noindent ($\Leftarrow$)
Suppose $\operatorname{macro-confidence}_\Omega(p_1,p_2,V_1,V_2) = 1$ holds. Then, $E_1 = E_2$ for each $g \in G$ does not hold in general, because if the sum of $n$ values in $[0,1]$ equals $n$, then each summand is equal to $1$,  which is the case if $E_1 \subseteq E_2$.
\end{proof}

\begin{proposition}\label{proposition-confidence-DIS-GPARmacro}
Assuming that $\hat{P}_\Omega(E_{g,p_1,V_1;p_2}) > 0$ for each $g \in G$, \\then $\texttt{DIS} \Longleftrightarrow \operatorname{macro-confidence}_G(p_1,p_2,V_1,V_2) = 0$.
\end{proposition}
\begin{proof}
($\Rightarrow$) Suppose that it holds that $E_1$ and $E_2$ are disjoint for each $g \in G$, then\\ %
$\operatorname{macro-confidence}_G(p_1,p_2,V_1,V_2) = 0$ holds, because\\
$\operatorname{macro-confidence}_G(p_1,p_2,V_1,V_2) = \frac{1}{|G|} \sum_{g\in G} \operatorname{confidence}_g(p_1,p_2,V_1,V_2) \overset{\texttt{DIS}, Prop. \ref{proposition-confidence-DIS-ISAR-GPARsinglemicro}}{=} \frac{1}{|G|} \sum_{g\in G} 0 = 0$.

\noindent $(\Leftarrow$) Suppose that $\operatorname{macro-confidence}_G(p_1,p_2,V_1,V_2) = 0$ holds, then $E_1$ and $E_2$ are disjoint for each $g \in G$ holds, because if the sum of $n$ values in $[0,1]$ equals $0$, then each summand is equal to $0$, thus, $E_1$ and $E_2$ are disjoint for each $g \in G$.
\end{proof}

\begin{proposition}\label{proposition-confidence-IND-GPARmacro}
Assuming that $\hat{P}_\Omega(E_{g,p_1,V_1;p_2}) > 0$ for each $g \in G$, \\then $\texttt{IND} \Longrightarrow \operatorname{macro-confidence}_G(p_1,p_2,V_1,V_2) = \frac{1}{|G|} \sum_{g \in G} \hat{P}_g(E_{g,p_2,V_2;p_1})$.
\end{proposition}
\begin{proof}

($\Rightarrow$) Suppose that for all $g \in G$ it holds that $\hat{P}_g(E_{g,p_1,V_1;p_2} \cap E_{g,p_2,V_2;p_1}) = \hat{P}_g(E_{g,p_1,V_1;p_2}) \cdot \hat{P}_g(E_{g,p_2,V_2;p_1})$, then $\operatorname{macro-confidence}_G(p_1,p_2,V_1,V_2) = \frac{1}{|G|} \sum_{g \in G} \hat{P}_g(E_{g,p_2,V_2;p_1})$, because \\
$\operatorname{macro-confidence}_G(p_1,p_2,V_1,V_2) = \frac{1}{|G|} \sum_{g\in G} \operatorname{confidence}_g(p_1,p_2,V_1,V_2) \overset{\texttt{IND}, Prop. \ref{proposition-confidence-IND-ISAR-GPARsinglemicro}}{=}\\ \frac{1}{|G|} \sum_{g\in G} \hat{P}_g(E_{g,p_2,V_2;p_1})$.

\noindent ($\Leftarrow$) The reverse is not true, because the same sum can be obtained from different collections of summands.
\end{proof}

\begin{proposition}\label{proposition-confidence-POS-GPARmacro}
Assuming that $\hat{P}_\Omega(E_{g,p_1,V_1;p_2}) > 0$ for each $g \in G$, \\then $\texttt{POS} \Longrightarrow \operatorname{macro-confidence}_G(p_1,p_2,V_1,V_2) > \frac{1}{|G|} \sum_{g \in G} \hat{P}_g(E_{g,p_2,V_2;p_1})$.
\end{proposition}
\begin{proof}Omitted -- can be proven in the same way as Proposition \ref{proposition-confidence-IND-GPARmacro}, using Prop. \ref{proposition-confidence-POS-ISAR-GPARsinglemicro}. The only thing that needs to be changed is to replace $=$ with $>$.
\end{proof}

\begin{proposition}\label{proposition-confidence-NEG-GPARmacro}
Assuming that $\hat{P}_\Omega(E_{g,p_1,V_1;p_2}) > 0$ for each $g \in G$, \\then $\texttt{NEG} \Longrightarrow \operatorname{macro-confidence}_G(p_1,p_2,V_1,V_2) < \frac{1}{|G|} \sum_{g \in G} \hat{P}_g(E_{g,p_2,V_2;p_1})$.
\end{proposition}
\begin{proof}Omitted -- can be proven in the same way as Proposition \ref{proposition-confidence-IND-GPARmacro}, using Prop. \ref{proposition-confidence-NEG-ISAR-GPARsinglemicro}. The only thing that needs to be changed is to replace $=$ with $<$.
\end{proof}

\subsection{Proofs about Characteristics of Lift}

\subsubsection{Contexts: ISAR, GPAR (single), and GPAR (micro)}

\begin{proposition}\label{proposition-lift-IDE-ISAR-GPARsinglemicro}
Assuming that $\hat{P}_\Omega(E_1) > 0$ and $\hat{P}_\Omega(E_2) > 0$, then $\texttt{IDE} \Longrightarrow \operatorname{lift}_\Omega(r) = 1 / \hat{P}_\Omega(E_2)$.
\end{proposition}
\begin{proof}
($\Rightarrow$) Suppose $E_1 = E_2$ holds, then $\operatorname{lift}_\Omega(r) = 1 / \hat{P}_\Omega(E_2)$ holds, because\\ $\operatorname{lift}_\Omega(r) = \dfrac{|E_1 \cap E_2| \cdot |\Omega|}{|E_1| \cdot |E_2|} \overset{\texttt{IDE}}{=} \dfrac{|E_1| \cdot |\Omega|}{|E_1| \cdot |E_2|} = 1/\hat{P}_\Omega(E_2)$.

\noindent ($\Leftarrow$) Suppose $\operatorname{lift}_\Omega(r) = 1 / \hat{P}_\Omega(E_2)$ holds, then $E_1 = E_2$ does not hold in general, because
\begin{flushleft}
$\begin{aligned}
& \operatorname{lift}_\Omega(r) = 1/\hat{P}_\Omega(E_2)\\
\Longleftrightarrow~ & \dfrac{|E_1 \cap E_2| \cdot |\Omega|}{|E_1| \cdot |E_2|} = \frac{|\Omega|}{|E_2|}~~~\text{(divide by } |\Omega| \text{ which is non-zero because } \hat{P}_\Omega(E_1) > 0 \text{)}\\
\Longleftrightarrow~ & \dfrac{|E_1 \cap E_2|}{|E_1| \cdot |E_2|} = \frac{1}{|E_2|}\\
\Longleftrightarrow~ & |E_1 \cap E_2| = \frac{|E_1| \cdot |E_2|}{|E_2|}\\
\Longleftrightarrow~ & |E_1 \cap E_2| = |E_1|\\
\Longleftrightarrow~ & E_1 \subseteq E_2 \centernot \Longleftrightarrow E_1 = E_2.
\end{aligned}$
\end{flushleft}
\end{proof}
\begin{proposition}\label{proposition-lift-DIS-ISAR-GPARsinglemicro}Assuming that $\hat{P}_\Omega(E_1) > 0$ and $\hat{P}_\Omega(E_2) > 0$, then $\texttt{DIS} \Longleftrightarrow \operatorname{lift}_\Omega(r) = 0$.
\end{proposition}
\begin{proof}
($\Rightarrow$) Suppose $E_1 \cap E_2 = \varnothing$ holds, then $\operatorname{lift}_\Omega(r) = 0$ holds, because \\$\operatorname{lift}_\Omega(r) = \dfrac{|E_1 \cap E_2| \cdot |\Omega|}{|E_1| \cdot |E_2|} \overset{\texttt{DIS}}{=} \dfrac{|\varnothing| \cdot |\Omega|}{|E_1| \cdot |E_2|} = 0$.

\noindent ($\Leftarrow$) Suppose $\operatorname{lift}_\Omega(r) = 0$ holds, then $E_1 \cap E_2 = \varnothing$ holds, because 
\begin{flushleft}
$\begin{aligned}
& \operatorname{lift}_\Omega(r) = 0\\
\Longleftrightarrow~ & \dfrac{|E_1 \cap E_2| \cdot |\Omega|}{|E_1| \cdot |E_2|} = 0.
\end{aligned}$
\end{flushleft}
The fraction can only be 0 if $E_1 \cap E_2 = \varnothing$, because $\Omega \neq \varnothing$, because $\hat{P}_\Omega(E_1)> 0$.
\end{proof}

\begin{proposition}\label{proposition-lift-IND-ISAR-GPARsinglemicro}Assuming that $\hat{P}_\Omega(E_1) > 0$ and $\hat{P}_\Omega(E_2) > 0$, then $\texttt{IND} \Longleftrightarrow \operatorname{lift}_\Omega(r) = 1$.
\end{proposition}
\begin{proof}
Proof via a chain of bidirectional logical equivalences:
\begin{flushleft}
$\begin{aligned}
& \operatorname{lift}_\Omega(r) = 1\\
\Longleftrightarrow~ & \dfrac{|E_1 \cap E_2| \cdot |\Omega|}{|E_1| \cdot |E_2|} = 1 ~~~\text{(divide by } |\Omega|^2 \text{ which is non-zero because } \hat{P}_\Omega(E_1) > 0 \text{)}\\
\Longleftrightarrow~ & \dfrac{|E_1 \cap E_2|}{|\Omega|}= \frac{|E_1|}{|\Omega|} \cdot \frac{|E_2|}{|\Omega|}\\
\Longleftrightarrow~ & \hat{P}_\Omega(E_1 \cap E_2) = \hat{P}_\Omega(E_1) \cdot \hat{P}_\Omega(E_2).
\end{aligned}$
\end{flushleft}
\end{proof}

\begin{proposition}\label{proposition-lift-POS-ISAR-GPARsinglemicro}Assuming that $\hat{P}_\Omega(E_1) > 0$ and $\hat{P}_\Omega(E_2) > 0$, then $\texttt{IND} \Longleftrightarrow \operatorname{lift}_\Omega(r) > 1$.
\end{proposition}
\begin{proof}Omitted -- can be proven in the same way as Proposition \ref{proposition-lift-IND-ISAR-GPARsinglemicro}. The only thing that needs to be changed is to replace $=$ with $>$.
\end{proof}

\begin{proposition}\label{proposition-lift-NEG-ISAR-GPARsinglemicro}Assuming that $\hat{P}_\Omega(E_1) > 0$ and $\hat{P}_\Omega(E_2) > 0$, then $\texttt{IND} \Longleftrightarrow \operatorname{lift}_\Omega(r) < 1$.
\end{proposition}
\begin{proof}Omitted -- can be proven in the same way as Proposition \ref{proposition-lift-IND-ISAR-GPARsinglemicro}. The only thing that needs to be changed is to replace $=$ with $<$.
\end{proof}

\subsubsection{Context: GPAR (macro)}

\[\operatorname{macro-lift}_G(p_1,p_2,V_1,V_2) = \frac{1}{|G|} \sum_{g\in G} \operatorname{lift}_g(p_1,p_2,V_1,V_2)\]

\begin{proposition}\label{proposition-lift-IDE-GPARmacro}
Assuming that $\hat{P}_g(E_{g,p_1,V_1;p_2}) > 0$ and $\hat{P}_g(E_{g,p_2,V_2;p_1}) > 0$ for each $g \in G$,
then $\texttt{IDE} \Longrightarrow \operatorname{macro-lift}_G(p_1,p_2,V_1,V_2) = \frac{1}{|G|}
\sum_{g \in G } 1 / \hat{P}_g(E_{g,p_2,V_2;p_1})$.
\end{proposition}
\begin{proof}
($\Rightarrow$) Suppose $E_{g,p_1,V_1;p_2} = E_{g,p_2,V_2;p_1}$ for each $g \in G$  holds, then $\operatorname{macro-lift}_G(p_1,p_2,V_1,V_2)\\ 
= \frac{1}{|G|} \sum_{g \in G } 1 / \hat{P}_g(E_{g,p_2,V_2;p_1})$ holds, because 

$\operatorname{macro-lift}_G(p_1,p_2,V_1,V_2) = \frac{1}{|G|} \sum_{g\in G} \operatorname{lift}_g(p_1,p_2,V_1,V_2) \overset{\texttt{IDE}, Prop. \ref{proposition-lift-IDE-ISAR-GPARsinglemicro}}{=} \frac{1}{|G|} \sum_{g\in G} \hat{P}_g(E_{g,p_2,V_2;p_1})$.

\noindent ($\Leftarrow$) Suppose $\operatorname{macro-lift}_G(p_1,p_2,V_1,V_2) = 
\sum_{g \in G } 1 / \hat{P}_g(E_{g,p_2,V_2;p_1})$ holds, then $E_{g,p_1,V_1;p_2} = \\E_{g,p_2,V_2;p_1}$ does not hold, because the same sum can be obtained from different collections of summands.
\end{proof}

\begin{proposition}\label{proposition-lift-DIS-GPARmacro}Assuming that $\hat{P}_g(E_{g,p_1,V_1;p_2}) > 0$ and $\hat{P}_g(E_{g,p_2,V_2;p_1}) > 0$ for each $g \in G$, then $\texttt{DIS} \Longleftrightarrow \operatorname{macro-lift}_G(p_1,p_2,V_1,V_2) = 0$.
\end{proposition}
\begin{proof}
($\Rightarrow$) Suppose $E_{g,p_1,V_1;p_2} \cap E_{g,p_2,V_2;p_1} = \varnothing$ for each $g \in G$ holds, then \\$\operatorname{macro-lift}_G(p_1,p_2,V_1,V_2) = 0$ holds, because 
$\operatorname{macro-lift}_G(p_1,p_2,V_1,V_2) = \\\frac{1}{|G|} \sum_{g\in G} \operatorname{lift}_g(p_1,p_2,V_1,V_2) \overset{\texttt{DIS}, Prop. \ref{proposition-lift-DIS-ISAR-GPARsinglemicro}}{=} \frac{1}{|G|} \sum_{g\in G} 0 = 0$.

\noindent ($\Leftarrow$) Suppose $\operatorname{macro-lift}_G(p_1,p_2,V_1,V_2) = 0$ holds, then $E_{g,p_1,V_1;p_2} \cap E_{g,p_2,V_2;p_1} = \varnothing$ holds, because if the sum of $n$ values in $[0,1]$ equals $0$, then each summand is equal to $0$, thus, $E_{g,p_1,V_1;p_2}$ and $E_{g,p_2,V_2;p_1}$ are disjoint for each $g \in G$.
\end{proof}

\begin{proposition}\label{proposition-lift-IND-GPARmacro}Assuming that $\hat{P}_g(E_{g,p_1,V_1;p_2}) > 0$ and $\hat{P}_g(E_{g,p_2,V_2;p_1}) > 0$ for each $g \in G$, then $\texttt{IND} \Longrightarrow \operatorname{macro-lift}_G(p_1,p_2,V_1,V_2) = 1$.
\end{proposition}
\begin{proof}
($\Rightarrow$) Suppose $\hat{P}_g(E_{g,p_1,V_1;p_2} \cap E_{g,p_2,V_2;p_1}) = \hat{P}_g(E_{g,p_1,V_1;p_2}) \cdot \hat{P}_g(E_{g,p_2,V_2;p_1})$ for each $g \in G$ holds, then $\operatorname{macro-lift}_G(p_1,p_2,V_1,V_2) = 1$ holds, because\\
$\operatorname{lift}_G(p_1,p_2,V_1,V_2) = \frac{1}{|G|} \sum_{g\in G} \operatorname{lift}_g(p_1,p_2,V_1,V_2) \overset{\texttt{IND}, Prop. \ref{proposition-lift-IND-ISAR-GPARsinglemicro}}{=} \frac{1}{|G|} \sum_{g\in G} 1 = 1$.

\noindent ($\Leftarrow$) Suppose $\operatorname{macro-lift}_G(p_1,p_2,V_1,V_2) = 1$ holds, then $\hat{P}_g(E_{g,p_1,V_1;p_2} \cap E_{g,p_2,V_2;p_1}) = \\\hat{P}_g(E_{g,p_1,V_1;p_2}) \cdot \hat{P}_g(E_{g,p_2,V_2;p_1})$ for each $g \in G$ does not hold, because lift scores are values in the interval $[0,\infty)$ and the same sum can be obtained from different collections of summands.
\end{proof}

\begin{proposition}\label{proposition-lift-POS-GPARmacro}Assuming that $\hat{P}_g(E_{g,p_1,V_1;p_2}) > 0$ and $\hat{P}_g(E_{g,p_2,V_2;p_1}) > 0$ for each $g \in G$, then $\texttt{POS} \Longrightarrow \operatorname{macro-lift}_G(p_1,p_2,V_1,V_2) > 1$.
\end{proposition}
\begin{proof}Omitted -- can be proven in the same way as Proposition \ref{proposition-lift-IND-GPARmacro}, based on  Prop. \ref{proposition-lift-POS-ISAR-GPARsinglemicro}. The only thing that needs to be changed is to replace $=$ with $>$.
\end{proof}

\begin{proposition}\label{proposition-lift-NEG-GPARmacro}Assuming that $\hat{P}_g(E_{g,p_1,V_1;p_2}) > 0$ and $\hat{P}_g(E_{g,p_2,V_2;p_1}) > 0$ for each $g \in G$, then $\texttt{IND} \Longrightarrow \operatorname{macro-lift}_G(p_1,p_2,V_1,V_2) < 1$.
\end{proposition}
\begin{proof}Omitted -- can be proven in the same way as Proposition \ref{proposition-lift-IND-GPARmacro}, based on  Prop. \ref{proposition-lift-NEG-ISAR-GPARsinglemicro}. The only thing that needs to be changed is to replace $=$ with $<$.
\end{proof}

\subsection{Proofs about Characteristics of Leverage}

\subsubsection{Contexts: ISAR, GPAR (single), and GPAR (micro)}
\begin{proposition}\label{proposition-leverage-IDE-ISAR-GPARsinglemicro}$\texttt{IDE} \Longrightarrow \operatorname{leverage}_\Omega(r) = \hat{P}_\Omega(E_1)(1 - \hat{P}_\Omega(E_2))$.

\end{proposition}
\begin{proof}
($\Rightarrow$) Suppose $E_1 = E_2$ holds, then $\operatorname{leverage}_\Omega(r) = \hat{P}_\Omega(E_1)(1 - \hat{P}_\Omega(E_2))$ holds, because
$ \operatorname{leverage}_\Omega(r) = \hat{P}_\Omega(E_1 \cap E_2) - \hat{P}_\Omega(E_1) \cdot \hat{P}_\Omega(E_2) \overset{\texttt{IDE}}{=} \hat{P}_\Omega(E_1) - \hat{P}_\Omega(E_1) \cdot \hat{P}_\Omega(E_2).$

\noindent ($\Leftarrow$) Suppose $\operatorname{leverage}_\Omega(r) = \hat{P}_\Omega(E_1)(1 - \hat{P}_\Omega(E_2))$ holds, then $E_1 = E_2$ does not hold, because 
\begin{flushleft}
$\begin{aligned}
& \operatorname{leverage}_\Omega(r) = \hat{P}_\Omega(E_1)(1 - \hat{P}_\Omega(E_2))\\
\Longleftrightarrow~ & \hat{P}_\Omega(E_1 \cap E_2) - \hat{P}_\Omega(E_1) \cdot \hat{P}_\Omega(E_2) = \hat{P}_\Omega(E_1)(1 - \hat{P}_\Omega(E_2))\\
\Longleftrightarrow~ & \hat{P}_\Omega(E_1 \cap E_2) - \hat{P}_\Omega(E_1) \cdot \hat{P}_\Omega(E_2) = \hat{P}_\Omega(E_1) - \hat{P}_\Omega(E_1)\cdot \hat{P}_\Omega(E_2)\\
\Longleftrightarrow~ & \hat{P}_\Omega(E_1 \cap E_2) = \hat{P}_\Omega(E_1)\\
\Longleftrightarrow~ & E_1 \subseteq E_2 \centernot \Longleftrightarrow E_1 = E_2.
\end{aligned}$
\end{flushleft}
\end{proof}

\begin{proposition}\label{proposition-leverage-DIS-ISAR-GPARsinglemicro}$\texttt{DIS} \Longleftrightarrow \operatorname{leverage}_\Omega(r) = - \hat{P}_\Omega(E_1) \cdot \hat{P}_\Omega(E_2)$.
\end{proposition}
\begin{proof}
($\Rightarrow$) Suppose $E_1 \cap E_2 = \varnothing$ holds, then $\operatorname{leverage}_\Omega(r) = - \hat{P}_\Omega(E_1) \cdot \hat{P}_\Omega(E_2)$ holds, because

$ \operatorname{leverage}_\Omega(r) = \hat{P}_\Omega(E_1 \cap E_2) - \hat{P}_\Omega(E_1) \cdot \hat{P}_\Omega(E_2) \overset{\texttt{DIS}}{=} - \hat{P}_\Omega(E_1) \cdot \hat{P}_\Omega(E_2).$

\noindent ($\Leftarrow$) Suppose $\operatorname{leverage}_\Omega(r) = - \hat{P}_\Omega(E_1) \cdot \hat{P}_\Omega(E_2)$ holds, then $E_1 \cap E_2 = \varnothing$ holds, because

\begin{flushleft}
$\begin{aligned}
& \operatorname{leverage}_\Omega(r) = -\hat{P}_\Omega(E_1)\cdot\hat{P}_\Omega(E_1)\\
\Longleftrightarrow~ & \hat{P}_\Omega(E_1 \cap E_2) - \hat{P}_\Omega(E_1) \cdot \hat{P}_\Omega(E_2) = -\hat{P}_\Omega(E_1)\cdot\hat{P}_\Omega(E_1)\\
\Longleftrightarrow~ & \hat{P}_\Omega(E_1 \cap E_2) = 0\\
\Longleftrightarrow~ & E_1 \cap E_2 = \varnothing.
\end{aligned}$
\end{flushleft}
\end{proof}

\begin{proposition}\label{proposition-leverage-IND-ISAR-GPARsinglemicro}$\texttt{IND} \Longleftrightarrow \operatorname{leverage}_\Omega(r) = 0$.
\end{proposition}
\begin{proof}
($\Rightarrow$) Suppose $\hat{P}_\Omega(E_1 \cap E_2) = \hat{P}_\Omega(E_1) \cdot \hat{P}_\Omega(E_2)$ holds, then $\operatorname{leverage}_\Omega(r) = 0$ holds, because\\
$ \operatorname{leverage}_\Omega(r) = \hat{P}_\Omega(E_1 \cap E_2) - \hat{P}_\Omega(E_1) \cdot \hat{P}_\Omega(E_2) \overset{\texttt{IND}}{=} \hat{P}_\Omega(E_1) \cdot \hat{P}_\Omega(E_2) - \hat{P}_\Omega(E_1) \cdot \hat{P}_\Omega(E_2) = 0.$

\noindent ($\Leftarrow$) Suppose $\operatorname{leverage}_\Omega(r) = 0$ holds, then $\hat{P}_\Omega(E_1 \cap E_2) = \hat{P}_\Omega(E_1) \cdot \hat{P}_\Omega(E_2)$ holds, because
\begin{flushleft}
$\begin{aligned}
& \operatorname{leverage}_\Omega(r) = 0\\
\Longleftrightarrow~ & \hat{P}_\Omega(E_1 \cap E_2) - \hat{P}_\Omega(E_1) \cdot \hat{P}_\Omega(E_2) = 0\\
\Longleftrightarrow~ & \hat{P}_\Omega(E_1 \cap E_2) = \hat{P}_\Omega(E_1) \cdot \hat{P}_\Omega(E_2).\\
\end{aligned}$
\end{flushleft}
\end{proof}

\begin{proposition}\label{proposition-leverage-POS-ISAR-GPARsinglemicro}$\texttt{POS} \Longleftrightarrow \operatorname{leverage}_\Omega(r) > 0$.
\end{proposition}
\begin{proof}Omitted -- can be proven in the same way as Proposition \ref{proposition-leverage-IND-ISAR-GPARsinglemicro}. The only thing that needs to be changed is to replace $=$ with $>$.
\end{proof}

\begin{proposition}\label{proposition-leverage-NEG-ISAR-GPARsinglemicro}$\texttt{NEG} \Longleftrightarrow \operatorname{leverage}_\Omega(r) < 0$.
\end{proposition}
\begin{proof}Omitted -- can be proven in the same way as Proposition \ref{proposition-leverage-IND-ISAR-GPARsinglemicro}. The only thing that needs to be changed is to replace $=$ with $<$.
\end{proof}

\subsubsection{Context: GPAR (macro)}
\[\operatorname{macro-leverage}_G(p_1,p_2,V_1,V_2) = \frac{1}{|G|} \sum_{g\in G} \operatorname{leverage}_g(p_1,p_2,V_1,V_2)\]
 
\begin{proposition}\label{proposition-leverage-IDE-GPARmacro}$\texttt{IDE} \Longrightarrow \operatorname{macro-leverage}_G(p_1,p_2,V_1,V_2) = \\
\frac{1}{|G|}\sum_{g \in G } \hat{P}_g(E_{g,p_1,V_1;p_2})(1-\hat{P}_g(E_{g,p_1,V_1;p_2})). $\end{proposition}
\begin{proof}
($\Rightarrow$) Suppose $E_1 = E_2$ for each $g \in G$ holds, then\\$\operatorname{macro-leverage}_G(p_1,p_2,V_1,V_2) = \frac{1}{|G|} \sum_{g \in G } \hat{P}_g(E_{g,p_1,V_1;p_2})(1-\hat{P}_g(E_{g,p_1,V_1;p_2}))$ holds, because 

$\operatorname{macro-leverage}_G(p_1,p_2,V_1,V_2) = \frac{1}{|G|} \sum_{g\in G} \operatorname{leverage}_g(p_1,p_2,V_1,V_2) \overset{\texttt{IDE}, Prop. \ref{proposition-leverage-IDE-ISAR-GPARsinglemicro}}{=}\\ \frac{1}{|G|} \sum_{g\in G} \hat{P}_g(E_{g,p_1,V_1;p_2})(1-\hat{P}_g(E_{g,p_2,V_2;p_1}))$.

\noindent ($\Leftarrow$) Suppose $\operatorname{macro-leverage}_G(r) = 
\frac{1}{|G|}\sum_{g \in G } \hat{P}_g(E_{g,p_1,V_1;p_2})(1-\hat{P}_g(E_{g,p_1,V_1;p_2}))$ holds, then $E_1 = E_2$ for each $g \in G$ does not hold, because the same sum can be obtained from different collections of summands.
\end{proof}

\begin{proposition}\label{proposition-leverage-DIS-GPARmacro}\texttt{DIS} $\Longrightarrow \operatorname{macro-lev.}_G(p_1,p_2,V_1,V_2) {=} -\frac{1}{|G|}\sum_{g \in G} \hat{P}_g(E_{p_1,V_1;p_2}) \cdot \hat{P}_\Omega(E_{p_2,V_2;p_1}).$
\end{proposition}
\begin{proof}
($\Rightarrow$) Suppose $E_1$ and $E_2$ are disjoint for each $g \in G$ holds, then \\$\operatorname{macro-leverage}_G(p_1,p_2,V_1,V_2) = -\frac{1}{|G|}\sum_{g \in G} \hat{P}_g(E_{p_1,V_1;p_2}) \cdot \hat{P}_\Omega(E_{p_2,V_2;p_1})$ holds, because\\
$\operatorname{macro-leverage}_G(p_1,p_2,V_1,V_2) = \\\frac{1}{|G|} \sum_{g\in G} \operatorname{leverage}_g(p_1,p_2,V_1,V_2) \overset{\texttt{DIS}, Prop. \ref{proposition-leverage-DIS-ISAR-GPARsinglemicro}}{=} \frac{1}{|G|} \sum_{g\in G} - \hat{P}_g(E_{g,p_1,V_1;p_2}) \cdot \hat{P}_g(E_{g,p_2,V_2;p_1})$.

\noindent ($\Leftarrow$) Suppose $\operatorname{macro-leverage}_G(p_1,p_2,V_1,V_2) = -\frac{1}{|G|}\sum_{g \in G} \hat{P}_g(E_{g,p_1,V_1;p_2}) \cdot \hat{P}_\Omega(E_{g,p_2,V_2;p_1})$ holds, then $E_1$ and $E_2$ are disjoint for each $g \in G$ does not hold, because the same sum can be obtained from different collections of summands.
\end{proof}

\begin{proposition}\label{proposition-leverage-IND-GPARmacro}\texttt{IND} $\Longrightarrow \operatorname{macro-leverage}_G = 0$.
\end{proposition}
\begin{proof}
($\Rightarrow$) Suppose $\hat{P}_g(E_{g,p_1,V_1;p_2} \cap E_{g,p_2,V_2;p_1}) = \hat{P}_g(E_{g,p_1,V_1;p_2}) \cdot \hat{P}_g(E_{g,p_2,V_2;p_1})$ for each $g \in G$ holds, then $\operatorname{macro-leverage}_G = 0$ holds, because\\
$\operatorname{macro-leverage}_G(p_1,p_2,V_1,V_2) = \frac{1}{|G|} \sum_{g\in G} \operatorname{leverage}_g(p_1,p_2,V_1,V_2) \overset{\texttt{IND}, Prop. \ref{proposition-leverage-IND-ISAR-GPARsinglemicro}}{=} \frac{1}{|G|} \sum_{g\in G} 0 = 0$.

\noindent ($\Leftarrow$) Suppose $\operatorname{macro-leverage}_G(p_1,p_2,V_1,V_2) = 0$ holds, then $\hat{P}_g(E_{g,p_1,V_1;p_2} \cap E_{g,p_2,V_2;p_1}) = \\\hat{P}_g(E_{g,p_1,V_1;p_2}) \cdot \hat{P}_g(E_{g,p_2,V_2;p_1})$ for each $g \in G$ does not hold, because leverage scores are values in the interval $[-\tfrac{1}{4},\tfrac{1}{4})$ and the same sum can be obtained from different collections of summands.
\end{proof}

\begin{proposition}\label{proposition-leverage-POS-GPARmacro}\texttt{POS} $\Longrightarrow \operatorname{macro-leverage}_G > 0$.
\end{proposition}
\begin{proof}
Omitted -- can be proven in the same way as Proposition \ref{proposition-lift-IND-GPARmacro}, based on  Prop. \ref{proposition-leverage-POS-ISAR-GPARsinglemicro}. The only thing that needs to be changed is to replace $=$ with $<$.
\end{proof}

\begin{proposition}\label{proposition-leverage-NEG-GPARmacro}\texttt{NEG} $\Longrightarrow \operatorname{macro-leverage}_G < 0$.
\end{proposition}
\begin{proof}
Omitted -- can be proven in the same way as Proposition \ref{proposition-lift-IND-GPARmacro}, based on  Prop. \ref{proposition-leverage-NEG-ISAR-GPARsinglemicro}. The only thing that needs to be changed is to replace $=$ with $<$.
\end{proof}

\subsection{Proofs about Characteristics of Conviction}
\subsubsection{Contexts: ISAR, GPAR (single), and GPAR (micro)}

\begin{proposition}\label{proposition-conviction-IDE-ISAR-GPARsinglemicro}Assuming that 
$\hat{P}_\Omega(E_1) > 0$ and $\hat{P}_\Omega(E_2) < 1$, then 
$\texttt{IDE} \Longrightarrow \operatorname{conviction}_\Omega(r) = \infty$.
\end{proposition}
\begin{proof}
($\Rightarrow$) Suppose $E_1 = E_2$ holds, then $\operatorname{conviction}_\Omega(r) = \infty$ holds, because if $E_1 = E_2$, then $\operatorname{confidence}_\Omega(r) = 1$, and then, by definition, $\operatorname{conviction}_\Omega(r) = \infty$.

\noindent ($\Leftarrow$) Suppose $\operatorname{conviction}_\Omega(r) = \infty$ holds, then $E_1 = E_2$ does not hold, because $\operatorname{conviction}_\Omega(r)$ can only be $\infty$ if $\operatorname{confidence}_\Omega(r) = 1$, for which it is necessary that $E_1 \subseteq E_2$.
\end{proof}

\begin{proposition}\label{proposition-conviction-DIS-ISAR-GPARsinglemicro}Assuming that $\hat{P}_\Omega(E_1) > 0$ and $\hat{P}_\Omega(E_2) < 1$, then $\texttt{DIS} \Longleftrightarrow \operatorname{conviction}_\Omega(r) = 1-\hat{P}_\Omega(E_2)$.
\end{proposition}
\begin{proof}
($\Rightarrow$) Suppose $E_1 \cap E_2 = \varnothing$ holds, then $\operatorname{conviction}_\Omega(r) = 1-\hat{P}_\Omega(E_2)$ holds, because\\
$\operatorname{conviction}_\Omega(r) = \dfrac{1 - \hat{P}_\Omega(E_2)}{1 - \hat{P}_\Omega(E_2 \mid E_1)} \overset{\texttt{DIS}}{=} \dfrac{1 - \hat{P}_\Omega(E_2)}{1 - 0} = 1 - \hat{P}_\Omega(E_2)$.

\noindent ($\Leftarrow$) Suppose $\operatorname{conviction}_\Omega(r) = \hat{P}_\Omega(E_1)\cdot(1-\hat{P}_\Omega(E_2))$ holds, then $E_1 \cap E_2 = \varnothing$ holds, because 
\begin{flushleft}
$\begin{aligned}
& \operatorname{conviction}_\Omega(r) = 1-\hat{P}_\Omega(E_2)\\
\Longleftrightarrow~ & \dfrac{1 - \hat{P}_\Omega(E_2)}{ 1 - \hat{P}_\Omega(E_2 \mid E_1)} = 1-\hat{P}_\Omega(E_2)~~~\text{(divide by } 1-\hat{P}_\Omega(E_2) \text{ which is non-zero)}\\
\Longleftrightarrow~ & \dfrac{1}{ 1 - \hat{P}_\Omega(E_2 \mid E_1)} = 1\\
\Longleftrightarrow~ & \dfrac{1}{\frac{\hat{P}_\Omega(E_1)}{\hat{P}_\Omega(E_1)} - \frac{\hat{P}_\Omega(E_1 \cap E_2)}{\hat{P}_\Omega(E_1)}} = 1\\
\Longleftrightarrow~ &  \dfrac{\hat{P}_\Omega(E_1)}{\hat{P}_\Omega(E_1) - \hat{P}_\Omega(E_1 \cap E_2)} = 1\\
\Longleftrightarrow~ & \hat{P}_\Omega(E_1) = \hat{P}_\Omega(E_1) - \hat{P}_\Omega(E_1 \cap E_2)\\
\Longleftrightarrow~ & \hat{P}_\Omega(E_1 \cap E_2) = 0 \Longrightarrow E_1 \cap E_2 = \varnothing.
\end{aligned}$
\end{flushleft}
\end{proof}

\begin{proposition}\label{proposition-conviction-IND-ISAR-GPARsinglemicro}Assuming that $\hat{P}_\Omega(E_1) > 0$ and $\hat{P}_\Omega(E_2) < 1$, then \\$\texttt{IND} \Longleftrightarrow \operatorname{conviction}_\Omega(r) = 1$.
\end{proposition}
\begin{proof}
($\Rightarrow$) Suppose $\hat{P}_g(E_{g,p_1,V_1;p_2} \cap E_{g,p_2,V_2;p_1}) = \hat{P}_g(E_{g,p_1,V_1;p_2}) \cdot \hat{P}_g(E_{g,p_2,V_2;p_1})$ holds, then $\operatorname{conviction}_\Omega(r) = 1$ holds, because

$\operatorname{conviction}_\Omega(r) = \dfrac{1 - \hat{P}_\Omega(E_2)}{1 - \hat{P}_\Omega(E_2 \mid E_1)} \overset{\texttt{IND}}{=} \dfrac{1 - \hat{P}_\Omega(E_2)}{1 - \hat{P}_\Omega(E_2)} = 1$.

\noindent ($\Leftarrow$) Suppose $\operatorname{conviction}_\Omega(r) = 1$ holds, then $\hat{P}_g(E_{g,p_1,V_1;p_2} \cap E_{g,p_2,V_2;p_1}) = \hat{P}_g(E_{g,p_1,V_1;p_2}) \cdot \hat{P}_g(E_{g,p_2,V_2;p_1})$ holds, because

\begin{flushleft}
$\begin{aligned}
& \operatorname{conviction}_\Omega(r) = 1\\
\Longleftrightarrow~ & \dfrac{1 - \hat{P}_\Omega(E_2)}{ 1 - \hat{P}_\Omega(E_2 \mid E_1)} = 1\\
\Longleftrightarrow~ & 1 - \hat{P}_\Omega(E_2) = 1 - \hat{P}_\Omega(E_2 \mid E_1)\\
\Longleftrightarrow~ & \hat{P}_\Omega(E_2 \mid E_1) = \hat{P}_\Omega(E_2).
\end{aligned}$
\end{flushleft}

\end{proof}

\begin{proposition}\label{proposition-conviction-POS-ISAR-GPARsinglemicro}Assuming that $\hat{P}_\Omega(E_1) > 0$ and $\hat{P}_\Omega(E_2) < 1$, then \\$\texttt{POS} \Longleftrightarrow \operatorname{conviction}_\Omega(r) > 1$.
\end{proposition}
\begin{proof}
Omitted -- can be proven in the same way as Proposition \ref{proposition-conviction-IND-ISAR-GPARsinglemicro}. The only thing that needs to be changed is to replace $=$ with $<$.
\end{proof}

\begin{proposition}\label{proposition-conviction-NEG-ISAR-GPARsinglemicro}Assuming that $\hat{P}_\Omega(E_1) > 0$ and $\hat{P}_\Omega(E_2) < 1$, then \\$\texttt{NEG} \Longleftrightarrow \operatorname{conviction}_\Omega(r) < 1$.
\end{proposition}
\begin{proof}Omitted -- can be proven in the same way as Proposition \ref{proposition-conviction-IND-ISAR-GPARsinglemicro}. The only thing that needs to be changed is to replace $=$ with $<$.
\end{proof}

\subsubsection{Context: GPAR (macro)}

\[\operatorname{macro-conviction}_G(p_1,p_2,V_1,V_2) = \frac{1}{|G|} \sum_{g\in G} \operatorname{conviction}_g(p_1,p_2,V_1,V_2)\]

\begin{proposition}\label{proposition-conviction-IDE-GPARmacro}
Assuming that $\hat{P}_\Omega(E_1) > 0$ and $\hat{P}_\Omega(E_2) < 1$ for each $g \in G$, then
$\texttt{IDE} \Longrightarrow \operatorname{macro-conviction}_G(p_1,p_2,V_1,V_2) = \infty$.
\end{proposition}
\begin{proof}
($\Rightarrow$) Suppose $\hat{P}_\Omega(E_1) > 0$ and $\hat{P}_\Omega(E_2) < 1$ for each $g \in G$ holds, then \\$\operatorname{macro-conviction}_G(p_1,p_2,V_1,V_2) = \infty$ holds, because\\
$\operatorname{macro-conv.}_G(p_1,p_2,V_1,V_2) = \frac{1}{|G|} \sum_{g\in G} \operatorname{conviction}_g(p_1,p_2,V_1,V_2) \overset{\texttt{IDE}, Prop. \ref{proposition-conviction-IDE-ISAR-GPARsinglemicro}}{=} \frac{1}{|G|} \sum_{g\in G} \infty = \infty$.

\noindent ($\Leftarrow$) Suppose $\operatorname{conviction}_G(p_1,p_2,V_1,V_2) = \infty$ holds, then $\hat{P}_\Omega(E_1) > 0$ and $\hat{P}_\Omega(E_2) < 1$ for each $g \in G$ does not hold, because it is sufficient for a sum to evaluate to $\infty$ if at least one of the summands are $\infty$.
\end{proof}

\begin{proposition}\label{proposition-conviction-DIS-GPARmacro}Assuming that $\hat{P}_\Omega(E_1) > 0$ and $\hat{P}_\Omega(E_2) < 1$ for each $g \in G$, then\\
$\texttt{DIS} \Longrightarrow \operatorname{macro-conviction}_G(p_1,p_2,V_1,V_2) = \frac{1}{|G|}\sum_{g \in G} (1 - \hat{P}_g(E_{g,p_2,V_2;p_1}))$.
\end{proposition}
\begin{proof}
($\Rightarrow$) Suppose $E_{g,p_1,V_1;p_2} \cap E_{g,p_2,V_2;p_1} = \varnothing$ for each $g \in G$ holds, then \\$\operatorname{macro-conv.}_G(p_1,p_2,V_1,V_2) = \frac{1}{|G|}\sum_{g \in G} (1 - \hat{P}_g(E_{g,p_2,V_2;p_1}))$ holds, because\\
$\operatorname{macro-conv.}_G(p_1,p_2,V_1,V_2) = \frac{1}{|G|} \sum_{g\in G} \operatorname{conv.}_g(p_1,p_2,V_1,V_2) \overset{\texttt{DIS}, Prop. \ref{proposition-conviction-DIS-ISAR-GPARsinglemicro}} = \\\frac{1}{|G|} \sum_{g\in G} (1 - \hat{P}_g(E_{g,p_1,V_1;p_2}))$.

\noindent ($\Leftarrow$) Suppose $\operatorname{macro-conviction}_G(p_1,p_2,V_1,V_2) = \frac{1}{|G|}\sum_{g \in G} (1 - \hat{P}_g(E_{g,p_2,V_2;p_1}))$ holds, then \\$E_{g,p_1,V_1;p_2} \cap E_{g,p_2,V_2;p_1} = \varnothing$ for each $g \in G$ does not hold, because the same sum can be obtained from different collections of summands.
\end{proof}

\begin{proposition}\label{proposition-conviction-IND-GPARmacro}Assuming that $\hat{P}_\Omega(E_1) > 0$ and $\hat{P}_\Omega(E_2) < 1$ for each $g \in G$, then
$\texttt{IND} \Longrightarrow \operatorname{macro-conviction}_G(p_1,p_2,V_1,V_2) = 1$.
\end{proposition}
\begin{proof}
($\Rightarrow$) Suppose $\hat{P}_g(E_{g,p_1,V_1;p_2} \cap E_{g,p_2,V_2;p_1}) = \hat{P}_g(E_{g,p_1,V_1;p_2}) \cdot \hat{P}_g(E_{g,p_2,V_2;p_1})$ holds, then $\operatorname{macro-conviction}_G(p_1,p_2,V_1,V_2) = 1$ holds, because

$\operatorname{macro-conv.}_G(p_1,p_2,V_1,V_2) = \frac{1}{|G|} \sum_{g\in G} \operatorname{conv.}_g(p_1,p_2,V_1,V_2) \overset{\texttt{IND}, Prop. \ref{proposition-conviction-IND-ISAR-GPARsinglemicro}}{=} \frac{1}{|G|} \sum_{g\in G} 1 = 1$.

\noindent ($\Leftarrow$) Suppose $\operatorname{macro-conviction}_G(p_1,p_2,V_1,V_2) = 1$ holds, then $\hat{P}_g(E_{g,p_1,V_1;p_2} \cap E_{g,p_2,V_2;p_1}) = \hat{P}_g(E_{g,p_1,V_1;p_2}) \cdot \hat{P}_g(E_{g,p_2,V_2;p_1})$ holds, 
because conviction scores are values in the interval $\mathbb{Q}_{>} \cap \lbrace \infty \rbrace$ and the same sum can be obtained from different collections of summands.
\end{proof}

\begin{proposition}\label{proposition-conviction-POS-GPARmacro}Assuming that $\hat{P}_\Omega(E_1) > 0$ and $\hat{P}_\Omega(E_2) < 1$ for each $g \in G$, then
$\texttt{POS} \Longrightarrow \operatorname{macro-conviction}_G(p_1,p_2,V_1,V_2) > 1$.
\end{proposition}
\begin{proof}Omitted -- can be proven in the same way as Proposition \ref{proposition-conviction-IND-GPARmacro}, based on  Prop. \ref{proposition-conviction-POS-ISAR-GPARsinglemicro}. The only thing that needs to be changed is to replace $=$ with $<$.
\end{proof}

\begin{proposition}\label{proposition-conviction-NEG-GPARmacro}Assuming that $\hat{P}_\Omega(E_1) > 0$ and $\hat{P}_\Omega(E_2) < 1$ for each $g \in G$, then
$\texttt{NEG} \Longrightarrow \operatorname{macro-conviction}_G(p_1,p_2,V_1,V_2) < 1$.
\end{proposition}
\begin{proof}Omitted -- can be proven in the same way as Proposition \ref{proposition-conviction-IND-GPARmacro}, based on  Prop. \ref{proposition-conviction-NEG-ISAR-GPARsinglemicro}. The only thing that needs to be changed is to replace $=$ with $<$.
\end{proof}

\end{document}